\documentclass[journal, twocolumn]{IEEEtran}

\usepackage{amsmath, amssymb, amsthm}
\usepackage{mathtools}
\usepackage{physics}
\usepackage{pgfplots}
\pgfplotsset{compat=1.18}
\usepackage{hyperref}
\usepackage{paralist}

\usepackage{tikz}
\usetikzlibrary{arrows.meta,calc}

\usepackage{etoolbox}
\usepackage{enumitem}

\usepackage{pgf}
\usepackage{subcaption}

\usepackage{cite}

\usepackage{bm}

\newtheorem{theorem}{Theorem}[section]
\newtheorem{lemma}[theorem]{Lemma}
\newtheorem{proposition}[theorem]{Proposition}
\newtheorem{corollary}[theorem]{Corollary}

\theoremstyle{definition}
\newtheorem{definition}{Definition}[section]
\theoremstyle{remark}
\newtheorem*{remark}{Remark}
\theoremstyle{definition}
\newtheorem{example}{Example}[section]

\DeclareMathOperator{\dom}{dom}
\DeclareMathOperator{\interior}{int}
\DeclareMathOperator{\closure}{cl}

\AtEndEnvironment{example}{\hfill$\square$}

\hypersetup{hidelinks}

\begin{document}

\title{Repeated Binary Direct Collinear Impacts Under Incremental Contact Laws With Residual Indentation: A Hybrid Systems Formulation}

\author{Mihails Milehins and Dan B. Marghitu%
\thanks{Both authors are with the Department of Mechanical Engineering, Auburn University, Auburn, AL 36849 (corresponding author: Mihails Milehins; e-mail: mzm0390@auburn.edu).}}

\maketitle

\begin{abstract}
Incremental contact laws specify the normal contact force through a differential relation, possibly carrying an internal state, driven by the indentation and its rate. The force may vanish at a nonzero indentation, through plasticity, an elastic aftereffect, or, under some laws with a viscous component, a load drawing the bodies apart rapidly. Such a residual indentation sits uneasily within rigid body dynamics, which admits no deformation. The tension is tolerable when the indentation is small relative to the bodies, so that it may be carried constitutively rather than geometrically. Even then, the law alone does not determine the interaction: because force and indentation no longer vanish together, conditions for the commencement and termination of contact must be supplied separately, as must the fate of the deformation and internal state at separation, both absent from the equations of motion. This article formulates the repeated direct collinear impact of two convex axisymmetric bodies under continuous, bounded external forces as a hybrid system. The contact interface is modeled as a massless element carrying the contact law and its state, coupled to the bodies through relative velocity and an interaction force dictated by the contact state. Consequently, all switching and resets are confined to the interface, leaving the geometry and equations of motion of the bodies unaltered. The principal analytical properties of the formulation are established and formally verified in Lean 4 with the help of Claude by Anthropic. The framework is demonstrated through numerical simulations.
\end{abstract}


\section{Introduction}\label{sec:introduction}

Models of systems of rigid bodies with contacts are conventionally divided into
nonsmooth dynamics formulations and continuous formulations
(e.g., see Ref. \cite{machado_compliant_2012}). In a nonsmooth dynamics
formulation, the system is posed as a complementarity problem, a differential
variational inequality, or a hybrid system, in a manner such that the
interpenetration of the bodies in contact is precluded
(e.g., see Refs. \cite{panagiotopoulos_inequality_1985, pfeiffer_multibody_2004,
stewart_dynamics_2011, goebel_hybrid_2012, brogliato_nonsmooth_2016,
sanfelice_hybrid_2021} for descriptions of a variety of approaches of this type). In a continuous formulation, the surfaces of the bodies
are endowed with a localized compliance, represented by virtual viscoelastic or
viscoelastoplastic elements, and a limited interpenetration is admitted as a proxy for
the deformation of the contact region \cite{cundall_measurement_1971, cundall_discrete_1979, terzopoulos_elastically_1987, khulief_continuous_1987,
platt_constraint_1988, moore_collision_1988, movahedi-lankarani_canonical_1988}. Neither formulation is closed by kinematics and kinetics alone: both require a constitutive statement about the
interaction, and such statements shall be referred to collectively as contact laws or collision laws (e.g., see Refs. \cite{chatterjee_rigid_1997, chatterjee_two_1998} and [MacSithigh (1995), as cited in \citenum{chatterjee_rigid_1997}]). Nonsmooth formulations
are closed by algebraic contact laws, which relate the states of the system
immediately before and immediately after an impact, the coefficient of
restitution in its kinematic, kinetic, and energetic forms being the canonical
examples (e.g., see Ref. \cite{stronge_impact_2018}).
Continuous formulations are closed by incremental contact laws, which are
dynamic models that govern the evolution of the contact force over the
duration of the contact (e.g., see Ref. \cite{machado_compliant_2012}).

Incremental contact laws are commonly classified by the assumptions made about
the materials in contact: elastic, viscoelastic, elastoplastic, or
viscoelastoplastic (e.g., see Ref. \cite{stronge_impact_2018}). 
These classes correspond, respectively, to the four constitutive theories of
solid mechanics, those of elasticity, viscoelasticity, plasticity, and
viscoplasticity, which can be distinguished by the two criteria of Table
\ref{tab:classification}: rate dependence and equilibrium hysteresis, the latter
being hysteresis that persists in the quasi-static limit (e.g., see Ref.
\cite{haupt_continuum_2000}).
An incremental contact law is not, however, a constitutive relation in this sense:
it relates the contact force to scalar measures of the deformation of the
contact region as a whole, and depends upon the geometry of that region as well
as upon the materials in contact. Nonetheless, the classification carries over
to contact laws by analogy, with the force in place of the stress and these
measures in place of the strain. The viscoelastic and the elastoplastic classes
differ in the nature of the hysteresis. In a viscoelastic law, the dissipation
is rate-dependent: the energy dissipated over a cycle of loading and unloading
depends on the rate at which the cycle is traversed and vanishes in the
quasi-static limit. In an elastoplastic law, the dissipation is
rate-independent: the energy dissipated over a cycle is invariant under a
reparametrization of time, and persists in the quasi-static limit; and it often,
but not necessarily, arises from yielding, part of the deformation then being
permanent. A viscoelastoplastic law admits both mechanisms.

\begin{table}[!b]
\caption{Classification of constitutive relations.}\label{tab:classification}
\centering
\renewcommand{\arraystretch}{1.2}
\begin{tabular}{lcc}
\hline
& \multicolumn{2}{c}{equilibrium hysteresis} \\
\cline{2-3}
& without & with \\
\hline
rate-independent & elastic & elastoplastic \\
rate-dependent & viscoelastic & viscoelastoplastic \\
\hline
\end{tabular}
\end{table}

In the direct collinear impacts considered here, no tangential force arises, and
the contact law reduces to a relation between the normal contact force and the
indentation, a single scalar measure of the deformation of the contact region.
For the purposes of this article, the pertinent question is not which
dissipation mechanism predominates, but whether the indentation vanishes at the
instant at which the contact force does. For an elastic law, the two vanish
together, and so they do for some viscoelastic laws, provided that no applied
load draws the bodies apart; the contact region has then recovered its
undeformed configuration by the time the bodies separate. A law that admits
yielding does not have this property: the unloading branch reaches zero force at
a nonzero indentation, which thereafter persists. An intermediate case arises
for certain types of viscoelastic solids: the bodies separate with a nonzero
indentation that is recovered only afterwards, a phenomenon known as the elastic
aftereffect (e.g., see Refs. \cite{schwager_coefficient_2008,
wang_advanced_2017}). The presence of a residual indentation at the instant of
separation is therefore not, by itself, diagnostic of plasticity (e.g., see
Refs. \cite{butcher_characterizing_2000, wang_advanced_2017,
stronge_impact_2018, zhang_continuous_2022, zhang_continuous_2024,
poursina_new_2025}), nor is it entailed by a rate-independent hysteresis (e.g.,
see Refs. \cite{sadd_contact_1993, bhattacharjee_interplay_2017}). It is the
feature that governs the structure of the model. It is found in elastoplastic,
viscoelastoplastic, and slowly recovering viscoelastic cases alike, and it
arises under some laws with a viscous component when an applied load draws the
bodies apart rapidly enough.

The admission of a residual indentation sits awkwardly within rigid body
dynamics, whose configuration manifold affords no coordinate for a change of shape. 
The tension is resolved by an argument of scale: the contact
patch is small relative to the bodies, and the residual indentation smaller
still. On this understanding, residual indentation associated 
with plastic yielding has been admitted into impact dynamics for a long time. 
The earliest instance is arguably the theory of Andrews
\cite{andrews_theory_1930}, who described the collision of
spheres by an elastoplastic relation between the contact force and
the indentation in 1930. Through the 1940s and 1950s, relations of this kind were
developed further (e.g., see Refs. \cite{tabor_simple_1948,
tabor_hardness_1951, crook_study_1952}) and, between 1950 and 1960, built
into dynamic models of flexible and multibody systems (e.g., see Refs.
\cite{smith_pile-driving_1950, smith_impact_1955, barnhart_transverse_1955,
barnhart_stresses_1957, goldsmith_impact_1960, smith_pile-driving_1960}). In all
of these, a collision was studied in isolation, so the modeling question of how
a residual indentation can be carried through a simulation in which contacts
are formed and broken repeatedly did not present itself; the reloading of a
recovered indentation was examined, but only quasi-statically.

The question became unavoidable in the 1970s with the discrete element method
(see Refs. \cite{cundall_measurement_1971} and \cite{cundall_discrete_1979}), in
which a body may participate in a great many contacts in succession, its
partners differing from one contact to the next, and the applied load acting
over the duration of a contact is whatever the surrounding assembly imposes.
Precedents for carrying a permanent indentation across repeated contacts can be
found in the same decade, but each was tied to a particular application (e.g.,
see Ref. \cite{holloway_effects_1978}); what the method required was an
interaction law of general applicability. In the mid-1980s, the
piecewise-linear hysteretic spring of Walton and Braun (see Refs.
\cite{walton_viscosity_1986} and \cite{walton_stress_1986}) became the archetype
of such laws: the loading and unloading branches are assigned distinct
stiffnesses, the maximum indentation is retained as a history variable, and the
unloading branch reaches zero force at a residual overlap; 
on separation, the interface is restored to its undeformed
state, no record of the contact spot being retained. 
It should be noted that a spring of a similar kind
had been employed earlier in the program of Ref. \cite{cundall_measurement_1971}, 
where, however, its rules are specified in full
only by the listing of the program. From the 1990s, 
the law of Walton and Braun was extended in several directions, 
among them nonlinear branches and adhesion 
(e.g., see Refs. \cite{walton_numerical_1993, ning_elastic-plastic_1993,
ning_elasto-plastic_1995, vu-quoc_elastoplastic_1999, tomas_particle_2000,
tomas_fundamentals_2004, luding_cohesive_2008, morrissey_discrete_2013,
thakur_micromechanical_2014, rathbone_accurate_2015}). 

More recently, formulations have been given in which the onset and the
termination of contact follow from a unilateral constraint, through which a
compliant element is made to meet one of the bodies (e.g., see Refs.
\cite{kikuuwe_incorporating_2007, xiong_differential_2013,
grzesikiewicz_mathematical_2018}), but only for particular arrangements of
rheological elements. In each of them, the deformation of the element survives
a separation, recovering, if at all, under its own constitutive relations, and
is thereby carried in the geometry of the system. In two or three dimensions,
where the bodies need not meet at the same points from one contact to the
next, it would then have to be attached either to the pair of bodies, and so
follow them to wherever they next meet, or to the location of the contact,
which would have to be recorded, together with a rule for a later contact in
its vicinity, a difficulty noted for the discrete element method 
in Ref. \cite{rathbone_accurate_2015}.

The manner in which a law is to be carried through a simulation in which
contacts are formed and broken repeatedly is a matter on which practice varies.
For many laws, the question is left to the reader, the constitutive
relation being furnished without the disposition of the internal state
of the interface once a contact is broken; where such a statement is supplied,
the conventions adopted differ from one another (e.g., see Refs.
\cite{hunt_coefficient_1975, walton_viscosity_1986, walton_stress_1986, tornambe_modelling_1996, indri_impact_1999, krivtsov_dry_1999, tomas_particle_2000, pavlovskaia_modeling_2001, kikuuwe_incorporating_2007, luding_cohesive_2008, ismail_impact_2008, depouhon_drifting_2013, xiong_differential_2013, xiong_contact_2014, rathbone_accurate_2015, wang_advanced_2017,  grzesikiewicz_mathematical_2018, zhang_continuous_2022, zhang_continuous_2024, poursina_new_2025}). 
A consequence of this variety is that the properties one would wish to establish for a system of rigid bodies admitting contacts are settled not once for a class of laws, but afresh for each combination of a force law with a set of conventions. What is wanted is a formalism in which the conditions
governing the onset and the termination of contact, and the treatment of the
internal state across a separation, are components of the model on the same
footing as the force law itself, so that they may be stated once and their
consequences examined.

A description of that kind exists, and the structure it addresses is not
peculiar to contact. Two bodies that are out of contact and the same two bodies in
contact obey different equations, the passage from one to the other occurs at
instants determined by the state itself, and, at those instants, a quantity may be
reset; systems of this character -- evolving continuously except at isolated
instants, at which the state changes abruptly -- are commonplace in mechanical engineering. 
A class of such systems was identified as early as the 1960s
\cite{witsenhausen_class_1966}, and generalized further in the following decades
\cite{tavernini_differential_1987}; the subject then developed rapidly through
the 1990s, at the confluence of control theory and computer science, under the
name of hybrid systems. 
The formulations cited above that rest upon unilateral constraints are related to
this subject: there is a close interplay between hybrid systems and the
formalisms commonly used for mechanical systems with unilateral constraints, and,
depending on the formalisms involved, systems of one kind can be presented as
systems of the other, and conversely (e.g., see Refs.
\cite{van_der_schaft_complementarity_1998, pang_differential_2008,
stewart_dynamics_2011, brogliato_nonsmooth_2016}).
The formalism adopted here is that of Goebel et al.
\cite{goebel_hybrid_2004, goebel_hybrid_2012}, extended to
admit exogenous inputs \cite{cai_results_2005} and outputs
\cite{sanfelice_results_2010, sanfelice_hybrid_2021}. A system is specified
here by five objects: a set of states and inputs from which
continuous evolution may proceed, a map governing that evolution, a set
from which an abrupt transition is admitted, a map that effects it, and a
map assigning an output to a state and an input.

Impact is among the standard illustrations of the hybrid systems framework,
usually in its rigid form, in which the transition is a reversal of velocity
governed by a coefficient of restitution. Compliant and impulsive descriptions
have also been combined, a compliant force acting during flows and a reset upon
the velocity being enabled above a threshold approach speed. The combination
serves as a regularization of the rigid description in Ref.
\cite{sanfelice_nested_2008} and as a model of contact in Ref.
\cite{naldi_passivity-based_2013}, where the force is viscoelastic and the reset
accounts for the plastic deformation that the force law omits.
A bilinear law admitting a permanent penetration, carried from one contact to
the next, has been built into a hybrid model of percussive drilling in Ref.
\cite{depouhon_drifting_2013}. No threshold upon the approach speed is invoked here, no reset acts upon the
velocities of the bodies, and no indentation is carried from one contact to the
next. The repeated direct collinear impact of two rigid bodies,
each subject to a continuous, bounded applied force, is formulated as a hybrid system for a
contact law that admits a residual indentation. The
contact interface is treated as a massless element that carries an internal state
of its own and is coupled to the bodies through the contact force and the
relative velocity. The equations of motion of the bodies are thereby free of
switching and of resets, the geometry of the bodies is left unaltered, and the
onset of contact, the separation, and the treatment of the internal state across
a separation appear as explicit constituents of the objects above. 

The treatment is given for a class of contact laws rather than for a particular
one. The contact force is furnished by an element that carries a
finite-dimensional internal state driven by the relative velocity, the
indentation being among the components of that state, and the
analysis proceeds from structural hypotheses on that element. The class so delimited
is not confined to laws that are rate-independent, nor to any one arrangement of
rheological elements, nor to any one manner of recovery once the bodies have parted:
the deformation that remains at a separation is discarded, and no geometric
account of its recovery is attempted, so that the question of where it is to be
attached does not arise. A variety of analytical properties
are then established for the resulting models. The models are demonstrated through 
numerical simulations. 

The analytical properties are, moreover, formally verified in the proof assistant Lean~4
\cite{de_moura_lean_2015, de_moura_lean_2021}, using its mathematical library Mathlib
\cite{the_mathlib_community_lean_2020}. The formalization was semi-automated: human-guided, but
largely performed by Claude Code, an AI coding agent by Anthropic \cite{anthropic_model_2023,
anthropic_claude_2025}. A proof assistant checks every step of a proof mechanically; what
remains to be trusted is the proof assistant itself and whether the formal statements, including
their definitions, express the intended results \cite{ilin_semi-autonomous_2026}. Writing such
proofs by hand is laborious, and recent progress in autoformalization, the automatic translation
of informal mathematics into formal statements and proofs, made the verification undertaken here
practical. Early approaches relied on restricted or controlled natural languages
\cite{bobrow_natural_1964, verchinine_system_2007, cramer_naproche_2010, lyaletski_evidence_2010}, 
later ones on statistical and neural translation
\cite{kaliszyk_developing_2014, kaliszyk_learning_2015, kaliszyk_automating_2017,
wang_first_2018, wang_exploration_2020}, and, since at least 2022, on Large Language Models
\cite{wu_autoformalization_2022, jiang_draft_2023}. In 2024, reinforcement learning on
autoformalized problems brought formal theorem proving to olympiad level
\cite{hubert_olympiad-level_2026}, and, more recently, agents based on Large Language Models
have produced substantial formalizations, autonomously or under human guidance (e.g., see Refs.
\cite{math_inc_strong_2025, baudart_putnam_2026, xu_agentic_2026, liu_numina-lean-agent_2026,
urban_130k_2026, bryant_munkres_2026, gloeckle_automatic_2026, ilin_semi-autonomous_2026,
ilin_sorries_2026, kappelmann_just_2026, kappelmann_PIDE_2026, soltani_moakhar_beyond_2026,
anthropic_formalizing_2026, openai_ten_2026, openai_finite_2026, openai_openai_2026}).

Hybrid systems have long been studied with proof assistants and theorem provers (e.g., see
Refs. \cite{lamport_hybrid_1993, vitt_assertional_1996, archer_verifying_1997,
henzinger_reachability_1998, manna_deductive_1998, abraham-mumm_verification_2001,
platzer_differential_2008, fulton_keymaera_2015, bohrer_formally_2017, platzer_logical_2018,
foster_hybrid_2021, huerta_y_munive_predicate_2022, huerta_y_munive_isavodes_2024,
doll_foundations_2026, enes_uniform_2026, ladner_coralean_2026, ogmutlu_dlean_2026,
vargas-garcia_shslib_2026}), and some of this work rests upon a formalization of the theory of
ordinary differential equations (e.g., see Refs. \cite{immler_numerical_2012,
immler_verified_2018}); more recently, Large Language Models have been used to automate proofs
about them \cite{kabra_llm-powered_2026}. However, to the best of the authors' knowledge, the
framework of Goebel et al. has not been formalized before, at least not in Lean~4. Therefore, no
existing library on hybrid systems was used; instead, the elements of the framework that this
article requires were formalized in a small library developed for the purpose.

\section{Contributions and Outline}\label{sec:cao}

The contributions of this article are listed below:
\begin{itemize}
\item The repeated direct collinear impact of two rigid bodies, each subject to
a continuous, bounded applied force, is formulated as a hybrid system with
inputs and outputs, for a class of contact laws admitting a residual
indentation. The onset of contact, the separation, and the disposition of the
internal state across a separation are constituents of the specification, and
the equations of motion of the bodies contain neither switching nor resets.
\item The principal analytical properties of the formulation are established:
the hybrid basic conditions; the continuity of the contact force, for a
compressive contact law; passivity, for a passive contact law; the existence of complete solutions
under continuous and bounded applied
forces, for a contact law whose storage function confines its internal state;
and the branching of the solutions where a contact commences or terminates.
\item A numerical methodology in correspondence with the formulation is given, together with simulation results for three contact laws, two of them constructed from the Bouc-Wen model of hysteresis.
\item The definitions of Sections~\ref{sec:hybrid} to~\ref{sec:properties} are
formalized, and the analytical properties verified, in Lean~4 with the help of
Claude Code. The respects in which the formal statements depart
from those of this article are set out in Section~\ref{sec:formalization}.
\end{itemize}

The remainder of the article is organized as follows.
Section~\ref{sec:hybrid} records the elements of the hybrid systems
framework that are used in the sequel. Section~\ref{sec:model} develops the
model: the motion of the colliding bodies, the contact interface and the
structural hypotheses upon the contact law, and the hybrid system obtained by
the interconnection of the two. Section~\ref{sec:properties} establishes the
analytical properties of the model. Section~\ref{sec:simulation} describes the
numerical methodology and reports the simulation results.
Section~\ref{sec:formalization} outlines the formalization in Lean~4 of the
definitions and results of Sections~\ref{sec:hybrid} to~\ref{sec:properties}.
Section~\ref{sec:conclusions} offers conclusions and indicates directions for
further work.

\section{Hybrid Systems}\label{sec:hybrid}

\subsection{Basics}\label{sec:hybrid_basics}

The exposition that follows is a specialization of the framework of Refs.
\cite{goebel_hybrid_2004, goebel_hybrid_2012}, extended to admit inputs and
outputs after Refs. \cite{cai_results_2005, sanfelice_results_2010,
sanfelice_hybrid_2021}, to the case in which the flow map and the jump map are
single-valued. Only the notions that are required in the sequel are recorded. General mathematical notations and conventions that are employed in this article can be found in Appendix \ref{sec:NCF}. Throughout this section, and this section alone, the symbols introduced are
those of an abstract hybrid system and are not to be identified with the
quantities of the same name that are introduced in the subsequent sections for
the description of the applications.

The evolution of a hybrid system is indexed by an ordinary time $t \in
\mathbb{R}_{\geq 0}$, which advances during the intervals of continuous
evolution, and a counter $j \in \mathbb{Z}_{\geq 0}$, which is incremented at each abrupt
transition.

\begin{definition}[Hybrid time domain]\label{def:htd}
A set $E \subseteq \mathbb{R}_{\geq 0} \times \mathbb{Z}_{\geq 0}$ is a
\emph{compact hybrid time domain} if there exist $n \in \mathbb{Z}_{\geq 1}$
and a non-decreasing sequence $\{ t_j \in \mathbb{R}_{\geq 0} \}_{j = 0}^{n}$
with $t_0 = 0$ such that
\begin{equation}\label{eq:htd}
E = \bigcup_{j = 0}^{n - 1} \left( [t_j, t_{j + 1}] \times \{ j \} \right)
\end{equation}
Such a sequence is said to \emph{generate} $E$, and $E$ is said to be
\emph{generated by} it. The set $E$ is a \emph{hybrid time domain} if there exists a sequence
$\{ E_i \}_{i \in \mathbb{Z}_{\geq 0}}$ of compact hybrid time domains such that
$E_i \subseteq E_{i + 1}$ for all $i \in \mathbb{Z}_{\geq 0}$ and
$E = \bigcup_{i \in \mathbb{Z}_{\geq 0}} E_i$.
\end{definition}

For a hybrid time domain $E$ and $j \in \mathbb{Z}_{\geq 0}$, let
\begin{equation}\label{eq:flow_interval}
I_E^j \triangleq \{ t \in \mathbb{R}_{\geq 0} : (t, j) \in E \}
\end{equation}
let
\begin{equation}\label{eq:J_E}
\mathcal{J}_E \triangleq \{ j \in \mathbb{Z}_{\geq 0} :
\interior I_E^j \neq \emptyset \}
\end{equation}
and, for $(T, J) \in E$, let
\begin{equation}\label{eq:truncation}
E_{(T, J)} \triangleq E \cap \left( [0, T] \times \{ 0, 1, \dots, J \} \right)
\end{equation}

Lemma \ref{lem:flow_interval} below is an immediate consequence of the definitions. Lemma \ref{lem:generating_unique} and Lemma \ref{lem:truncation} are direct consequences of Lemma \ref{lem:flow_interval} and the definitions. Lemma \ref{lem:nested} is a direct consequence 
of Lemma \ref{lem:flow_interval}, Lemma \ref{lem:truncation}, and the definitions.
\begin{lemma}[Flow intervals]\label{lem:flow_interval}
Let $E$ be a compact hybrid time domain generated by $\{ t_j \}_{j = 0}^{n}$.
Then,
\begin{equation}\label{eq:flow_interval_id}
I_E^j = [t_j, t_{j + 1}]
\end{equation}
for all $j \in \{ 0, 1, \dots, n - 1 \}$, and
\begin{equation}\label{eq:flow_interval_empty}
I_E^j = \emptyset
\end{equation}
for all $j \in \mathbb{Z}_{\geq 0}$ with $j \geq n$.
\end{lemma}

\begin{lemma}[Uniqueness of the generating sequence]\label{lem:generating_unique}
Let $E$ be a compact hybrid time domain and let $\{ t_j \}_{j = 0}^{n}$ and $\{
t'_j \}_{j = 0}^{n'}$ generate $E$. Then, $n = n'$ and $t_j = t'_j$ for all $j
\in \{ 0, 1, \dots, n \}$.
\end{lemma}

\begin{lemma}[Truncation]\label{lem:truncation}
Let $E$ be a hybrid time domain and let $(T, J) \in E$. Then, $E_{(T, J)}$ is a
compact hybrid time domain, generated by a sequence $\{ t_j \}_{j = 0}^{J + 1}$
with $t_{J + 1} = T$.
\end{lemma}

\begin{lemma}[Nested hybrid time domains]\label{lem:nested}
Let $E$ and $E'$ be hybrid time domains with $E \subseteq E'$, and let $(T, J)
\in E$. Then, $E'_{(T, J)} = E_{(T, J)}$.
\end{lemma}

\begin{lemma}[Union of a chain of hybrid time domains]\label{lem:chain}
Let $\mathcal{E}$ be a nonempty set of hybrid time domains such that $E_1
\subseteq E_2$ or $E_2 \subseteq E_1$ for all $E_1, E_2 \in \mathcal{E}$, and
let $E^* \triangleq \bigcup_{E \in \mathcal{E}} E$. Then, $E^*$ is a hybrid time
domain, and $E^*_{(T, J)} = E_{(T, J)}$ for all $E \in \mathcal{E}$ and $(T, J)
\in E$.
\end{lemma}

\begin{proof}
For $(T, J) \in E^*$, let $E^*_{(T, J)} \triangleq E^* \cap ([0, T] \times \{ 0,
1, \dots, J \})$, as in Eq. \eqref{eq:truncation}. Let $E \in \mathcal{E}$, let
$(T, J) \in E$, and let $(t, j) \in E^*_{(T, J)}$. Then, $(t, j) \in E_1$ for
some $E_1 \in \mathcal{E}$. Let $E_2$ be the larger
of $E$ and $E_1$, which exists by hypothesis. Then, $(t, j) \in (E_2)_{(T, J)}$,
and $(E_2)_{(T, J)} = E_{(T, J)}$ by Lemma~\ref{lem:nested}, $E$ being a subset
of $E_2$. Hence, $E^*_{(T, J)} \subseteq E_{(T, J)}$, and the reverse inclusion
holds since $E \subseteq E^*$.

Let $(t, j), (t', j') \in E^*$ with $j < j'$. Both belong to some element of
$\mathcal{E}$, and, hence, by Definition~\ref{def:htd}, to a compact hybrid time
domain generated by some $\{ s_k \}_{k = 0}^{n}$, whence $t \leq s_{j + 1} \leq
s_{j'} \leq t'$ by Lemma~\ref{lem:flow_interval}. Consequently, if $(t, j), (t',
j') \in E^*$ and $t + j \leq t' + j'$, then $t \leq t'$ and $j \leq j'$. Indeed,
$j' < j$ would give $t' \leq t$, by the foregoing, and, hence, $t' + j' < t +
j$, contrary to the hypothesis; and, if $j \leq j'$, then $t \leq t'$, by the
foregoing if $j < j'$, and since $t + j \leq t' + j'$ if $j = j'$.

Let $\sigma \triangleq \sup \{ t + j : (t, j) \in E^* \}$. For every $(T, J)
\in E^*$, the set $E^*_{(T, J)}$ is a compact hybrid time domain, by the first
part of the proof and Lemma~\ref{lem:truncation}. If $T + J = \sigma$ for some
$(T, J) \in E^*$, then $E^* = E^*_{(T, J)}$ by the second part of the proof, and
$E^*$ is a compact hybrid time domain, and, hence, a hybrid time domain.
Otherwise, $t + j < \sigma$ for all $(t, j) \in E^*$. Let $\{ (T_i, J_i) \}_{i
\in \mathbb{Z}_{\geq 0}}$ be a sequence in $E^*$ such that $\{ T_i + J_i \}_{i
\in \mathbb{Z}_{\geq 0}}$ is strictly increasing and tends to $\sigma$. By the
second part of the proof, $T_i \leq T_{i + 1}$ and $J_i \leq J_{i + 1}$, whence
$E^*_{(T_i, J_i)} \subseteq E^*_{(T_{i + 1}, J_{i + 1})}$, for all $i \in
\mathbb{Z}_{\geq 0}$. Let $(t, j) \in E^*$. Since $t + j < \sigma$, there exists
$i \in \mathbb{Z}_{\geq 0}$ with $t + j < T_i + J_i$, whence $t \leq T_i$ and
$j \leq J_i$ by the second part of the proof, and $(t, j) \in E^*_{(T_i,
J_i)}$. Hence, $E^*$ is the union of the sets $E^*_{(T_i, J_i)}$, and is a
hybrid time domain by Definition~\ref{def:htd}.
\end{proof}

\begin{definition}[Hybrid arc]\label{def:arc}
Let $n \in \mathbb{Z}_{\geq 1}$ and let $E$ be a hybrid time domain. A map $x :
E \longrightarrow \mathbb{R}^n$ is a \emph{hybrid arc} if, for all $j \in
\mathcal{J}_E$, the map $x(\cdot, j)$ is continuously differentiable upon
$I_E^j$, the derivative at a point of $I_E^j$ that is not interior to $I_E^j$
being one-sided. For such a map, $\dot{x}(t, j)$ denotes the derivative of
$x(\cdot, j)$ at $t$, for all $j \in \mathcal{J}_E$ and $t \in I_E^j$.
\end{definition}

\begin{definition}[Hybrid signal]\label{def:signal}
Let $n \in \mathbb{Z}_{\geq 1}$ and let $E$ be a hybrid time domain. A map $u :
E \longrightarrow \mathbb{R}^n$ is a \emph{hybrid signal} if, for all $j \in
\mathcal{J}_E$, the map $u(\cdot, j)$ is continuous upon $I_E^j$.
\end{definition}

Maps upon a common hybrid time domain are composed pointwise. Let $E$ be a
hybrid time domain, let $m, n_1, n_2 \in \mathbb{Z}_{\geq 1}$, let $\phi_1 : E
\longrightarrow \mathbb{R}^{n_1}$ and $\phi_2 : E \longrightarrow
\mathbb{R}^{n_2}$, and let $\psi : \mathbb{R}^{n_1} \times \mathbb{R}^{n_2}
\longrightarrow \mathbb{R}^m$. Then, $\psi(\phi_1, \phi_2) : E \longrightarrow
\mathbb{R}^m$ denotes the map $(t, j) \mapsto \psi(\phi_1(t, j), \phi_2(t, j))$,
which is a hybrid signal if $\psi$ is continuous and each of $\phi_1$ and
$\phi_2$ is a hybrid arc or a hybrid signal. The same applies to a map of a
single argument; in particular, if $u$ is a hybrid signal, then $\abs{u}$
denotes the hybrid signal $(t, j) \mapsto \abs{u(t, j)}$.

\begin{definition}[Hybrid system with inputs and outputs]\label{def:hybrid}
Let $n_s, n_i, n_o \in \mathbb{Z}_{\geq 1}$. A \emph{hybrid system with inputs
and outputs} is a quintuple $\mathcal{H} \triangleq (C, f, D, g, h)$ in which
the \emph{flow set} $C$ and the \emph{jump set} $D$ are subsets of
$\mathbb{R}^{n_s} \times \mathbb{R}^{n_i}$, and the \emph{flow map} $f$, the
\emph{jump map} $g$, and the \emph{output map} $h$ are maps such that
\begin{equation}\label{eq:hybrid_maps}
f, g : \mathbb{R}^{n_s} \times \mathbb{R}^{n_i} \longrightarrow
  \mathbb{R}^{n_s}, \qquad
h : \mathbb{R}^{n_s} \times \mathbb{R}^{n_i} \longrightarrow \mathbb{R}^{n_o}
\end{equation}
\end{definition}

\begin{remark}
A hybrid system with inputs and outputs is written suggestively as
\begin{equation}\label{eq:hybrid_system}
\mathcal{H} :
\begin{cases}
\dot{x} = f(x, u) & (x, u) \in C \\
x^{+} = g(x, u) & (x, u) \in D \\
y = h(x, u) &
\end{cases}
\end{equation}
where $x^{+}$ denotes the value of the state immediately after a transition. The
flow set is the set of pairs of states and inputs from which continuous
evolution may proceed, and the jump set is the set of such pairs from which an
abrupt transition may occur. 
\end{remark}

\begin{definition}[Solution pair]\label{def:solution}
Let $\mathcal{H} = (C, f, D, g, h)$ be a hybrid system with inputs and outputs.
Suppose $x : E \longrightarrow \mathbb{R}^{n_s}$ is a hybrid arc and $u : E \longrightarrow \mathbb{R}^{n_i}$ is a hybrid signal. The pair $(x, u)$ is a \emph{solution pair} to $\mathcal{H}$ if 
\[
(x(0, 0), u(0, 0)) \in
\closure C \cup D
\]
and the two following conditions are satisfied:
\begin{enumerate}[label=(S\arabic*), ref=S\arabic*, leftmargin=*, nosep]
\item\label{s:flow} for all $j \in \mathcal{J}_E$ and for all $t \in \interior
I_E^j$, the pair $(x(t, j), u(t, j))$ belongs to $C$ and
\begin{equation}\label{eq:flow_dynamics}
\dot{x}(t, j) = f(x(t, j), u(t, j))
\end{equation}
\item\label{s:jump} for all $(t, j) \in E$ such that $(t, j + 1) \in E$, the
pair $(x(t, j), u(t, j))$ belongs to $D$ and
\begin{equation}\label{eq:jump_dynamics}
x(t, j + 1) = g(x(t, j), u(t, j))
\end{equation}
\end{enumerate}
The \emph{output} of the solution pair is the map $y : E \longrightarrow
\mathbb{R}^{n_o}$ given by $y(t, j) \triangleq h(x(t, j), u(t, j))$.
\end{definition}

For a hybrid time domain $E$, let
\begin{equation}\label{eq:sup_t}
\sup\nolimits_t E \triangleq \sup \{ t \in \mathbb{R}_{\geq 0} : \exists j \in
\mathbb{Z}_{\geq 0}, (t, j) \in E \}
\end{equation}
\begin{equation}\label{eq:sup_j}
\sup\nolimits_j E \triangleq \sup \{ j \in \mathbb{Z}_{\geq 0} : \exists t \in
\mathbb{R}_{\geq 0}, (t, j) \in E \}
\end{equation}
which are, respectively, the extent of the ordinary time and the number of the
transitions attained upon $E$.

\begin{definition}[Classification of solution pairs]\label{def:classification}
Let $(x, u)$ be a solution pair to $\mathcal{H}$ and let $E \triangleq \dom x$.
The solution pair is
\begin{itemize}[leftmargin=*, nosep]
\item \emph{trivial} if $E$ consists of a single point, and \emph{non-trivial}
otherwise
\item \emph{complete} if $\sup_t E = +\infty$ or $\sup_j E = +\infty$
\item \emph{Zeno} if it is complete and $\sup_t E < +\infty$
\item \emph{eventually continuous} if $\sup_j E < +\infty$ and $\sup_t E =
+\infty$
\item \emph{bounded} if the range of $x$ is a bounded subset of
$\mathbb{R}^{n_s}$
\end{itemize}
\end{definition}

\begin{definition}[Hybrid basic conditions]\label{def:hbc}
A hybrid system with inputs and outputs $\mathcal{H} = (C, f, D, g, h)$ with the
dimensions $n_s$, $n_i$, and $n_o$ satisfies the \emph{hybrid basic conditions}
if
\begin{enumerate}[label=(A\arabic*), ref=A\arabic*, leftmargin=*, nosep]
\item\label{a:closed} $C$ and $D$ are closed subsets of $\mathbb{R}^{n_s}
\times \mathbb{R}^{n_i}$
\item\label{a:flow} $f$ is continuous
\item\label{a:jump} $g$ is continuous
\item\label{a:output} $h$ is continuous
\end{enumerate}
\end{definition}

A hybrid signal need not be determined by the ordinary time alone: it may take
distinct values at $(t, j)$ and at $(t, j')$ with $j \neq j'$. The external
loads considered in the sequel are, however, prescribed as functions of the
ordinary time, and the hybrid signals to which they give rise are of the following form.

\begin{definition}[Induced hybrid signal]\label{def:induced}
Let $n_i \in \mathbb{Z}_{\geq 1}$, let $w : \mathbb{R}_{\geq 0}
\longrightarrow \mathbb{R}^{n_i}$ be continuous, and let $E$ be a hybrid time
domain. The map $u : E \longrightarrow \mathbb{R}^{n_i}$ given by
\begin{equation}\label{eq:induced}
u(t, j) \triangleq w(t)
\end{equation}
for all $(t, j) \in E$ is a hybrid signal in the sense of
Definition~\ref{def:signal}, the map $u(\cdot, j)$ being the restriction of $w$
to $I_E^j$ for all $j \in \mathbb{Z}_{\geq 0}$; it is the hybrid signal
\emph{induced} by $w$ upon $E$.
\end{definition}

\begin{definition}[Solution pair under an input]\label{def:maximal}
Let $\mathcal{H}$ be a hybrid system with inputs and outputs with the
dimensions $n_s$, $n_i$, and $n_o$, and let $w : \mathbb{R}_{\geq 0}
\longrightarrow \mathbb{R}^{n_i}$ be continuous. A solution pair $(x, u)$ to
$\mathcal{H}$ is a \emph{solution pair under $w$} if $u$ is the hybrid signal
induced by $w$ upon $\dom x$. A solution pair $(x, u)$ under $w$ is
\emph{maximal under $w$} if there exists no solution pair $(x', u')$ under $w$
such that $\dom x \subset \dom x'$ and $x' = x$ upon $\dom x$.
\end{definition}

\begin{definition}[Uniqueness]\label{def:uniqueness}
Let $\mathcal{H}$ be a hybrid system with inputs and outputs with the
dimensions $n_s$, $n_i$, and $n_o$, let $x_0 \in \mathbb{R}^{n_s}$, and let $w
: \mathbb{R}_{\geq 0} \longrightarrow \mathbb{R}^{n_i}$ be continuous. The
solutions to $\mathcal{H}$ \emph{from $x_0$ under $w$} are \emph{unique} if,
for any two solution pairs $(x_1, u_1)$ and $(x_2, u_2)$ to $\mathcal{H}$ under
$w$ such that $x_1(0, 0) = x_2(0, 0) = x_0$, one has $\dom x_1 \subseteq \dom
x_2$ or $\dom x_2 \subseteq \dom x_1$, and $x_1 = x_2$ upon $\dom x_1 \cap \dom
x_2$.
\end{definition}

\subsection{Integration}\label{sec:hybrid_integration}

\begin{definition}[Integral upon a hybrid time domain]\label{def:hint}
Let $E$ be a compact hybrid time domain, generated by $\{ t_j \}_{j = 0}^{n}$
with $n \in \mathbb{Z}_{\geq 1}$, and let $\phi : E \longrightarrow \mathbb{R}$
be such that $\phi(\cdot, j)$ is continuous upon $I_E^j$ for all $j \in
\mathcal{J}_E$. The \emph{integral of $\phi$ upon $E$} is
\begin{equation}\label{eq:hint}
\int_E \phi \triangleq \sum_{j = 0}^{n - 1} \int_{t_j}^{t_{j + 1}} \phi(t, j)
  \, \mathrm{d}t
\end{equation}
\end{definition}

If $E$ is a compact hybrid time domain and $\phi$ is a real-valued map whose
domain contains $E$, then $\int_E \phi$ denotes the integral upon $E$ of the
restriction of $\phi$ to $E$, whenever the latter satisfies the hypothesis of
Definition~\ref{def:hint}.

\begin{proposition}[Integral of an induced signal]\label{prop:induced_integral}
Let $n_i \in \mathbb{Z}_{\geq 1}$, let $w : \mathbb{R}_{\geq 0}
\longrightarrow \mathbb{R}^{n_i}$ be continuous, let $E$ be a hybrid time
domain, let $u$ be the hybrid signal induced by $w$ upon $E$, and let $(T, J)
\in E$. Then,
\begin{equation}\label{eq:induced_integral}
\int_{E_{(T, J)}} \abs{u} = \int_0^{T} \abs{w(t)} \, \mathrm{d}t
\end{equation}
\end{proposition}

\begin{proof}
By Lemma~\ref{lem:truncation}, the set $E_{(T, J)}$ is a compact hybrid time
domain, generated by a sequence $\{ t_j \}_{j = 0}^{J + 1}$ with $t_{J + 1} =
T$. By Definition~\ref{def:induced}, $\abs{u}(t, j) = \abs{w(t)}$ for all $(t,
j) \in E$, whence, $w$ and $\abs{\cdot}$ being continuous, the map
$\abs{u}(\cdot, j)$ is continuous upon $I_{E_{(T, J)}}^j$ for all $j \in
\mathcal{J}_{E_{(T, J)}}$, and the left-hand side of Eq.
\eqref{eq:induced_integral} is defined by Definition~\ref{def:hint}. By Eq.
\eqref{eq:hint},
\[
\int_{E_{(T, J)}} \abs{u}
  = \sum_{j = 0}^{J} \int_{t_j}^{t_{j + 1}} \abs{w(t)} \, \mathrm{d}t
  = \int_0^{T} \abs{w(t)} \, \mathrm{d}t
\]
the intervals $[t_j, t_{j + 1}]$ being consecutive and covering $[0, T]$.
\end{proof}

\subsection{Dissipativity}\label{sec:hybrid_passivity}

The notion of dissipativity that is adopted here is similar to that of Ref.
\cite{naldi_passivity-based_2013}, transcribed to the setting of
Section~\ref{sec:hybrid_basics}. Throughout this subsection, $\mathcal{H} = (C, f, D,
g, h)$ is a hybrid system with inputs and outputs in the sense of
Definition~\ref{def:hybrid} that satisfies the hybrid basic conditions in the sense of Definition~\ref{def:hbc}, and $n_s$, $n_i$, and $n_o$ are its dimensions.

\begin{definition}[Dissipativity]\label{def:dissipativity}
A \emph{supply rate} is a continuous map $s : \mathbb{R}^{n_i} \times
\mathbb{R}^{n_o} \longrightarrow \mathbb{R}$. The system $\mathcal{H}$ is
\emph{dissipative with respect to a supply rate $s$} if there exists a
continuously differentiable map $V : \mathbb{R}^{n_s} \longrightarrow
\mathbb{R}_{\geq 0}$ such that
\begin{enumerate}[label=(D\arabic*), ref=D\arabic*, leftmargin=*, nosep]
\item\label{d:flow} for all $(x, u) \in C$,
\begin{equation}\label{eq:diss_flow}
\left\langle \nabla V(x), f(x, u) \right\rangle \leq s\left( u, h(x, u) \right)
\end{equation}
\item\label{d:jump} for all $(x, u) \in D$,
\begin{equation}\label{eq:diss_jump}
V\left( g(x, u) \right) \leq V(x)
\end{equation}
\end{enumerate}
Such a map $V$ is a \emph{storage function} for $\mathcal{H}$ with respect to
$s$.
\end{definition}

\begin{definition}[Passivity]\label{def:passivity}
Suppose $n_i = n_o$. The system $\mathcal{H}$ is \emph{passive} if it is
dissipative with respect to the supply rate $s : \mathbb{R}^{n_i} \times
\mathbb{R}^{n_o} \longrightarrow \mathbb{R}$ given by $s(u, y) \triangleq
\left\langle u, y \right\rangle$.
\end{definition}

Throughout the remainder of this subsection, $\mathcal{H}$ is dissipative with
respect to $s$ with the storage function $V$, and $(x, u)$ is a solution pair
to $\mathcal{H}$, with the output $y$; it is written $E \triangleq \dom x$, and
$(T, J) \in E$.

\begin{definition}[Dissipation rates]\label{def:rates}
The \emph{dissipation rate upon the flow set} and the \emph{dissipation rate
upon the jump set} are the maps $d_C, d_D : \mathbb{R}^{n_s} \times
\mathbb{R}^{n_i} \longrightarrow \mathbb{R}$ given by
\begin{equation}\label{eq:d_C}
d_C(x, u) \triangleq s\left( u, h(x, u) \right) -
  \left\langle \nabla V(x), f(x, u) \right\rangle
\end{equation}
\begin{equation}\label{eq:d_D}
d_D(x, u) \triangleq V(x) - V\left( g(x, u) \right)
\end{equation}
respectively.
\end{definition}

The maps $d_C$ and $d_D$ are continuous, by \ref{a:flow}, \ref{a:jump},
\ref{a:output}, and the continuous differentiability of $V$; the first is
non-negative upon $C$ and the second upon $D$, by Eqs. \eqref{eq:diss_flow} and
\eqref{eq:diss_jump}, respectively.

\begin{definition}[Supply]\label{def:supply}
The \emph{supply} received by $\mathcal{H}$ up to $(T, J)$ is
\begin{equation}\label{eq:supply}
\mathcal{I}_s(x, u, T, J) \triangleq \int_{E_{(T, J)}} s(u, y)
\end{equation}
\end{definition}

\begin{definition}[Dissipated energies]\label{def:dissipated}
Let $\{ t_j \}_{j = 0}^{J + 1}$ be the sequence generating $E_{(T, J)}$. The
energy \emph{dissipated during the continuous evolutions} up to $(T, J)$ is
\begin{equation}\label{eq:Delta_C}
\Delta_C(x, u, T, J) \triangleq \int_{E_{(T, J)}} d_C(x, u)
\end{equation}
and the energy \emph{dissipated at the transitions} up to $(T, J)$ is
\begin{equation}\label{eq:Delta_D}
\Delta_D(x, u, T, J) \triangleq \sum_{j = 0}^{J - 1}
  d_D\left( x(t_{j + 1}, j), u(t_{j + 1}, j) \right)
\end{equation}
\end{definition}

The quantities of Definitions~\ref{def:supply} and \ref{def:dissipated} are
defined and finite. For all $j \in \mathcal{J}_E$, the maps $s(u, y)(\cdot, j)$
and $d_C(x, u)(\cdot, j)$ are continuous upon $I_E^j$, by
Definitions~\ref{def:arc} and \ref{def:signal} and by the continuity of $s$,
$h$, and $d_C$; the integrals of Eqs. \eqref{eq:supply} and \eqref{eq:Delta_C}
are accordingly defined by Definition~\ref{def:hint}. 

\begin{proposition}[Balance of energy]\label{prop:balance}
The quantities of Eqs. \eqref{eq:Delta_C} and \eqref{eq:Delta_D} are
non-negative, and
\begin{equation}\label{eq:balance}
V\left( x(T, J) \right) - V\left( x(0, 0) \right) + \Delta_C + \Delta_D
  = \mathcal{I}_s
\end{equation}
in which the arguments of $\mathcal{I}_s$, $\Delta_C$, and $\Delta_D$ have been
suppressed. In particular,
\begin{equation}\label{eq:dissipation_inequality}
V\left( x(T, J) \right) - V\left( x(0, 0) \right) \leq \mathcal{I}_s
\end{equation}
\end{proposition}

\begin{proof}
Let $\{ t_j \}_{j = 0}^{J + 1}$ generate $E_{(T, J)}$ and let $j \in \{ 0, 1,
\dots, J \}$.

Suppose first that $t_j < t_{j + 1}$, whence $j \in \mathcal{J}_E$. By
Definition~\ref{def:arc} and the continuous differentiability of $V$, the map
$V(x(\cdot, j))$ is continuously differentiable upon $I_E^j$, with
\[
\frac{\mathrm{d}}{\mathrm{d}t} V\left( x(t, j) \right)
  = s(u, y)(t, j) - d_C(x, u)(t, j)
\]
and $d_C(x, u)(t, j) \geq 0$, for all $t \in \interior I_E^j$ by condition
\ref{s:flow} and Eqs. \eqref{eq:diss_flow} and \eqref{eq:d_C}, and, hence, for
all $t \in I_E^j$ by continuity. Integration of the foregoing identity upon
$[t_j, t_{j + 1}]$ yields
\[
\begin{multlined}
V\left( x(t_{j + 1}, j) \right) - V\left( x(t_j, j) \right) = \\
  \int_{t_j}^{t_{j + 1}} s(u, y)(t, j) \, \mathrm{d}t  - \int_{t_j}^{t_{j + 1}} d_C(x, u)(t, j) \, \mathrm{d}t
\end{multlined}
\]
If, instead, $t_j = t_{j + 1}$, then the same identity holds, every term
vanishing.

Suppose now that $j \leq J - 1$. By condition \ref{s:jump}, the pair $(x(t_{j +
1}, j), u(t_{j + 1}, j))$ belongs to $D$, and, by Eqs. \eqref{eq:jump_dynamics}
and \eqref{eq:d_D},
\[
\begin{multlined}
V\left( x(t_{j + 1}, j + 1) \right) - V\left( x(t_{j + 1}, j) \right) = \\
 - d_D\left( x(t_{j + 1}, j), u(t_{j + 1}, j) \right)
\end{multlined}
\]
the right-hand side of which is non-positive by Eq. \eqref{eq:diss_jump}.

The quantity of Eq. \eqref{eq:Delta_C} is therefore non-negative, its integrand
being non-negative upon $I_E^j$ for all $j \in \mathcal{J}_E$, and so is that
of Eq. \eqref{eq:Delta_D}, its summands being non-negative by Eq.
\eqref{eq:diss_jump}.

Adding the left-hand sides of the identities obtained upon the intervals $[t_j,
t_{j + 1}]$, for $j \in \{ 0, 1, \dots, J \}$, and at the transitions, for $j
\in \{ 0, 1, \dots, J - 1 \}$, yields
\[
\begin{split}
& \sum_{j = 0}^{J} \left[ V\left( x(t_{j + 1}, j) \right) -
    V\left( x(t_j, j) \right) \right] \\
& \quad + \sum_{j = 0}^{J - 1} \left[ V\left( x(t_{j + 1}, j + 1) \right) -
    V\left( x(t_{j + 1}, j) \right) \right] \\
& \qquad = V\left( x(t_{J + 1}, J) \right) - V\left( x(t_0, 0) \right) \\
& \qquad \qquad  = V(x(T, J)) - V(x(0, 0))
\end{split}
\]
since $t_0 = 0$ and $t_{J + 1} = T$. The right-hand sides
add to $\mathcal{I}_s - \Delta_C - \Delta_D$, and Eq. \eqref{eq:balance} follows. Eq.
\eqref{eq:dissipation_inequality} is obtained upon discarding $\Delta_C$ and
$\Delta_D$, which are non-negative.
\end{proof}

\begin{proposition}[Growth of the storage]\label{prop:growth}
Let $\mathcal{H}$ be dissipative with respect to a supply rate $s$ with the
storage function $V$, and suppose that there exists $\gamma \in
\mathbb{R}_{>0}$ such that
\begin{equation}\label{eq:supply_bound}
s\left( u, h(x, u) \right) \leq \gamma \abs{u} \sqrt{V(x)}
\end{equation}
for all $(x, u) \in C$. Let $(x, u)$ be a solution pair to $\mathcal{H}$ and
let $(T, J) \in \dom x$. Then,
\begin{equation}\label{eq:growth_bound}
\sqrt{V\left( x(T, J) \right)} \leq \sqrt{V\left( x(0, 0) \right)} +
  \gamma \int_{E_{(T, J)}} \abs{u}
\end{equation}
\end{proposition}

\begin{proof}
Let $\{ t_j \}_{j = 0}^{J + 1}$ generate $E_{(T, J)}$ and let $(\tau, k) \in
E_{(T, J)}$. Applying Eqs. \eqref{eq:dissipation_inequality} and
\eqref{eq:supply} at $(\tau, k)$, and using Eq. \eqref{eq:supply_bound},
\[
V\left( x(\tau, k) \right) \leq V\left( x(0, 0) \right) +
  \gamma \int_{E_{(\tau, k)}} \abs{u} \sqrt{V(x)}
\]
Let
\[
\sigma \triangleq \sup \left\{ \sqrt{V\left( x(\tau, k) \right)} :
  (\tau, k) \in E_{(T, J)} \right\}
\]
and $\Lambda \triangleq \int_{E_{(T, J)}} \abs{u}$. The set $E_{(T, J)}$ is
the union of the finitely many sets $[t_k, t_{k + 1}]
\times \{ k \}$ with $k \in \{ 0, 1, \dots, J \}$, upon each of which the map
$V \circ x$ is continuous by Definitions~\ref{def:arc} and \ref{def:dissipativity};
it is therefore bounded upon $E_{(T, J)}$, and $\sigma$ is finite. Since
$E_{(\tau, k)} \subseteq E_{(T, J)}$ and the integrand of the integral
inequality above is non-negative, that inequality gives $V\left( x(\tau, k)
\right) \leq V\left( x(0, 0) \right) + \gamma \sigma \Lambda$ for all $(\tau,
k) \in E_{(T, J)}$, whence, upon taking the supremum, $\sigma^2 \leq V\left(
x(0, 0) \right) + \gamma \sigma \Lambda$. Were $\sigma > \sqrt{V\left( x(0, 0)
\right)} + \gamma \Lambda$, then $\sigma > \sqrt{V\left( x(0, 0) \right)}$,
whence
\[
\sigma^2 > \sigma \left( \sqrt{V\left( x(0, 0) \right)} + \gamma \Lambda
  \right) \geq V\left( x(0, 0) \right) + \gamma \sigma \Lambda
\]
contradicting the upper bound upon $\sigma^2$ obtained above. Eq.
\eqref{eq:growth_bound}
follows, $(T, J)$ belonging to $E_{(T, J)}$ and $\sqrt{V(x(T, J))}$ being
accordingly bounded above by $\sigma$.
\end{proof}

\subsection{Existence}\label{sec:existence}

\begin{lemma}[Maximal extension]\label{lem:maximal}
Let $\mathcal{H} = (C, f, D, g, h)$ be a hybrid system with inputs and outputs
with the dimensions $n_s$, $n_i$, and $n_o$, let $w : \mathbb{R}_{\geq 0}
\longrightarrow \mathbb{R}^{n_i}$ be continuous, and let $(x, u)$ be a solution
pair to $\mathcal{H}$ under $w$. Then, there exists a solution pair $(\bar{x},
\bar{u})$ to $\mathcal{H}$ that is maximal under $w$ and such that $\dom x
\subseteq \dom \bar{x}$ and $\bar{x} = x$ upon $\dom x$.
\end{lemma}

\begin{proof}
Let $\mathcal{A}$ be the set of the solution pairs $(x', u')$ to $\mathcal{H}$
under $w$ such that $\dom x \subseteq \dom x'$ and $x' = x$ upon $\dom x$, and,
for $(x_1, u_1), (x_2, u_2) \in \mathcal{A}$, let $(x_1, u_1) \preceq (x_2,
u_2)$ if $\dom x_1 \subseteq \dom x_2$ and $x_2 = x_1$ upon $\dom x_1$. The
relation $\preceq$ is a partial order upon $\mathcal{A}$, the signal of each
element of $\mathcal{A}$ being determined by $w$ and by the domain of its arc,
and $(x, u) \in \mathcal{A}$. By Zorn's lemma, $\mathcal{A}$ has an element
that is maximal with respect to $\preceq$ if every nonempty subset of
$\mathcal{A}$ that is totally ordered by $\preceq$ has an upper bound in
$\mathcal{A}$.

Let $\mathcal{Q}$ be a nonempty subset of $\mathcal{A}$ that is totally ordered
by $\preceq$, let $\mathcal{E}$ be the set of the domains of the arcs of its
elements, which is totally ordered by inclusion, and let $E^* \triangleq
\bigcup_{E \in \mathcal{E}} E$, which is a hybrid time domain by
Lemma~\ref{lem:chain}. Let $x^* : E^* \longrightarrow \mathbb{R}^{n_s}$ be given
by $x^*(t, j) \triangleq x'(t, j)$ for all $(x', u') \in \mathcal{Q}$ and $(t,
j) \in \dom x'$, which is well defined, any two elements of $\mathcal{Q}$ being
comparable under $\preceq$, and let $u^*$ be the hybrid signal induced by $w$
upon $E^*$. For every $(T, J) \in E^*$, there exists $(x', u') \in \mathcal{Q}$
with $(T, J) \in \dom x'$, and Lemma~\ref{lem:chain} gives $E^*_{(T, J)} = (\dom
x')_{(T, J)}$; moreover, $(x^*, u^*)$ coincides with $(x', u')$ upon $\dom x'$,
$u'$ being induced by $w$.

Let $j \in \mathcal{J}_{E^*}$. For every $s \in I_{E^*}^j$,
Lemmas~\ref{lem:truncation} and \ref{lem:flow_interval} give $I_{E^*}^j \cap
[0, s] = [t_j, s]$, where $\{ t_k \}_{k = 0}^{j + 1}$ generates $E^*_{(s, j)}$;
hence, $a \triangleq \min I_{E^*}^j$ exists, and $I_{E^*}^j \cap [0, s] = [a,
s]$ for all $s \in I_{E^*}^j$. Let $s \in I_{E^*}^j$ with $s > a$. By the
foregoing, with $(T, J) \triangleq (s, j)$, there exists $(x', u') \in
\mathcal{Q}$ such that $[a, s] \subseteq I_{\dom x'}^j$, whence $j \in
\mathcal{J}_{\dom x'}$, and $x^*(\cdot, j)$ coincides upon
$[a, s]$ with $x'(\cdot, j)$, which is continuously differentiable there by
Definition~\ref{def:arc}. Every point of $I_{E^*}^j$ has a neighborhood relative
to $I_{E^*}^j$ of the form $[a, s]$, and $x^*(\cdot, j)$ is accordingly
continuously differentiable upon $I_{E^*}^j$, the derivative at a point of
$I_{E^*}^j$ that is not interior to $I_{E^*}^j$ being one-sided. Hence, $x^*$ is
a hybrid arc. If, moreover, $t \in \interior I_{E^*}^j$, and $s$ is chosen
greater than $t$, then $t \in \interior I_{\dom x'}^j$, and condition
\ref{s:flow} for $(x', u')$ at $(t, j)$ gives the same condition for $(x^*,
u^*)$ at $(t, j)$. Let $(t, j) \in E^*$ with $(t, j + 1) \in E^*$. By the
foregoing, with $(T, J) \triangleq (t, j + 1)$, there exists $(x', u') \in
\mathcal{Q}$ such that $(t, j), (t, j + 1) \in \dom x'$, and condition
\ref{s:jump} for $(x', u')$ at $(t, j)$ gives the same condition for $(x^*,
u^*)$. Finally, $\dom x \subseteq E^*$, and $x^* = x$ upon
$\dom x$, every element of $\mathcal{Q}$ belonging to $\mathcal{A}$, whence
\[
(x^*(0, 0), u^*(0, 0)) = (x(0, 0), u(0, 0)) \in \closure C \cup D
\]
The pair
$(x^*, u^*)$ is accordingly a solution pair to $\mathcal{H}$ under $w$ that
belongs to $\mathcal{A}$, and it is an upper bound of $\mathcal{Q}$.

By Zorn's lemma, $\mathcal{A}$ thus has an element $(\bar{x}, \bar{u})$ that is
maximal with respect to $\preceq$. A solution pair $(x', u')$ to $\mathcal{H}$
under $w$ such that $\dom \bar{x} \subset \dom x'$ and $x' = \bar{x}$ upon $\dom
\bar{x}$ would belong to $\mathcal{A}$ and satisfy $(\bar{x}, \bar{u}) \preceq
(x', u') \neq (\bar{x}, \bar{u})$, contrary to the choice of $(\bar{x},
\bar{u})$. Hence, $(\bar{x}, \bar{u})$ is maximal under $w$ in the sense of
Definition~\ref{def:maximal}, and, belonging to $\mathcal{A}$, it extends $(x,
u)$ as the lemma asserts.
\end{proof}

Throughout the rest of this subsection, $\mathcal{H} = (C, f, D, g, h)$ is a hybrid system
with inputs and outputs in the sense of Definition~\ref{def:hybrid} that
satisfies the hybrid basic conditions in the sense of Definition~\ref{def:hbc},
and $n_s$, $n_i$, and $n_o$ are its dimensions.

\begin{definition}[Invariance]\label{def:invariance}
A set $\mathcal{S} \subseteq \mathbb{R}^{n_s}$ is \emph{invariant} for
$\mathcal{H}$ if the two following conditions are satisfied:
\begin{enumerate}[label=(I\arabic*), ref=I\arabic*, leftmargin=*, nosep]
\item\label{i:flow} for all $a, b \in \mathbb{R}$ with $a < b$, all continuous
$v : [a, b] \longrightarrow \mathbb{R}^{n_i}$, and all continuously
differentiable $z : [a, b] \longrightarrow \mathbb{R}^{n_s}$ such that $z(a)
\in \mathcal{S}$ and $\dot{z}(t) = f(z(t), v(t))$ for all $t \in [a, b]$, one
has $z(t) \in \mathcal{S}$ for all $t \in [a, b]$
\item\label{i:jump} for all $(x, \mu) \in D$ with $x \in \mathcal{S}$, one has
$g(x, \mu) \in \mathcal{S}$
\end{enumerate}
\end{definition}

\begin{lemma}[Solution pairs remain in an invariant set]\label{lem:invariance}
Let $\mathcal{S} \subseteq \mathbb{R}^{n_s}$ be invariant for $\mathcal{H}$,
and let $(x, u)$ be a solution pair to $\mathcal{H}$ with $x(0, 0) \in
\mathcal{S}$. Then, $x(t, j) \in \mathcal{S}$ for all $(t, j) \in \dom x$.
\end{lemma}

\begin{proof}
Let $E \triangleq \dom x$, let $(T, J) \in E$, and let $\{ t_j \}_{j = 0}^{J +
1}$ generate $E_{(T, J)}$. Suppose that $x(t_j, j) \in \mathcal{S}$ for some $j
\in \{ 0, 1, \dots, J \}$, as holds for $j = 0$ by hypothesis. If $t_j < t_{j +
1}$, then $x(\cdot, j)$ is continuously differentiable upon $[t_j, t_{j + 1}]$
by Definition~\ref{def:arc} and satisfies $\dot{x}(t, j) = f(x(t, j), u(t,
j))$ for all $t \in (t_j, t_{j + 1})$ by condition \ref{s:flow}, and, hence, by
continuity, for all $t \in [t_j, t_{j + 1}]$; condition \ref{i:flow} then gives
$x(t, j) \in \mathcal{S}$ for all $t \in [t_j, t_{j + 1}]$, which holds
trivially if $t_j = t_{j + 1}$. If, moreover, $j < J$, then $(x(t_{j + 1}, j),
u(t_{j + 1}, j)) \in D$ and $x(t_{j + 1}, j + 1) = g(x(t_{j + 1}, j), u(t_{j +
1}, j))$ by condition \ref{s:jump}, whence $x(t_{j + 1}, j + 1) \in
\mathcal{S}$ by condition \ref{i:jump}. By induction upon $j$, $x(t, j) \in
\mathcal{S}$ for all $(t, j) \in E_{(T, J)}$, and the assertion follows, $(T,
J) \in E$ being arbitrary.
\end{proof}

The proposition that follows is closely related to Proposition~9 of Ref.
\cite{heemels_hybrid_2021}, which concerns a broader class of hybrid systems
with inputs; a direct proof is included for completeness.

\begin{proposition}[Completeness]\label{prop:completeness}
Let $w : \mathbb{R}_{\geq 0} \longrightarrow \mathbb{R}^{n_i}$ be continuous,
and let $\mathcal{S} \subseteq \mathbb{R}^{n_s}$ be invariant for $\mathcal{H}$
and such that
\begin{equation}\label{eq:covering_S}
C \cup D = \mathcal{S} \times \mathbb{R}^{n_i}
\end{equation}
Let $(x, u)$ be a solution pair to $\mathcal{H}$ that is maximal under $w$ and
such that, for all $T \in \mathbb{R}_{\geq 0}$, the set $\left\{ x(t, j) : (t,
j) \in \dom x, \; t \leq T \right\}$ is bounded. Then, $(x, u)$ is complete.
\end{proposition}

\begin{proof}
Let $E \triangleq \dom x$, and suppose, to the contrary, that $(x, u)$ is not
complete, whence $J \triangleq \sup_j E$ and $T \triangleq \sup_t E$ are
finite. By Definition~\ref{def:htd} and
Lemmas~\ref{lem:flow_interval}--\ref{lem:truncation}, there exist $t_0, \dots,
t_J \in \mathbb{R}_{\geq 0}$ with $0 = t_0 \leq \cdots \leq t_J$ such that $E$
is the union of the sets $[t_j, t_{j + 1}] \times \{ j \}$ with $j \in \{ 0,
\dots, J - 1 \}$ and $I_E^J \times \{ J \}$, where $I_E^J$ is either $[t_J, T]$
or $[t_J, T)$ with $t_J < T$. By Definition~\ref{def:solution}, \ref{a:closed},
and Eq. \eqref{eq:covering_S}, $x(0, 0) \in \mathcal{S}$, whence $x(t, j) \in
\mathcal{S}$ for all $(t, j) \in E$ by Lemma~\ref{lem:invariance}. In each of
the cases that follow, $x$ is extended to a map $x'$ upon a compact hybrid time
domain $E' \supset E$ such that $(x', u')$, with $u'$ the hybrid signal induced
by $w$ upon $E'$, is a solution pair under $w$, contradicting the maximality of
$(x, u)$; the requirements of Definitions~\ref{def:arc} and \ref{def:solution}
that the extension does not affect hold for $(x', u')$ because they hold for
$(x, u)$.

Suppose first that $I_E^J = [t_J, T)$. By hypothesis, the set $\{ x(t, J) : t
\in I_E^J \}$ is bounded and is therefore contained in a compact set $K
\subseteq \mathbb{R}^{n_s}$, and the set 
\[
L \triangleq \{ w(t) : t \in [t_J,T] \}
\] 
is compact, $w$ being continuous. The set $K \times L$ is accordingly
compact, and $f$, being continuous by \ref{a:flow}, is bounded upon it. By
condition \ref{s:flow}, the derivative of $x(\cdot, J)$ upon $\interior I_E^J$
is $t \mapsto f(x(t, J), w(t))$, which is accordingly bounded, $(x(t, J),
w(t))$ belonging to $K \times L$ for all $t \in I_E^J$; hence, the limit $\xi
\triangleq \lim_{t \rightarrow T^{-}} x(t, J)$ exists. Let $E'$ be generated
by $(t_0, \dots, t_J, T)$, and let $x'(T, J) \triangleq \xi$. Since $\lim_{t
\rightarrow T^{-}} f(x(t, J), w(t)) = f(\xi, w(T))$, $f$ and $w$ being
continuous, the mean value theorem shows that $x'(\cdot, J)$ is continuously
differentiable upon $[t_J, T]$, with $\dot{x}'(T, J) = f(\xi, w(T))$.
Conditions \ref{s:flow} and \ref{s:jump} are unaffected, since $\interior
I_{E'}^J = \interior I_E^J$ and $(T, J + 1) \notin E'$.

Suppose now that $I_E^J = [t_J, T]$, and let $\xi \triangleq x(T, J)$, whence
$(\xi, w(T)) \in C \cup D$ by Eq. \eqref{eq:covering_S}. If $(\xi, w(T)) \in
D$, let $E'$ be generated by $(t_0, \dots, t_J, T, T)$, and let $x'(T, J + 1)
\triangleq g(\xi, w(T))$. Then, $(x'(T, J), u'(T, J)) = (\xi, w(T)) \in D$ and
$x'(T, J + 1) = g(x'(T, J), u'(T, J))$, as condition \ref{s:jump} requires at
$(T, J)$, while $I_{E'}^{J + 1} = \{ T \}$, whence $J + 1 \notin
\mathcal{J}_{E'}$, and Definition~\ref{def:arc} and condition \ref{s:flow}
impose nothing upon $x'(\cdot, J + 1)$. Otherwise, $(\xi, w(T)) \in C
\setminus D$, and, by the existence theorem of Peano, there exist $\varepsilon
\in \mathbb{R}_{>0}$ and a continuously differentiable map $z : [T, T +
\varepsilon] \longrightarrow \mathbb{R}^{n_s}$ such that $z(T) = \xi$ and
$\dot{z}(t) = f(z(t), w(t))$ for all $t \in [T, T + \varepsilon]$. By
condition \ref{i:flow}, $z(t) \in \mathcal{S}$ for all such $t$, whence, $D$
being closed, $\varepsilon$ can be chosen so that 
\[
(z(t), w(t)) \in (\mathcal{S} \times \mathbb{R}^{n_i}) \setminus D \subseteq C
\]
for all $t \in [T, T + \varepsilon]$. Let $E'$ be generated by the finite sequence
$(t_0, \dots, t_J, T + \varepsilon)$, and let $x'(t, J) \triangleq z(t)$ for all 
$t \in (T, T + \varepsilon]$, so that $x'(\cdot, J)$ is continuous upon 
$[t_J, T + \varepsilon]$, $z(T)$ being equal to $\xi = x(T, J)$. If $t_J < T$, then
condition \ref{s:flow}, the choice of $z$, and the continuity of $f$, $w$,
$x(\cdot, J)$, and $z$ give
\[
\begin{split}
\lim_{t \rightarrow T^{-}} \dot{x}'(t, J)
  & = \lim_{t \rightarrow T^{-}} f(x(t, J), w(t)) = f(\xi, w(T)) \\
  & = \lim_{t \rightarrow T^{+}} f(z(t), w(t))
    = \lim_{t \rightarrow T^{+}} \dot{x}'(t, J)
\end{split}
\]
whence, by the mean value theorem, $x'(\cdot, J)$ is differentiable at $T$,
with $\dot{x}'(T, J) = f(\xi, w(T))$, and its derivative is continuous at $T$.
In either case, $x'(\cdot, J)$ is therefore continuously differentiable upon
$[t_J, T + \varepsilon]$, and $(x'(t, J), u'(t, J)) \in C$ and $\dot{x}'(t, J)
= f(x'(t, J), u'(t, J))$ for all $t \in [T, T + \varepsilon]$, which, with
condition \ref{s:flow} for $(x, u)$, gives condition \ref{s:flow} for $(x',
u')$. Condition \ref{s:jump} is unaffected, since $(t, J + 1) \notin E'$ for
all $t \in [T, T + \varepsilon]$.
\end{proof}

\section{Model}\label{sec:model}

\subsection{Binary Direct Collinear Collisions of Rigid Bodies}\label{sec:bodies}

The discussion that follows is with reference to Fig. \ref{fig:nc}. The
notational conventions for mechanics are adopted from Refs.
\cite{roithmayr_dynamics_2016} and \cite{stronge_impact_2018}. It is
assumed that $\mathcal{B}_1$ and $\mathcal{B}_2$ are compact convex rigid bodies,
each of which is axisymmetric about the line $A'B'$, and that the motion of each
body is a translation along $A'B'$. The centers of mass $G_1$ and $G_2$ of the
two bodies lie upon $A'B'$ and are assumed to be distinct in the initial configuration; let
$\hat{\bm{n}}_1$ be a unit vector directed along $A'B'$ from $G_2$ toward
$G_1$, and let $O \triangleq (0, 0, 0)$ be a point of $A'B'$ that is fixed in an
inertial frame. Let $C_1$ denote the point of $\mathcal{B}_1 \cap A'B'$ that is
extreme in the direction of $-\hat{\bm{n}}_1$, and $C_2$ the point of
$\mathcal{B}_2 \cap A'B'$ that is extreme in the direction of
$\hat{\bm{n}}_1$; then, $\bm{r}_{C_1/G_1} = -l_1 \hat{\bm{n}}_1$ and
$\bm{r}_{C_2/G_2} = l_2 \hat{\bm{n}}_1$ for some $l_1, l_2 \in
\mathbb{R}_{\geq 0}$. It is assumed, further, that $C_1$ is the only point of
$\mathcal{B}_1$ that is extreme in the direction of $-\hat{\bm{n}}_1$, or $C_2$
the only point of $\mathcal{B}_2$ that is extreme in the direction of
$\hat{\bm{n}}_1$. Initially, $\mathcal{B}_1 \cap \mathcal{B}_2 \subseteq \{ C_1 \} \cap \{ C_2
\}$, so that the bodies either are apart, and can meet without passing through one another, or touch at a single point. Since the
bodies are convex and axisymmetric, and their motion is confined to the axis of
symmetry, any region of contact between them is symmetric about $A'B'$; the
resultant of the contact tractions therefore acts along $A'B'$ and exerts no
moment upon either body. The configuration, as
hereinbefore described, corresponds to a binary direct collinear impact (e.g.,
see Ref. \cite{stronge_impact_2018}). Building upon the methodology proposed, for example, in Refs. \cite{movahedi-lankarani_canonical_1988, flores_contact_2016, stronge_impact_2018}, it is assumed that while the bodies remain in contact, the motion of the system is governed by the laws of rigid body dynamics (these laws are usually attributed to Sir Isaac Newton \cite{newton_mathematical_1729} and Leonhard Euler \cite{euler_theoria_1765}). 

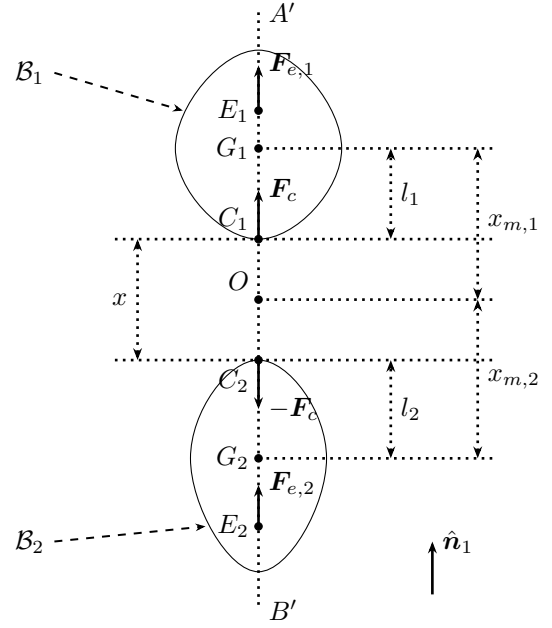
\begin{figure}
\centering

\begin{tikzpicture}[>={Stealth[scale=0.6]}]
 
    \coordinate (O) at (0, 0);
    \coordinate (C1) at (0, 0.8);
    \coordinate (C2) at (0, -0.8);
    \coordinate (G1) at (0, 2.0);
    \coordinate (G2) at (0, -2.1);
    \coordinate (E1) at (0, 2.5);
    \coordinate (E2) at (0, -3.0);
 
    \draw plot [smooth cycle, tension=.8]
        coordinates {(0,0.8) (1.1,2.0) (0,3.3) (-1.1,2.0)};
    \draw plot [smooth cycle, tension=.8]
        coordinates {(0,-0.8) (0.9,-2.1) (0,-3.6) (-0.9,-2.1)};
 
    \draw[dotted, line width=1](0,3.8) node [right] {$A'$}
        -- (0,-4.1) node [right] {$B'$};
 
    \draw[dotted, line width=1](0,0) -- (3.1,0);
    \draw[dotted, line width=1](-1.9,0.8) -- (3.1,0.8);
    \draw[dotted, line width=1](-1.9,-0.8) -- (3.1,-0.8);
    \draw[dotted, line width=1](0,2.0) -- (3.1,2.0);
    \draw[dotted, line width=1](0,-2.1) -- (3.1,-2.1);
 
    \draw[<->, dotted, line width=1](-1.6,-0.8) -- (-1.6,0.8)
        node [midway, left] {$x$};
    \draw[<->, dotted, line width=1](1.75,0.8) -- (1.75,2.0)
        node [midway, right] {$l_1$};
    \draw[<->, dotted, line width=1](1.75,-2.1) -- (1.75,-0.8)
        node [midway, right] {$l_2$};
    \draw[<->, dotted, line width=1](2.9,0) -- (2.9,2.0)
        node [midway, right] {$x_{m,1}$};
    \draw[<->, dotted, line width=1](2.9,-2.1) -- (2.9,0)
        node [midway, right] {$x_{m,2}$};
 
    \draw[->, line width=1](C1) -- (0,1.45) node [right] {$\bm{F}_c$};
    \draw[->, line width=1](C2) -- (0,-1.45) node [right] {$-\bm{F}_c$};
    \draw[->, line width=1](E1) -- ($(E1) + (0,0.6)$)
        node [right] {$\bm{F}_{e,1}$};
    \draw[->, line width=1](E2) -- ($(E2) + (0,0.55)$)
        node [right] {$\bm{F}_{e,2}$};
 
    \draw[fill=black](O) circle (1.4 pt) node [above left] {$O$};
    \draw[fill=black](C1) circle (1.4 pt) node [above left] {$C_1$};
    \draw[fill=black](C2) circle (1.4 pt) node [below left] {$C_2$};
    \draw[fill=black](G1) circle (1.4 pt) node [left] {$G_1$};
    \draw[fill=black](G2) circle (1.4 pt) node [left] {$G_2$};
    \draw[fill=black](E1) circle (1.4 pt) node [left] {$E_1$};
    \draw[fill=black](E2) circle (1.4 pt) node [left] {$E_2$};
 
    \draw[->, line width=0.75, dashed](-2.7,3.0) node [left] {$\mathcal{B}_1$}
        -- (-1.0,2.5);
    \draw[->, line width=0.75, dashed](-2.7,-3.2) node [left] {$\mathcal{B}_2$}
        -- (-0.7,-3.0);
 
    \draw[->, line width=1](2.3,-3.9) -- (2.3,-3.2)
        node [right] {$\hat{\bm{n}}_1$};
 
\end{tikzpicture}
\caption{Bodies $\mathcal{B}_1$ and $\mathcal{B}_2$.}\label{fig:nc}
\end{figure}

The interaction force $\bm{F}_c \triangleq F_c \hat{\bm{n}}_1$ with $F_c :
\mathbb{R}_{\geq 0} \longrightarrow \mathbb{R}$ is applied to
$\mathcal{B}_1$ at the point $C_1$, and the corresponding force that acts on the
body $\mathcal{B}_2$ is $-\bm{F}_c \triangleq -F_c \hat{\bm{n}}_1$, applied
at the point $C_2$. The bodies may or may not be in contact at any given
instant; when they are not, $F_c$ vanishes identically. What it is for the bodies
to be in contact is left unspecified for the present, and is made precise in
Section~\ref{sec:interface}. Throughout the present section, $F_c$ is likewise
unspecified, being furnished by a contact law, as explained in Section~\ref{sec:interface};
its continuity is examined in Section~\ref{sec:continuity}.

It is also assumed that the resultant of the external forces, other than
$\bm{F}_c$, that act on the body $\mathcal{B}_i$ is $\bm{F}_{e,i} \triangleq
F_{e,i} \hat{\bm{n}}_1$,\footnote{In what follows, it shall always be assumed that
$i$ ranges over the set $\{ 1, 2 \}$.} and that its point of application $E_i$ (with reference
to $\mathcal{B}_i$) lies upon $A'B'$, so that it exerts no moment upon
$\mathcal{B}_i$. It shall be assumed that $F_{e,i}$ is continuous and bounded. 
Where a narrower class of loads is required, it will be identified at the point of use.

The motion of the system is confined
to a single dimension, and the two coordinates $x_{m,1} \in \mathbb{R}$ and
$x_{m,2} \in \mathbb{R}$ such that $\bm{r}_{G_1/O} \triangleq x_{m,1}
\hat{\bm{n}}_1$ and $\bm{r}_{G_2/O} \triangleq x_{m,2} \hat{\bm{n}}_1$
suffice to describe the dynamics of the system. More specifically, the motion
can be described by the following initial value problem (IVP):
\begin{equation}\label{eq:primary}
\begin{cases}
\ddot{x}_{m,1} = m_1^{-1} F_c + m_1^{-1} F_{e,1}  \\
\ddot{x}_{m,2} = -m_2^{-1} F_c + m_2^{-1} F_{e,2}  \\
x_{m,1}(0) = l_1 + x_{1,0}, \dot{x}_{m,1}(0) = v_{m,1,0}  \\
x_{m,2}(0) = -l_2 + x_{2,0}, \dot{x}_{m,2}(0) = v_{m,2,0}
\end{cases}
\end{equation}
where $m_1, m_2 \in \mathbb{R}_{> 0}$ are the masses of the bodies and
$x_{1,0}, x_{2,0}, v_{m,1,0}, v_{m,2,0} \in \mathbb{R}$ are parameters.

Defining $x_1 \triangleq x_{m,1} - l_1$, $x_2 \triangleq x_{m,2} + l_2$, $v_{1,
0} \triangleq v_{m,1,0}$, and $v_{2, 0} \triangleq v_{m,2,0}$, the IVP given by
Eq. \eqref{eq:primary} can be restated as
\begin{equation}\label{eq:primary_contact}
\begin{cases}
\ddot{x}_1 = m_1^{-1} F_c + m_1^{-1} F_{e,1} \\
\ddot{x}_2 = -m_2^{-1} F_c + m_2^{-1} F_{e,2} \\
x_1(0) = x_{1,0}, \dot{x}_1(0) = v_{1,0} \\
x_2(0) = x_{2,0}, \dot{x}_2(0) = v_{2,0} 
\end{cases}
\end{equation}
Since $\bm{r}_{C_1/O} = x_1 \hat{\bm{n}}_1$ and $\bm{r}_{C_2/O} =
x_2 \hat{\bm{n}}_1$, this IVP describes the evolution of the extremal points
$C_1$ and $C_2$ of the two bodies.

Denoting
\begin{equation}
m \triangleq \frac{m_1 m_2}{m_1 + m_2}
\end{equation}
\begin{equation}
x \triangleq x_1 - x_2
\end{equation}
\begin{equation}
v \triangleq \dot{x} = \dot{x}_1 - \dot{x}_2
\end{equation}
\begin{equation}
x_0 \triangleq x_{1,0} - x_{2,0}
\end{equation}
\begin{equation}
v_0 \triangleq v_{1,0} - v_{2,0}
\end{equation}
\begin{equation}\label{eq:u}
F_{e,r} \triangleq \frac{m_2 F_{e,1} - m_1 F_{e,2}}{m_1 + m_2}
\end{equation}
the equations of motion can be transformed to
\begin{equation}\label{eq:main}
\begin{cases}
\dot{x} = v & x(0) = x_0 \\
\dot{v} = m^{-1} F_c + m^{-1} F_{e,r} & v(0) = v_0 \\
\end{cases}
\end{equation}
It should be noted that $m \in \mathbb{R}_{>0}$ describes the effective mass of
the colliding bodies (e.g., see Ref. \cite{nikravesh_determination_2023}). The
quantity $x$ will be referred to as the relative displacement and $v$ as the
relative velocity. Since the bodies are convex and axisymmetric, $C_1$ and $C_2$
are extreme in the direction of one another among all points of the bodies, and
not merely of their intersections with $A'B'$; $x$ is therefore the signed gap
between the bodies, and, by the assumption of uniqueness, the bodies are apart
if $x > 0$, they touch at a single point if $x = 0$, and, if $x < 0$, they
overlap or have passed through one another.\footnote{Under these constraints, the detection of
contact reduces to the monitoring of a single scalar quantity, the relative displacement $x$;
for the general case, the methodology suggested in Ref.
\cite{pfeiffer_multibody_2004} may be employed.} The assumptions made above concerning the initial configuration are expressed in
these terms as $x_0 \in \mathbb{R}_{\geq 0}$ and 
\begin{equation}
x_{m,1}(0) - x_{m,2}(0) = x_0 + l_1 + l_2 \in \mathbb{R}_{>0}
\end{equation}

It should be noted that if
(by abuse of notation) $m_2 = +\infty$, then $m^{-1} = m_1^{-1}$ and $F_{e,r} =
F_{e,1}$. This situation corresponds to the collision of a body $\mathcal{B}_1$
of finite mass with a stationary body $\mathcal{B}_2$.

The IVP given by Eq. \eqref{eq:main} governs the relative motion of the two
bodies alone. To recover the motion of each body, it must be supplemented by a
description of the motion of the center of mass of the system, which is
considered next; both masses are taken to be finite throughout. First, introduce
the parameters $\eta, \eta_c \in (0, 1)$ given by
\begin{equation}
\eta \triangleq \frac{m_1}{m_1 + m_2}
\end{equation}
\begin{equation}
\eta_c \triangleq 1 - \eta = \frac{m_2}{m_1 + m_2}
\end{equation}
The location of the center of mass is given by
\begin{equation}
\bm{r}_{G/O} = \eta \bm{r}_{G_1/O} + \eta_c \bm{r}_{G_2/O}
  \triangleq x_m \hat{\bm{n}}_1
\end{equation}
Then,
\begin{equation}\label{eq:x_m}
x_m \triangleq \eta x_{m,1} + \eta_c x_{m,2}
  = \eta x_1 + \eta_c x_2 + \eta l_1 - \eta_c l_2
\end{equation}
Introducing $v_m \triangleq \dot{x}_m$, the IVP associated with the evolution of
the location of the center of mass of the system is given by
\begin{equation}\label{eq:com}
\begin{cases}
\dot{x}_m = v_m & x_m (0) = x_{m,0} \\
\dot{v}_m = (m_1 + m_2)^{-1} (F_{e,1} + F_{e,2}) & v_m (0) = v_{m,0}
\end{cases}
\end{equation}
where $x_{m,0} \in \mathbb{R}$ and $v_{m,0} \in \mathbb{R}$ are given by
\begin{equation}\label{eq:x_m_0}
x_{m,0} \triangleq \eta x_{1,0} + \eta_c x_{2,0} + \eta l_1 - \eta_c l_2
\end{equation}
\begin{equation}\label{eq:v_m_0}
v_{m,0} \triangleq \eta v_{m,1,0} + \eta_c v_{m,2,0}
\end{equation}
respectively. Once $x$
and $x_m$ are known, the evolution of the locations of the centers of mass of
the individual bodies can be recovered via
\begin{equation}\label{eq:x_m_1}
x_{m,1} = x_m + \eta_c x + \eta_c (l_1 + l_2)
\end{equation}
\begin{equation}\label{eq:x_m_2}
x_{m,2} = x_m - \eta x - \eta (l_1 + l_2)
\end{equation}
Then, the velocities of the centers of mass of the bodies can be obtained via
\begin{equation}\label{eq:v_m_1}
v_{m,1} \triangleq \dot{x}_{m,1} = v_m + \eta_c v
\end{equation}
\begin{equation}\label{eq:v_m_2}
v_{m,2} \triangleq \dot{x}_{m,2} = v_m - \eta v
\end{equation}

The system of Eqs. \eqref{eq:main} and \eqref{eq:com}, with the state $(x, v,
x_m, v_m)$, the inputs $F_{e,1}$ and $F_{e,2}$ and the contact force $F_c$, and
the outputs $(v_{m,1}, v_{m,2})$ given by Eqs. \eqref{eq:v_m_1} and
\eqref{eq:v_m_2} and $v$, shall be called the rigid body subsystem and denoted
$\mathcal{G}_b$.

\subsection{Contact Interface}\label{sec:interface}

Two matters were left unspecified in Section~\ref{sec:bodies}: the relation by
which $F_c$ is furnished, and what it is for the bodies to be in contact. The
second is taken up first. A contact commences when the bodies, out of contact,
touch or overlap while approaching one another; it persists for as long as the
interface presses them apart; and it is terminated at the instant at which the
contact force ceases to be positive, the bodies not approaching one another. A contact is thus delimited
by the contact force. It might have been delimited by the relative
displacement, a contact being taken to occupy the interval during which the
bodies overlap; and, for an elastic contact law, nothing turns upon the choice,
the force and the overlap vanishing together.

The contact laws with which this article is concerned need not possess this
property: the contact force may be extinguished at an indentation that is
strictly negative, the interface being still deformed when it ceases to transmit
force, whether because the law leaves a residual indentation or, for a law with
a viscous component, because an applied force draws the bodies apart faster
than the interface recovers.
The two criteria come apart, and the choice made above is no longer
idle. Two consequences follow from it. The bodies part while the interface is
deformed. And the indentation is not the relative displacement:
it is obtained from the configuration at which the contact commenced, which,
the bodies having parted while overlapping, need not be the configuration in
which they touch. The deformation of the interface is accordingly a quantity
that the system of Eq.~\eqref{eq:main} does not contain. It must be carried
separately, and what becomes of it when the bodies part must be said.

It was observed in Section~\ref{sec:introduction} that a permanent indentation
cannot be recorded in the configuration of a system of rigid bodies, and that
the tension so arising is relieved by an argument of scale. That argument is now
made to do further work. Two assumptions are required. The indentation is taken to be small in comparison with the dimensions of the bodies, whence the surfaces are left as they are and the residual deformation is borne entirely by the constitutive description. The area of surface that a
residual deformation occupies is, on the same reckoning, small; and it is taken
that a subsequent contact does not fall upon it. What warrants this is that the
point at which the bodies meet does not stand still. It travels over the
surfaces as the bodies move, and an impact that is oblique, or that sets a body
turning, shifts it by a distance far exceeding the width of an indented patch.
It follows from the second assumption that, at the outset of every contact, the
interface is undeformed, and that each contact is therefore a fresh loading of
material that has not been indented before. For a solid displaying an elastic aftereffect, the second assumption may be exchanged for another. If the recovery is complete within the
interval separating two contacts, then the interface is undeformed at the outset
of each whether or not the same region of the surface is engaged again. The
requirement is then upon the recovery time rather than upon the migration of the
point of contact.

A word is owed upon the consistency of the second assumption with the
configuration of Section~\ref{sec:bodies}, in which the motion is collinear and
the bodies meet at the same point upon every occasion. Taken by itself, that
configuration is precisely the one in which the point of contact does not
migrate. It is nonetheless the configuration in which a normal contact law is customarily
developed, for it isolates the normal response from the tangential one; and a
law so developed is ordinarily employed as the normal component of a model of
oblique impact, in which the coupling is unidirectional, the normal force
entering the relations that govern the friction and the transfer of angular
momentum, and not conversely. The direct collinear collision is
retained here as the simplest configuration in which the interface may be
exhibited, and not as that in which the assumptions are most nearly satisfied.

The two assumptions being granted, the model may be assembled. The system of
Eq. \eqref{eq:main} evolves in two distinct regimes. While the bodies are out of
contact, $F_c$ vanishes and the relative motion is governed by the external force
alone. While the bodies are in contact, $F_c$ is furnished by an incremental
contact law, and the deformation of the interface accumulates from the instant
at which that contact commenced. The passage from either regime to the other
occurs at instants determined by the motion itself, and the deformation
accumulated during a contact does not survive it, the interface being restored
to its undeformed configuration. Both assumptions bear upon this restoration: by
the first, the geometry of the bodies retains no record of the indentation, and,
by the second, the material engaged by the contact that follows has not been
indented before. It is a further consequence of the second that the contact law
need describe a single loading from an undeformed state, and not a sequence of
loadings; the laws that are ordinarily employed are of exactly this kind, and
are thereby admitted into the framework unaltered.

The instants at which the passage between the regimes occurs, left unspecified
in Section~\ref{sec:bodies}, may now be given. A contact commences, as above,
when the bodies meet while approaching one another. It is terminated when the
contact force is no longer positive, the bodies not approaching one another,
which may happen while they still overlap.

These two conditions, upon the meeting of the bodies and upon the sign of
the force, are moreover conditions upon different quantities. The
relative displacement, which governs the commencement of a contact, is a
quantity of the system of Eq. \eqref{eq:main}; the contact force, which governs
its termination, is a quantity of the interface, depending upon the indentation
rather than upon the relative displacement, and these two may differ by the gap
at which the contact commenced. The interface is accordingly furnished with a state of its own, comprising the
indentation, the internal state of the contact law, and a variable recording
which of the two regimes is active. It receives the relative velocity of the bodies and returns the force exerted
upon it by them, the contact force being its negative, and it is by this
exchange alone that it is coupled to them. The commencement of a contact and the
separation of the bodies are transitions of this state. What results is a hybrid system in
the sense of Section~\ref{sec:hybrid}. 

\begin{figure*}
\centering
 
\begin{tikzpicture}[>={Stealth[scale=0.6]}]
 
    \def\rb{0.55}
 
    \foreach \i in {-0.6,-0.3,...,14.6}{
        \draw[line width=0.35](\i,0) -- (\i-0.22,-0.28);}
 
    \draw[dotted, line width=1] plot [smooth, tension=0.7]
        coordinates {(0,2.2) (1.25,0.55) (2.5,0.13) (3.75,0.33) (5.0,1.35)
                     (6.25,1.8) (7.5,0.55)};
    \draw[dotted, line width=1] plot [smooth, tension=0.7]
        coordinates {(10.0,0.28) (11.25,0.42) (12.5,0.28) (13.75,0.25)};
 
    \foreach \p/\h/\d in {0/2.2/-1, 1.25/0.55/-1, 2.5/0.13/0, 3.75/0.33/1,
                          5.0/1.35/1, 6.25/1.8/0, 7.5/0.55/-1,
                          10.0/0.28/1, 11.25/0.42/0, 12.5/0.28/-1,
                          13.75/0.25/0}{
        \begin{scope}
            \clip (-1,-1.0) rectangle (14.9,0);
            \fill[white](\p,\h) circle (\rb);
            \fill[black!14](\p,\h) circle (\rb);
            \draw[dashed, line width=0.7](\p,\h) circle (\rb);
        \end{scope}
        \begin{scope}
            \clip (-1,0) rectangle (14.9,3.2);
            \draw[line width=0.9](\p,\h) circle (\rb);
        \end{scope}
        \ifnum\d=-1 \draw[->, line width=0.9](\p,\h+0.28) -- (\p,\h-0.28);\fi
        \ifnum\d=1 \draw[->, line width=0.9](\p,\h-0.28) -- (\p,\h+0.28);\fi}
 
    \draw[<->, line width=0.9](13.75,0.03) -- (13.75,0.47);
 
    \draw[line width=1](-0.9,0) -- (14.6,0);
 
    \node at (8.75,0.6) {$\cdots$};
 
    \foreach \p/\l in {0/a, 1.25/b, 2.5/c, 3.75/d, 5.0/e, 6.25/f, 7.5/g,
                       10.0/h, 11.25/i, 12.5/j, 13.75/k}{
        \node at (\p,-0.8) {(\l)};}
 
    \draw[->, line width=0.9](-0.9,-1.35) -- (14.6,-1.35);
    \node at (6.85,-1.75) {time, not to scale};
 
\end{tikzpicture}
 
\caption{Successive stages of a sequence of collisions of a ball with a fixed
obstacle under gravity. The portion of the
ball lying below the undeformed surface is shaded. The indentation is exaggerated.}\label{fig:stages}
 
\end{figure*}
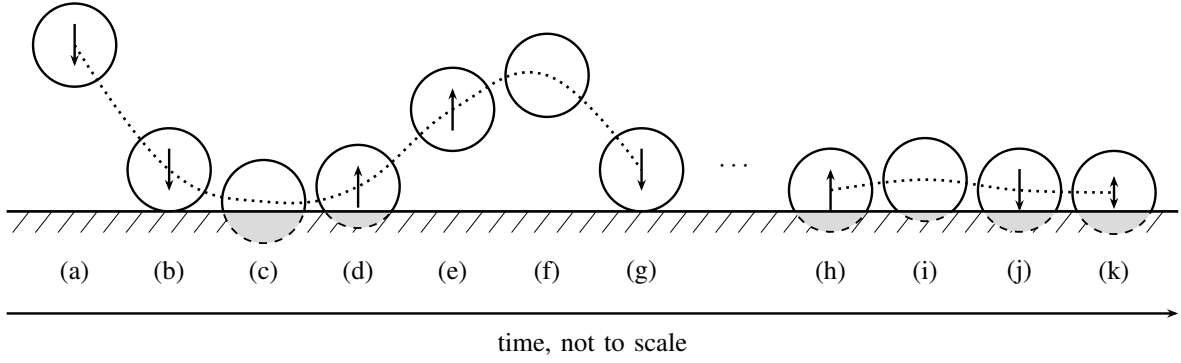

The two regimes, the transitions between them, and the restoration effected at
those transitions are exhibited together in Fig. \ref{fig:stages}, for a ball
falling upon a fixed obstacle under gravity. The ball descends (a) and meets the
obstacle while approaching it (b), whereupon a contact commences and the
interface, undeformed at that instant, begins to be compressed. The compression
attains its greatest extent (c) and thereafter diminishes, and, at (d), the
contact force is extinguished while the bodies are receding and still overlap:
the interface is restored to its undeformed configuration and the bodies part.
The ball ascends (e), reverses (f), and meets the obstacle a second time (g),
the interface being undeformed at the outset of this contact as of the first. As
the sequence proceeds, the ball parts with an ever smaller velocity, and, at (h), a
separation occurs from which it does not rise clear of the surface. It is
thereafter in free flight while the bodies overlap, and a contact commences
afresh at the instant at which its velocity vanishes (i), the bodies then
overlapping already; throughout the contact that follows (j), the indentation and
the relative displacement differ by the overlap at (i). At length, the contact
force is not again extinguished, and the ball settles into an oscillation about
a resting indentation (k).

The commencement of a contact at (i) is a modeling choice, and an approximate
one. The residual indentation of the preceding contact having been discarded at
the separation, the bodies are deemed to meet at the instant at which they cease
to recede, and the interface is engaged from an undeformed configuration at that
instant. Were the impression retained, the bodies would meet only upon regaining
it, and the material then engaged would be that indented before, which the
second assumption above excludes. The configuration at which the contact
commences is thereby in error by the distance through which the bodies rebound
within the overlap and, whatever the recovery of the impression in the interval,
by at most the overlap at the preceding separation; the errors of successive returns may accumulate.

Returns within an overlap typically arise in two ways. In a sequence of collisions such
as that of Fig.~\ref{fig:stages}, they occur late, when the overlap left by a
contact is already small; the error is smaller still, and below what a
phenomenological law can be expected to resolve in that regime. For a material
whose deformation under small loads is fully and promptly recovered, such
returns are an artifact of the law, the contacts concerned being of small
incident velocity, whereas most laws of the differential viscoelastoplastic
kind, those of Examples~\ref{ex:bwshc} and~\ref{ex:bwm} with $\beta > 0$ among them, leave a
residual indentation even under very small contact forces. It is not evident,
however, that every material behaves so. A residual deformation, not
necessarily plastic, may persist under arbitrarily small loads, as it may in a
solid displaying an elastic aftereffect, and the response of real interfaces in the limit of small loads is a
question for experiment. Returns within an overlap arise also when an external
force pulls the bodies apart during a contact and then reverses. For some laws with
a viscous component, a sufficiently rapid recession extinguishes the contact
force while the interface is still deformed, and the overlap at the separation
need not then be small; nor is this an artifact of the law, a viscoelastic
solid recovering more slowly than the bodies may be made to recede. Returns
within an overlap cannot, therefore, be excluded in general by the choice of
the law, and a model of repeated collisions must provide for them, as the
present one does. In the second case, the error need not be small, although a
large error requires that the pull extinguish the contact force at a large
overlap and then reverse quickly enough to return the bodies before they rise
clear, which is an uncommon combination. The error is, in either case, not
concealed: the overlap at the separation, which bounds it, and the distance
through which the bodies rebound within that overlap are quantities of the
computed motion, from which its significance in a given application may be
judged.

What the diagram described above exhibits informally, the remainder of this subsection
constructs. A contact law is first given in the form in which such laws are
customarily written, and the conditions that it is required to satisfy are
stated; the data of the hybrid system are assembled from the law thereafter,
without alteration to it.

An incremental contact law is assumed to be specified by a differential relation
driven by the deformation of the contact interface. Let $n_z \in
\mathbb{Z}_{\geq 0}$, and let $\Phi : \mathbb{R} \times \mathbb{R}^{n_z} \times \mathbb{R} \longrightarrow
\mathbb{R}$ and $\Psi : \mathbb{R} \times \mathbb{R}^{n_z} \times \mathbb{R} \longrightarrow
\mathbb{R}^{n_z}$ be continuous. The incremental contact laws considered in this article have the form
\begin{equation}\label{eq:law}
\begin{cases}
\dot{x}_c = v_c & x_c(0) = 0 \\
\dot{z}_c = \Psi(x_c, z_c, v_c) & z_c(0) = 0 \\
F_s = \Phi(x_c, z_c, v_c) &
\end{cases}
\end{equation}
in which $x_c$ is the indentation of the interface, $v_c$ is the rate at which
the interface is deformed, $z_c$ is the internal state of the law, and $F_s$ is
the force exerted upon the interface by the bodies. The contact force $F_c$ of
Section~\ref{sec:bodies}, which is exerted upon the bodies by the interface, is
its negative: $F_c = -F_s$. The indentation
is measured from the undeformed configuration of the interface and is negative
in compression, whence $F_s$ is negative in compression likewise, and the power
delivered to the interface by the bodies is $F_s v_c$. The instant $t = 0$ is
that at which the contact commences, at which instant the interface is
undeformed and its internal state is at rest. It is the pair $(\Phi, \Psi)$
that constitutes the law, the initial conditions of Eq. \eqref{eq:law}
recording the configuration at the commencement of a contact. The case $n_z = 0$
is admitted, the internal state being then absent.

The maps $\Phi$ and $\Psi$ are defined upon $\mathbb{R} \times
\mathbb{R}^{n_z} \times \mathbb{R}$. In practice, contact laws are often defined
only on $\mathbb{R}_{\leq 0} \times \mathbb{R}^{n_z} \times \mathbb{R}$. 
A law is nonetheless to be specified by its author upon the whole of $\mathbb{R}
\times \mathbb{R}^{n_z} \times \mathbb{R}$, and the properties considered below,
criteria for which are given in Section~\ref{sec:properties}, are those of the
law so specified.

It is assumed throughout that the initial value problem of Eq. \eqref{eq:law}
possesses a unique solution for every $v_c$ that arises; the establishment of this for a particular law is not attempted here. However, one condition is
imposed upon $\Phi$, namely
\begin{equation}\label{eq:vanishing}
\Phi(0, 0, v_c) = 0
\end{equation}
for all $v_c \in \mathbb{R}$: an interface that has not been deformed transmits
no force, no matter the rate at which it is approached. It is by this condition that the contact force is continuous at the
commencement of a contact: the interface is undeformed at that instant, so that
the force furnished by the law is $\Phi(0, 0, v_c)$, which vanishes, as does
the contact force. The construction of the hybrid interface
model does not otherwise depend upon this condition.

A second property, that of passivity, is desirable but is not made a
requirement. A law is \emph{passive}, in the sense of
Definition~\ref{def:passive_law}, if it admits a stored energy that never
increases faster than the power delivered to the interface by the bodies. The
property is expected of any law
that is to represent an inactive physical interface, and it is possessed by the laws in
common use; but its verification is not always straightforward. It is
accordingly assumed only where it is needed.

A third property bears upon the force transmitted during a contact. A law is
\emph{compressive}, in the sense of Definition~\ref{def:compressive_law}, if the
interface can come into tension only if, at some instant, it transmits tension
while its indentation increases.

Like passivity, the property is not made a requirement, and it is assumed only
where it is needed. A contact being unable to continue while the interface
transmits tension with the bodies receding, a compressive law keeps the force
exerted upon the interface non-positive throughout every contact, by
Lemmas~\ref{lem:contact_compression_Hc} and \ref{lem:contact_compression};
hence, the name. Under a compressive law, the contact force is accordingly
continuous, whatever the relative velocity delivered to the interface, by
Corollary~\ref{cor:force_continuous_Hc}, and, in particular, along the
solutions of the interconnected system, by
Corollary~\ref{cor:force_continuous}. For a law without an internal state, it
suffices for compressiveness that $\Phi$ be non-positive wherever neither the
indentation nor its rate is positive, and positive for small positive
indentations at positive rates, by Proposition~\ref{prop:compressive_nz}.

It remains to describe how the bodies interact, and thereby to furnish the
contact force that Eq. \eqref{eq:main} requires. This is done by modeling the
contact interface as a hybrid system with inputs and outputs $\mathcal{H}_c
\triangleq (C_c, f_c, D_c, g_c, h_c)$ in the sense of
Definition~\ref{def:hybrid}, which receives the relative velocity of the bodies
and returns the force exerted upon the interface by them, and of which a
contact law of the form described above is a constituent.

The state of $\mathcal{H}_c$ is
\begin{equation}\label{eq:Hc_state}
\chi_c \triangleq (x, x_c, z_c, q_c) \in \mathbb{R} \times \mathbb{R} \times
\mathbb{R}^{n_z} \times \mathbb{R}
\end{equation}
in which $x$ is the relative displacement of the bodies, $x_c$ and $z_c$ are the
indentation and the internal state of the contact law, and $q_c$ takes the value
$1$ while the bodies are in contact and $0$ while they are out of contact. As
Definition~\ref{def:hybrid} requires, $q_c$ is nonetheless a real variable, the
maps below being defined for all of its values; that it takes no values other
than $0$ and $1$ along the solutions of the systems of
Section~\ref{sec:interconnection} is established in
Section~\ref{sec:properties}. The input is
the relative velocity $v$ and the output is the force $F_s$ exerted upon the
interface, the contact force being its negative. The
dimensions of $\mathcal{H}_c$ are therefore $n_s = n_z + 3$, $n_i = 1$, and $n_o
= 1$. It should be remarked that the interface carries its own copy of the relative displacement, which
coincides with that of Eq. \eqref{eq:main} along the solutions of the
interconnected system of Section~\ref{sec:interconnection}.

The maps of $\mathcal{H}_c$ are
\begin{equation}\label{eq:Hc_f}
f_c(\chi_c, v) \triangleq \left( v, \; q_c v, \; q_c \, \Psi(x_c, z_c, v), \; 0
\right)
\end{equation}
\begin{equation}\label{eq:Hc_g}
g_c(\chi_c, v) \triangleq \left( x, \; 0, \; 0, \; 1 - q_c \right)
\end{equation}
\begin{equation}\label{eq:Hc_h}
h_c(\chi_c, v) \triangleq q_c \, \Phi(x_c, z_c, v)
\end{equation}
While $q_c = 1$, the interface is deformed at the rate of the
relative velocity, its internal state evolves according to the contact law, and
the force transmitted is that furnished by the law, so that Eq. \eqref{eq:law} is
recovered; while $q_c = 0$, the interface does not evolve at all and the force
transmitted vanishes. At the termination of a contact, the interface is
restored to its undeformed configuration and $q_c$ is set to $0$, and, at the
commencement of a contact, the interface is assumed to be undeformed, so that nothing
is done but to set $q_c$ to $1$.

Writing $\{ q_c = 0 \}$ for the set of pairs $(\chi_c, v)$ whose component $q_c$
is $0$, and similarly for the remaining conditions, the flow set and the jump
set are
\begin{equation}\label{eq:Hc_C}
\begin{split}
C_c \triangleq{} & \left( \{ q_c = 0 \} \cap \left( \{ x \geq 0 \} \cup \{ v
\geq 0 \} \right) \right) \\
& \cup \left( \{ q_c = 1 \} \cap \left( \{ \Phi \leq 0 \} \cup \{ v \leq 0 \}
\right) \right)
\end{split}
\end{equation}
\begin{equation}\label{eq:Hc_D}
\begin{split}
D_c \triangleq{} & \left( \{ q_c = 0 \} \cap \{ x \leq 0 \} \cap \{ v \leq 0 \}
\right) \\
& \cup \left( \{ q_c = 1 \} \cap \{ \Phi \geq 0 \} \cap \{ v \geq 0 \} \right)
\end{split}
\end{equation}
where $\Phi$ is evaluated at $(x_c, z_c, v)$ throughout.

The initial condition is
\begin{equation}\label{eq:Hc_init}
\chi_c(0, 0) = (x_0, 0, 0, 0)
\end{equation}
with $x_0 \in \mathbb{R}_{\geq 0}$, the bodies being apart, or touching, at the
initial instant.

\begin{remark}
The factor $q_c$ in Eq. \eqref{eq:Hc_h} is redundant along the solutions of
$\mathcal{H}_c$ from Eq. \eqref{eq:Hc_init}. It is retained so that the output vanish while the bodies are out of contact at every state, and not merely at those that arise from Eq.
\eqref{eq:Hc_init}.
\end{remark}

\begin{example}[The Bouc-Wen-Simon-Hunt-Crossley Collision Law]\label{ex:bwshc}

The foregoing is abstract, and an instance may be of assistance. Let $n_z = 1$
and let
\begin{equation}\label{eq:bwshc_psi}
\Psi(x_c, z_c, v_c) = A v_c - \beta z_c \abs{z_c}^{n - 1} \abs{v_c}
  - \gamma \abs{z_c}^{n} v_c
\end{equation}
\begin{equation}\label{eq:bwshc_phi}
\begin{split}
\Phi(x_c, z_c, v_c) = {} & \alpha k \, x_c \abs{x_c}^{p - 1} \\
  & + (1 - \alpha) k \, z_c \abs{z_c}^{p - 1} \\
  & + c \abs{x_c}^{p} v_c
\end{split}
\end{equation}
with $k, A \in \mathbb{R}_{>0}$, $\alpha \in (0, 1)$, $\beta, c \in
\mathbb{R}_{\geq 0}$, $\gamma \in [-\beta, \beta]$, and $n, p \in
\mathbb{R}_{\geq 1}$. This is the Bouc-Wen-Simon-Hunt-Crossley Collision Law (BWSHCCL)
\cite{milehins_boucwen_2025, milehins_incremental_2026}. The first two terms of Eq. \eqref{eq:bwshc_phi},
weighted by $\alpha$ and $1 - \alpha$, are the forces borne by an elastic branch
depending upon the indentation and by a
hysteretic branch depending upon the internal state $z_c$, which has the
dimensions of a length and is governed by the model of hysteresis of Bouc and
Wen \cite{bouc_forced_1968, bouc_modemathematique_1971, wen_method_1976}. Taken
alone, these two terms constitute a law that is rate-independent. The third term is the rate-dependent term of
Simon, Hunt, and Crossley (see [Simon (1967), as cited in
\citenum{brogliato_nonsmooth_2016}] and Ref. \cite{hunt_coefficient_1975}).

The maps of Eqs. \eqref{eq:bwshc_psi} and \eqref{eq:bwshc_phi} are continuous,
and $\Phi(0, 0, v_c) = 0$ for all $v_c \in \mathbb{R}$, as Eq.
\eqref{eq:vanishing} requires. The contact interface subsystem for this law is obtained by
inserting the two maps into Eqs. \eqref{eq:Hc_f} to \eqref{eq:Hc_D}, the state
being of the dimension four. The jump map of Eq. \eqref{eq:Hc_g} and the initial
condition of Eq. \eqref{eq:Hc_init} are unaltered; and the same would be true of any
other admitted law.

With $\alpha = 1$, a value excluded above, the hysteretic branch bears no part
of the force, and the internal state, which then no longer influences it, is
redundant. Its omission, so that $n_z = 0$, yields the Simon-Hunt-Crossley
Collision Law (SHCCL), for which
\begin{equation}\label{eq:shc_phi}
\Phi(x_c, v_c) = k \, x_c \abs{x_c}^{p - 1} + c \abs{x_c}^{p} v_c
\end{equation}
with $k \in \mathbb{R}_{>0}$, $c \in \mathbb{R}_{\geq 0}$, and $p \in
\mathbb{R}_{\geq 1}$. The map is continuous and satisfies Eq.
\eqref{eq:vanishing}, and the state of the corresponding contact interface
subsystem is of the dimension three.
\end{example}

\begin{example}[The Bouc-Wen-Maxwell Collision Law]\label{ex:bwm}

A second instance exhibits an internal state of more than one component, and an
output that depends upon that state alone. Let $n_z = 2$, write $z_c =
(z_{c,1}, z_{c,2})$, and let
\begin{equation}\label{eq:bwm_psi1}
\Psi_1(x_c, z_c, v_c) = A \Psi_2 - \beta z_{c,1} \abs{z_{c,1}}^{n - 1}
  \abs{\Psi_2} - \gamma \abs{z_{c,1}}^{n} \Psi_2
\end{equation}
\begin{equation}\label{eq:bwm_psi2}
\Psi_2(x_c, z_c, v_c) = v_c - \frac{1}{c} \, \Phi(x_c, z_c, v_c)
\end{equation}
\begin{equation}\label{eq:bwm_phi}
\Phi(x_c, z_c, v_c) = \alpha k \, z_{c,2} \abs{z_{c,2}}^{p - 1}
  + (1 - \alpha) k \, z_{c,1} \abs{z_{c,1}}^{p - 1}
\end{equation}
in which $\Psi_2$ abbreviates $\Psi_2(x_c, z_c, v_c)$ throughout Eq.
\eqref{eq:bwm_psi1}, with $k, A, c \in \mathbb{R}_{>0}$, $\alpha \in (0, 1)$,
$\beta \in \mathbb{R}_{\geq 0}$, $\gamma \in [-\beta, \beta]$, and $n, p \in
\mathbb{R}_{\geq 1}$. This is the Bouc-Wen-Maxwell Collision Law (BWMCL) 
\cite{milehins_boucwen_2025, milehins_incremental_2026}.

The interface is a linear dashpot of coefficient $c$ disposed in series
with the elastic and hysteretic branches of Eq. \eqref{eq:bwshc_phi}. The
deflection of those branches is $z_{c,2}$, that of the dashpot is $x_c -
z_{c,2}$, and the two sum to the indentation. 

The maps are continuous, and $\Phi(0, 0, v_c) = 0$ for all $v_c \in
\mathbb{R}$, as Eq. \eqref{eq:vanishing} requires; the state of the contact
interface subsystem is of the dimension five. 
\end{example}

\subsection{The Interconnected System}\label{sec:interconnection}

\begin{figure}
\centering

\begin{tikzpicture}[>={Stealth[scale=0.7]}, line width=0.7]

    \draw[dashed, line width=0.6](-3.6,-1.8) rectangle (3.05,4.6);
    \node at (-3.15,-1.50) {$\mathcal{H}_f$};

    \draw (-1.4,1.2) rectangle (1.4,4.2);
    \node at (0,2.7) {$\mathcal{G}_b$};
    \draw (-1.4,-1.4) rectangle (1.4,-0.2);
    \node at (0,-0.8) {$\mathcal{H}_c$};

    \draw[->](-4.6,3.45) -- (-1.4,3.45);
    \node at (-4.10,3.62) {$F_{e,1}$};
    \draw[->](-4.6,2.70) -- (-1.4,2.70);
    \node at (-4.10,2.87) {$F_{e,2}$};

    \draw[->](1.4,3.45) -- (4.0,3.45);
    \node at (1.88,3.62) {$v_{m,1}$};
    \draw[->](1.4,2.70) -- (4.0,2.70);
    \node at (1.88,2.87) {$v_{m,2}$};

    \draw[->](1.4,1.95) -- (2.5,1.95) -- (2.5,-0.8) -- (1.4,-0.8);
    \node at (1.62,2.12) {$v$};

    \draw (-3.1,1.95) circle (0.18);
    \node at (-3.1,1.95) {$-$};

    \draw[->](-1.4,-0.8) -- (-3.1,-0.8) -- (-3.1,1.77);
    \node at (-1.64,-0.63) {$F_s$};
    \draw[->](-2.92,1.95) -- (-1.4,1.95);
    \node at (-1.72,2.12) {$F_c$};

\end{tikzpicture}

\caption{The interconnected system $\mathcal{H}_f$, formed of the rigid body
subsystem $\mathcal{G}_b$ and the contact interface subsystem $\mathcal{H}_c$.}\label{fig:interconnection}

\end{figure}

The rigid body subsystem $\mathcal{G}_b$ was given in Section~\ref{sec:bodies}
up to the contact force, and the contact interface $\mathcal{H}_c$ in
Section~\ref{sec:interface} up to the relative velocity. Coupling the two
closes the system, as shown informally in Fig.~\ref{fig:interconnection}: the
relative velocity is the input of $\mathcal{H}_c$, and the negative of the force
returned by $\mathcal{H}_c$, which is the contact force, is the input of
$\mathcal{G}_b$. The resulting system has the external forces upon the bodies as
its input and the velocities of their centers of mass as its output. The
coupling is unambiguous, since the contact force enters only the relative
acceleration, not the relative velocity. The relative displacement, a state of
both subsystems, is carried but once: its two copies are governed alike by
$\dot{x} = v$, are altered at neither transition, and agree initially.

Performed formally, this would require $\mathcal{G}_b$ to be recast as a hybrid
system with an empty jump set, the two systems to be interconnected in the
manner of Refs. \cite{sanfelice_results_2010, sanfelice_hybrid_2021}, and the duplicated relative displacement to be deleted
from the result. The apparatus that this requires is not developed here. The
system $\mathcal{H}_f$, of which the data are given directly below, is instead
taken as the object of study, and the reading of it as an interconnection is
nowhere relied upon: the properties of $\mathcal{H}_c$ and of $\mathcal{H}_f$
established in Section~\ref{sec:properties} are proved from the data of each.
Those of $\mathcal{H}_f$ that concern the interface alone are obtained from
those of $\mathcal{H}_c$, Lemma~\ref{lem:restriction} showing from the data of
the two systems that the components of a solution pair to $\mathcal{H}_f$ that
pertain to the interface constitute a solution pair to $\mathcal{H}_c$.

The data of $\mathcal{H}_f$ may now be described. The state of $\mathcal{H}_f$
is
\begin{equation}\label{eq:H_state}
\chi \triangleq (x, v, x_m, v_m, x_c, z_c, q_c) \in \mathbb{R}^4 \times \mathbb{R} \times \mathbb{R}^{n_z}
\times \mathbb{R}
\end{equation}
the input is $(F_{e,1}, F_{e,2}) \in \mathbb{R}^2$, the
output is $(v_{m,1}, v_{m,2}) \in \mathbb{R}^2$. 
The dimensions of $\mathcal{H}_f$ are therefore $n_s = n_z + 6$, $n_i = 2$, and
$n_o = 2$.

The maps of $\mathcal{H}_f$ are given by
\begin{equation}\label{eq:H_f}
\begin{aligned}
f_1(\chi, (F_{e,1}, F_{e,2})) &\triangleq v \\
f_2(\chi, (F_{e,1}, F_{e,2})) &\triangleq
  m^{-1} (-q_c \, \Phi(x_c, z_c, v) + F_{e,r}) \\
f_3(\chi, (F_{e,1}, F_{e,2})) &\triangleq v_m \\
f_4(\chi, (F_{e,1}, F_{e,2})) &\triangleq
  (m_1 + m_2)^{-1} (F_{e,1} + F_{e,2}) \\
f_5(\chi, (F_{e,1}, F_{e,2})) &\triangleq q_c v \\
f_{6:n_z+5}(\chi, (F_{e,1}, F_{e,2})) &\triangleq q_c \, \Psi(x_c, z_c, v) \\
f_{n_z+6}(\chi, (F_{e,1}, F_{e,2})) &\triangleq 0
\end{aligned}
\end{equation}
\begin{equation}\label{eq:H_g}
g(\chi, (F_{e,1}, F_{e,2})) \triangleq (x, v, x_m, v_m, 0, 0, 1 - q_c)
\end{equation}
\begin{equation}\label{eq:H_h}
\begin{aligned}
h_1(\chi, (F_{e,1}, F_{e,2})) &\triangleq v_m + \eta_c v \\
h_2(\chi, (F_{e,1}, F_{e,2})) &\triangleq v_m - \eta v
\end{aligned}
\end{equation}
in which $F_{e,r}$ is given by Eq.~\eqref{eq:u}. 

The flow set and the jump set are
\begin{equation}\label{eq:H_C}
\begin{split}
C \triangleq{} & \left( \{ q_c = 0 \} \cap \left( \{ x \geq 0 \} \cup \{ v \geq
0 \} \right) \right) \\
& \cup \left( \{ q_c = 1 \} \cap \left( \{ \Phi \leq 0 \} \cup \{ v \leq 0 \}
\right) \right)
\end{split}
\end{equation}
\begin{equation}\label{eq:H_D}
\begin{split}
D \triangleq{} & \left( \{ q_c = 0 \} \cap \{ x \leq 0 \} \cap \{ v \leq 0 \}
\right) \\
& \cup \left( \{ q_c = 1 \} \cap \{ \Phi \geq 0 \} \cap \{ v \geq 0 \} \right)
\end{split}
\end{equation}
these being subsets of $\left( \mathbb{R}^4 \times \mathbb{R} \times
\mathbb{R}^{n_z} \times \mathbb{R} \right) \times \mathbb{R}^2$. The initial
condition is
\begin{equation}\label{eq:H_init}
\chi(0, 0) \triangleq (x_0, v_0, x_{m,0}, v_{m,0}, 0, 0, 0)
\end{equation}
in which $x_0, v_0 \in \mathbb{R}$ are the parameters of Eq. \eqref{eq:main} and
$x_{m,0}, v_{m,0} \in \mathbb{R}$ are given by Eqs. \eqref{eq:x_m_0} and
\eqref{eq:v_m_0}, subject to the restrictions $x_0 \in \mathbb{R}_{\geq 0}$ and
$x_0 + l_1 + l_2 \in \mathbb{R}_{>0}$ imposed in Section~\ref{sec:bodies}. The
bodies are apart, or touching, at the initial instant, and the interface is
undeformed, whence $q_c(0, 0) = 0$.

The components $x_m$ and $v_m$ are absent from the flow set and from the jump
set, the jump map leaves them unaltered, and no other component of the flow map
or of the jump map involves them; the components of the flow map that govern
them, in turn, involve no other component of the state. They may accordingly be
deleted, together with the input and the output of the port to which they belong,
and the motion of the center of mass recovered afterwards from Eq.
\eqref{eq:com}.

The deletion is obligatory when one of the bodies is of infinite mass, the
center of mass of the system being then that of the heavier body. Suppose, by abuse of notation, $m_2 = +\infty$, so that $m = m_1$ and $F_{e,r} = F_{e,1}$, as remarked in Section \ref{sec:bodies}. Then, the symbol $\mathcal{H}_r$ shall denote the system that remains, of which the state is
\begin{equation}\label{eq:Hr_state}
\chi \triangleq (x, v, x_c, z_c, q_c) \in \mathbb{R}^2 \times \mathbb{R} \times \mathbb{R}^{n_z} \times \mathbb{R}
\end{equation}
the input is $F_{e,r} \in \mathbb{R}$, the output is $v \in \mathbb{R}$,
whence $n_s = n_z + 4$, $n_i = 1$, and $n_o = 1$. The symbols $\chi$, $f$, $g$,
$h$, $C$, and $D$ are employed for the data of $\mathcal{H}_r$ as they are for
those of $\mathcal{H}_f$. Its maps are
\begin{equation}\label{eq:Hr_f}
\begin{aligned}
f_1(\chi, F_{e,r}) &\triangleq v \\
f_2(\chi, F_{e,r}) &\triangleq
  m^{-1} (-q_c \, \Phi(x_c, z_c, v) + F_{e,r}) \\
f_3(\chi, F_{e,r}) &\triangleq q_c v \\
f_{4:n_z+3}(\chi, F_{e,r}) &\triangleq q_c \, \Psi(x_c, z_c, v) \\
f_{n_z+4}(\chi, F_{e,r}) &\triangleq 0
\end{aligned}
\end{equation}
\begin{equation}\label{eq:Hr_g}
g(\chi, F_{e,r}) \triangleq (x, v, 0, 0, 1 - q_c)
\end{equation}
\begin{equation}\label{eq:Hr_h}
h(\chi, F_{e,r}) \triangleq v
\end{equation}
its flow set and jump set are those of Eqs. \eqref{eq:H_C} and \eqref{eq:H_D},
read as subsets of $\left( \mathbb{R}^2 \times \mathbb{R} \times
\mathbb{R}^{n_z} \times \mathbb{R} \right) \times \mathbb{R}$, and its initial
condition is
\begin{equation}\label{eq:Hr_init}
\chi(0, 0) \triangleq (x_0, v_0, 0, 0, 0)
\end{equation}
in which $x_0, v_0 \in \mathbb{R}$ are the parameters of Eq. \eqref{eq:main}, 
subject to the restrictions $x_0 \in \mathbb{R}_{\geq 0}$ and
$x_0 + l_1 + l_2 \in \mathbb{R}_{>0}$ imposed in Section~\ref{sec:bodies}.

The flows of $\mathcal{H}_f$ and $\mathcal{H}_r$ are assumed to be unique: under every continuous applied force, any two solutions of the differential equation of the flow that issue from the same state at the same instant, and along which the state and the input remain in the flow set, agree upon the intersection of their intervals of existence. This does not follow from the assumption made of Eq.~\eqref{eq:law} in Section~\ref{sec:interface}, in which $v_c$ is prescribed, whereas the relative velocity that takes its place here is governed by the contact force. The assumption, imposed on physical grounds, is not invoked by the results of Section~\ref{sec:properties}. It confines the branching of the solutions to the states at which a transition is admitted, where the rules of Section~\ref{sec:nu} select among them.

\section{Properties}\label{sec:properties}

\subsection{Background}

\begin{lemma}[Hybrid basic conditions for $\mathcal{H}_c$]\label{lem:hbc_Hc}
The system $\mathcal{H}_c$ satisfies the hybrid basic conditions of
Definition~\ref{def:hbc}.
\end{lemma}

\begin{proof}
The maps $(\chi_c, v) \mapsto x$, $(\chi_c, v) \mapsto v$, $(\chi_c, v)
\mapsto q_c$, and $(\chi_c, v) \mapsto \Phi(x_c, z_c, v)$ are
continuous, the first three being projections and the last being the
composition of $\Phi$ with a projection. The sets $\{ q_c = 0 \}$ and $\{ q_c =
1 \}$ are the preimages of $\{ 0 \}$ and $\{ 1 \}$ under the third of these
maps; the sets $\{ x \geq 0 \}$ and $\{ x \leq 0 \}$ are the preimages of
$\mathbb{R}_{\geq 0}$ and $\mathbb{R}_{\leq 0}$ under the first; the sets $\{ v
\geq 0 \}$ and $\{ v \leq 0 \}$, of the same two sets under the second; and the
sets $\{ \Phi \geq 0 \}$ and $\{ \Phi \leq 0 \}$, of the same two sets under
the fourth. All eight are accordingly closed, and the sets $C_c$ and $D_c$ of
Eqs. \eqref{eq:Hc_C} and \eqref{eq:Hc_D} are finite unions of finite
intersections of them, whence condition \ref{a:closed}. The maps $f_c$, $g_c$,
and $h_c$ of Eqs. \eqref{eq:Hc_f} to \eqref{eq:Hc_h} are continuous, being
composed of $\Phi$, $\Psi$, and the projections onto the components of the
state and of the input by addition and multiplication, whence conditions \ref{a:flow},
\ref{a:jump}, and \ref{a:output}.
\end{proof}

\begin{lemma}[Hybrid basic conditions for $\mathcal{H}_f$]\label{lem:hbc_Hf}
The system $\mathcal{H}_f$ satisfies the hybrid basic conditions of
Definition~\ref{def:hbc}.
\end{lemma}

\begin{proof}
The sets $C$ and $D$ of Eqs. \eqref{eq:H_C} and \eqref{eq:H_D} are of the same
form as those of Eqs. \eqref{eq:Hc_C} and \eqref{eq:Hc_D}, and are closed by
the argument of the proof of Lemma~\ref{lem:hbc_Hc}, whence condition
\ref{a:closed}. The maps $f$, $g$, and $h$ of Eqs. \eqref{eq:H_f} to
\eqref{eq:H_h} are continuous, being composed of $\Phi$, $\Psi$, and the
projections onto the components of the state and of the input by addition and
multiplication, whence conditions \ref{a:flow}, \ref{a:jump}, and \ref{a:output}.
\end{proof}

\begin{lemma}[Hybrid basic conditions for $\mathcal{H}_r$]\label{lem:hbc_Hr}
The system $\mathcal{H}_r$ satisfies the hybrid basic conditions of
Definition~\ref{def:hbc}.
\end{lemma}

\begin{proof}
The flow set and the jump set of $\mathcal{H}_r$ are those of Eqs.
\eqref{eq:H_C} and \eqref{eq:H_D} and are closed by the argument of the proof
of Lemma~\ref{lem:hbc_Hc}, whence condition \ref{a:closed}. The maps $f$, $g$,
and $h$ of Eqs. \eqref{eq:Hr_f} to \eqref{eq:Hr_h} are continuous, being
composed of $\Phi$, $\Psi$, and the projections onto the components of the
state and of the input by addition and multiplication, whence conditions \ref{a:flow},
\ref{a:jump}, and \ref{a:output}.
\end{proof}

\begin{lemma}[Invariance of the regimes]\label{lem:invariant}
Let $\mathcal{K}$ be any of $\mathcal{H}_c$, $\mathcal{H}_f$, and
$\mathcal{H}_r$, with the state dimension $n_s$. The set
\begin{equation}\label{eq:S_regimes}
\mathcal{S} \triangleq \left\{ \chi \in \mathbb{R}^{n_s} :
  q_c \in \{ 0, 1 \} \right\}
\end{equation}
is invariant for $\mathcal{K}$ in the sense of
Definition~\ref{def:invariance}.
\end{lemma}

\begin{proof}
The component of the flow map that governs $q_c$ vanishes identically by Eqs.
\eqref{eq:Hc_f}, \eqref{eq:H_f}, and \eqref{eq:Hr_f}. If $z$ is as in condition
\ref{i:flow},
the component of $z$ that corresponds to $q_c$ accordingly has a vanishing
derivative upon $[a, b]$, and is constant; it belongs to $\{ 0, 1 \}$ at $a$,
and, hence, upon $[a, b]$, which is condition \ref{i:flow}. The component of the
jump map that corresponds to $q_c$ is $1 - q_c$ by Eqs. \eqref{eq:Hc_g},
\eqref{eq:H_g}, and \eqref{eq:Hr_g}, which belongs to $\{ 0, 1 \}$ whenever
$q_c$ does, the set $\{ 0, 1 \}$ being invariant under $q \mapsto 1 - q$; this
is condition \ref{i:jump}.
\end{proof}

The components of the state of $\mathcal{H}_f$ and of $\mathcal{H}_r$ that
pertain to the interface evolve along their solution pairs as along those of
$\mathcal{H}_c$, and results established for $\mathcal{H}_c$ may accordingly be
transferred to them.

\begin{lemma}[Restriction to the interface]\label{lem:restriction}
Let $\mathcal{K}$ be either of $\mathcal{H}_f$ and $\mathcal{H}_r$, let $(\chi,
u)$ be a solution pair to $\mathcal{K}$, and let $E \triangleq \dom \chi$. Let
$\chi_c \triangleq (x, x_c, z_c, q_c)$ and $v$ be the maps upon $E$ formed of
the corresponding components of $\chi$. Then, $(\chi_c, v)$ is a solution pair
to $\mathcal{H}_c$.
\end{lemma}

\begin{proof}
The map $\chi_c$ is a hybrid arc and $v$ is a hybrid signal, their components
being components of the hybrid arc $\chi$. The conditions that define the flow
set and the jump set of $\mathcal{K}$ in Eqs.~\eqref{eq:H_C} and
\eqref{eq:H_D} involve the components $x$, $v$, $x_c$, $z_c$, and $q_c$ of the
state alone, and coincide with those that define $C_c$ and $D_c$ in
Eqs.~\eqref{eq:Hc_C} and \eqref{eq:Hc_D}, the relative velocity being read as
the input of $\mathcal{H}_c$. For all $(t, j) \in E$, accordingly, $(\chi(t,
j), u(t, j))$ belongs to the flow set or to the jump set of $\mathcal{K}$ if and
only if $(\chi_c(t, j), v(t, j))$ belongs to $C_c$ or to $D_c$, respectively.
By Eqs.~\eqref{eq:H_f}, \eqref{eq:Hr_f}, and \eqref{eq:Hc_f}, the components
of the flow map of $\mathcal{K}$ that govern $x$, $x_c$, $z_c$, and $q_c$ are
those of $f_c(\chi_c, v)$, and, by Eqs.~\eqref{eq:H_g}, \eqref{eq:Hr_g}, and
\eqref{eq:Hc_g}, the components of the jump map of $\mathcal{K}$ that govern
them are those of $g_c(\chi_c, v)$. Conditions \ref{s:flow} and \ref{s:jump}
for $(\chi, u)$ accordingly give those for $(\chi_c, v)$. Finally, the flow set
of $\mathcal{K}$ being closed, by Lemma~\ref{lem:hbc_Hf} or \ref{lem:hbc_Hr},
$(\chi(0, 0), u(0, 0))$ belongs to it or to its jump set, whence $(\chi_c(0,
0), v(0, 0)) \in C_c \cup D_c$.
\end{proof}

\subsection{Compressive Laws}\label{sec:compressive}

\begin{definition}[Compressive contact law]\label{def:compressive_law}
A contact law of the form of Eq.~\eqref{eq:law} is \emph{compressive} if, for
all $b \in \mathbb{R}_{>0}$, all continuous $v_c : [0, b] \longrightarrow
\mathbb{R}$, and all solutions $(x_c, z_c) : [0, b] \longrightarrow \mathbb{R}
\times \mathbb{R}^{n_z}$ of Eq.~\eqref{eq:law} upon $[0, b]$ such that
$\Phi(x_c(s), z_c(s), v_c(s)) \leq 0$ for all $s \in (0, b]$ with $v_c(s) > 0$,
one has $\Phi(x_c(t), z_c(t), v_c(t)) \leq 0$ for all $t \in [0, b]$.
\end{definition}

The results that follow concern such laws: the first two, the law in
isolation; the remaining two, the systems of Sections~\ref{sec:interface} and
\ref{sec:interconnection}.

\begin{proposition}[Compressive laws without an internal state]%
\label{prop:compressive_nz}
Let $n_z = 0$, so that $\Phi$ is a map of $(x_c, v_c)$. If $\Phi(x_c, v_c) \leq
0$ for all $x_c, v_c \in \mathbb{R}_{\leq 0}$, and there exists $\varepsilon \in
\mathbb{R}_{>0}$ such that $\Phi(x_c, v_c) > 0$ for all $x_c \in (0,
\varepsilon)$ and $v_c \in \mathbb{R}_{>0}$, then the law is compressive.
\end{proposition}

\begin{proof}
Let $b$, $v_c$, and $x_c$ be as in
Definition~\ref{def:compressive_law}, satisfying the condition upon $(0, b]$
stated there. Suppose that $x_c(t) > 0$ for some $t \in [0, b]$, and let $a$ be
the greatest element of the set $\{ \tau \in [0, t] : x_c(\tau) \leq 0 \}$,
which is closed and contains $0$. Then, $a < t$, $x_c(a) = 0$, and $x_c$ is
positive upon $(a, t]$. By continuity, there exists $t' \in (a, t]$ such that
$x_c(\tau) < \varepsilon$ for all $\tau \in [a, t']$, and, by the mean value
theorem, there exists $s \in (a, t')$ such that $v_c(s) = \dot{x}_c(s) = x_c(t')
/ (t' - a) > 0$. Since $x_c(s) \in (0, \varepsilon)$, $\Phi(x_c(s), v_c(s)) >
0$, contrary to that condition, $s$ belonging to $(0, b]$. Hence, $x_c$ is
non-positive upon $[0, b]$. Let $t \in [0, b]$. If $v_c(t) \leq 0$, then
$\Phi(x_c(t), v_c(t)) \leq 0$ by hypothesis, $x_c(t)$ being non-positive; if
$v_c(t) > 0$ and $t > 0$, then $\Phi(x_c(t), v_c(t)) \leq 0$ by that condition;
and if $v_c(0) > 0$, then $\Phi(x_c(0), v_c(0)) = \Phi(0, v_c(0)) = 0$ by
Eq.~\eqref{eq:vanishing}. The law is accordingly compressive.
\end{proof}

\begin{example}[Compressiveness of the SHCCL]\label{ex:shc_compressive}
Let the contact law be the SHCCL of Example~\ref{ex:bwshc}. Both terms of
Eq.~\eqref{eq:shc_phi} are non-positive if $x_c, v_c \in \mathbb{R}_{\leq 0}$,
and $\Phi(x_c, v_c) \geq k \, x_c^{p} > 0$ if $x_c, v_c \in \mathbb{R}_{>0}$.
The SHCCL is therefore compressive, by Proposition~\ref{prop:compressive_nz}
with $\varepsilon \triangleq 1$.
\end{example}

The BWSHCCL and the BWMCL possess an internal state, and
Proposition~\ref{prop:compressive_nz} does not apply to them. Both are built
upon the model of hysteresis of Bouc and Wen, whose internal state obeys the
following bound under every input; the two examples that follow employ it to
establish that they are compressive.

\begin{lemma}[Bound for the model of Bouc and Wen]\label{lem:bouc_wen}
Let $A \in \mathbb{R}_{>0}$, $\beta \in \mathbb{R}_{\geq 0}$, $\gamma \in
[-\beta, \beta]$, and $n \in \mathbb{R}_{\geq 1}$, and let $\Xi : \mathbb{R}
\longrightarrow \mathbb{R}$ be given by
\[
\Xi(y) \triangleq A + \beta y \abs{y}^{n - 1} - \gamma \abs{y}^{n}
\]
Let $b \in \mathbb{R}_{>0}$, let $w : [0, b] \longrightarrow \mathbb{R}$ be
continuous, and let $\zeta : [0, b] \longrightarrow \mathbb{R}$ be a solution
of
\begin{equation}\label{eq:bouc_wen}
\dot{\zeta} = A w - \beta \zeta \abs{\zeta}^{n - 1} \abs{w} - \gamma
  \abs{\zeta}^{n} w
\end{equation}
upon $[0, b]$ with $\zeta(0) = 0$. Then, $\Xi(\zeta(t)) \geq 0$ for all $t \in
[0, b]$, and $\dot{\zeta}(t) \leq 0$ for all $t \in [0, b]$ with $w(t) \leq 0$.
\end{lemma}

\begin{proof}
For all $y \in \mathbb{R}_{\geq 0}$, $\Xi(y) = A + (\beta - \gamma) y^{n}$,
and, for all $y \in \mathbb{R}_{\leq 0}$, $\Xi(y) = A - (\beta + \gamma)
\abs{y}^{n}$. Since $\abs{\gamma} \leq \beta$, the map $\Xi$ is accordingly
non-decreasing, and $\Xi(y) \geq A$ for all $y \in \mathbb{R}_{\geq 0}$. By
Eq.~\eqref{eq:bouc_wen}, $\dot{\zeta} = \Xi(\zeta) \, w$ wherever $w \leq 0$,
and $\dot{\zeta} = \Xi(-\zeta) \, w$ wherever $w \geq 0$.

Suppose that $\Xi(\zeta(t)) < 0$ for some $t \in [0, b]$, and let $t'$ be the
greatest element of the set $\{ \theta \in [0, t] : \Xi(\zeta(\theta)) \geq 0
\}$, which is closed and contains $0$, $\Xi(\zeta(0))$ being $\Xi(0) = A$.
Then, $t' < t$ and $\Xi(\zeta) < 0$ upon $(t', t]$, whence $\zeta < 0$ and
$\Xi(-\zeta) \geq A$ there. Upon $(t', t]$, accordingly, $\dot{\zeta} \geq 0$,
both where $w \leq 0$ and where $w \geq 0$, and $\zeta$ is non-decreasing upon
$[t', t]$, as is $\Xi \circ \zeta$. Hence, $\Xi(\zeta(t)) \geq \Xi(\zeta(t'))
\geq 0$, a contradiction. Therefore, $\Xi \circ \zeta$ is non-negative upon
$[0, b]$, and $\dot{\zeta} = \Xi(\zeta) \, w \leq 0$ wherever $w \leq 0$.
\end{proof}

If $\beta + \gamma > 0$, then the first assertion of Lemma~\ref{lem:bouc_wen} states
that $\zeta \geq -(A / (\beta + \gamma))^{1 / n}$ upon $[0, b]$: it is the lower
half of the bound upon the internal state of the model of Bouc and Wen that is
familiar from the literature (e.g., see Ref. \cite{ikhouane_systems_2007}). The
second assertion, which is the one employed below, states that the internal
state does not increase while the input is non-positive.

\begin{example}[Compressiveness of the BWSHCCL]\label{ex:bwshc_compressive}
Let the contact law be the BWSHCCL of Example~\ref{ex:bwshc}, let $b$, $v_c$,
and $(x_c, z_c)$ be as in Definition~\ref{def:compressive_law}, satisfying
the condition upon $(0, b]$ stated there, and write $\Phi(t)$ for $\Phi(x_c(t),
z_c(t), v_c(t))$ for all $t \in [0, b]$. By Eqs.~\eqref{eq:law} and
\eqref{eq:bwshc_psi}, Eq.~\eqref{eq:bouc_wen} holds upon $[0, b]$ with $\zeta
\triangleq z_c$ and $w \triangleq v_c$; moreover, $z_c(0) = 0$. The parameters
of the law satisfying the hypotheses of Lemma~\ref{lem:bouc_wen}, that lemma
gives $\dot{z}_c \leq 0$ wherever $v_c \leq 0$.

Next, let $\Phi_0 : [0, b] \longrightarrow \mathbb{R}$ be given by
\[
\Phi_0(t) \triangleq \alpha k \, x_c(t) \abs{x_c(t)}^{p - 1} + (1 - \alpha) k
  \, z_c(t) \abs{z_c(t)}^{p - 1}
\]
the sum of the first two terms of Eq.~\eqref{eq:bwshc_phi}, so that 
\[
\Phi = \Phi_0 + c \abs{x_c}^{p} v_c
\]
upon $[0, b]$. Then, $\Phi_0(0) = 0$, and
$\Phi_0(t) \leq \Phi(t) \leq 0$ for all $t \in (0, b]$ with $v_c(t) > 0$, by
the condition, the third term of Eq.~\eqref{eq:bwshc_phi} being then
non-negative. Suppose that $\Phi_0(t) > 0$ for some $t \in [0, b]$, and let
$t_0$ be the greatest element of the set $\{ \theta \in [0, t] :
\Phi_0(\theta) \leq 0 \}$, which is closed and contains $0$. Then, $t_0 < t$
and $\Phi_0 > 0$ upon $(t_0, t]$, whence $v_c \leq 0$ there, and, by the
foregoing, $\dot{x}_c \leq 0$ and $\dot{z}_c \leq 0$. Hence, $x_c$ and $z_c$
are non-increasing upon $[t_0, t]$, as is $\Phi_0$, the map $y \mapsto y
\abs{y}^{p - 1}$ being increasing, and $\Phi_0(t) \leq \Phi_0(t_0) \leq 0$, a
contradiction. Therefore, $\Phi_0 \leq 0$ upon $[0, b]$, and $\Phi \leq \Phi_0
\leq 0$ wherever $v_c \leq 0$, the third term of Eq.~\eqref{eq:bwshc_phi} being
then non-positive. Elsewhere, $\Phi \leq 0$ by the condition, or, at $t = 0$,
by Eq.~\eqref{eq:vanishing}. The BWSHCCL is accordingly compressive.
\end{example}

\begin{example}[Compressiveness of the BWMCL]\label{ex:bwm_compressive}
Let the contact law be the BWMCL of Example~\ref{ex:bwm}, let $b$, $v_c$, and
$(x_c, z_c)$ be as in Definition~\ref{def:compressive_law}, satisfying the
condition upon $(0, b]$ stated there, and write $z_c = (z_{c,1}, z_{c,2})$.
Write $\Phi(t)$, $\Psi_1(t)$, and $\Psi_2(t)$ for the values of $\Phi$,
$\Psi_1$, and $\Psi_2$ at $(x_c(t), z_c(t), v_c(t))$ for all $t \in [0, b]$, so
that $\dot{z}_{c,1} = \Psi_1$ and $\dot{z}_{c,2} = \Psi_2$ upon $[0, b]$. By
Eq.~\eqref{eq:bwm_psi1}, Eq.~\eqref{eq:bouc_wen} holds upon $[0, b]$ with
$\zeta \triangleq z_{c,1}$ and with $w$ the map $t \mapsto \Psi_2(t)$, which is
continuous; moreover, $z_{c,1}(0) = 0$. The parameters of the law satisfying
the hypotheses of Lemma~\ref{lem:bouc_wen}, that lemma gives $\Psi_1 \leq 0$
wherever $\Psi_2 \leq 0$.

Next, suppose that $\Phi(t) > 0$ for some $t \in [0, b]$, and let $t_0$ be the
greatest element of the set $\{ \theta \in [0, t] : \Phi(\theta) \leq 0 \}$,
which is closed and contains $0$, $\Phi(0)$ vanishing by
Eq.~\eqref{eq:vanishing}. Then, $t_0 < t$ and $\Phi > 0$ upon $(t_0, t]$,
whence, there, $v_c \leq 0$ by the condition, $\Psi_2 = v_c - \Phi / c < 0$ by
Eq.~\eqref{eq:bwm_psi2}, and $\Psi_1 \leq 0$ by the foregoing. Hence, $z_{c,1}$
and $z_{c,2}$ are non-increasing upon $[t_0, t]$, as is $\Phi$, by
Eq.~\eqref{eq:bwm_phi}, the map $y \mapsto y \abs{y}^{p - 1}$ being increasing,
and $\Phi(t) \leq \Phi(t_0) \leq 0$, a contradiction. The BWMCL is accordingly
compressive.
\end{example}

\begin{lemma}[Compression during a contact for $\mathcal{H}_c$]%
\label{lem:contact_compression_Hc}
Let $(\chi_c, v)$ be a solution pair to $\mathcal{H}_c$ with $q_c(0, 0) = 0$,
let $E \triangleq \dom \chi_c$, and let $(t, j) \in E$ be such that $q_c(t, j) =
1$. If the contact law of Eq.~\eqref{eq:law} is compressive, then
$\Phi(x_c(\tau, j), z_c(\tau, j), v(\tau, j)) \leq 0$ for all $\tau \in I_E^j$.
\end{lemma}

\begin{proof}
By Lemmas~\ref{lem:invariance} and \ref{lem:invariant}, $q_c$ takes its values
in $\{ 0, 1 \}$ upon $E$, and $q_c(\cdot, j)$, being continuous, is accordingly
constant upon $I_E^j$, and equal to $1$. Let $s \triangleq \min I_E^j$. Then, $j
\geq 1$, $q_c$ being $0$ upon $I_E^0$, and $(s, j - 1) \in E$ by
Lemmas~\ref{lem:flow_interval} and \ref{lem:truncation}, whence $x_c(s, j) = 0$
and $z_c(s, j) = 0$ by condition \ref{s:jump} and Eq.~\eqref{eq:Hc_g}, and
$\Phi(x_c(s, j), z_c(s, j), v(s, j)) = 0$ by Eq.~\eqref{eq:vanishing}. The
assertion accordingly holds if $I_E^j = \{ s \}$.

Otherwise, $j \in \mathcal{J}_E$. Since $q_c = 1$ upon $I_E^j$, condition
\ref{s:flow} and Eq.~\eqref{eq:Hc_f} give $\dot{x}_c = v$ and $\dot{z}_c =
\Psi(x_c, z_c, v)$ upon $\interior I_E^j$, and, hence, upon $I_E^j$, the
derivatives of Definition~\ref{def:arc} and the map $v(\cdot, j)$ being
continuous; and condition \ref{s:flow} and Eq.~\eqref{eq:Hc_C} give
$\Phi(x_c(\tau, j), z_c(\tau, j), v(\tau, j)) \leq 0$ for all $\tau \in
\interior I_E^j$ with $v(\tau, j) > 0$. Let $\tau \in \interior I_E^j$ and $b
\triangleq \tau - s$. The map $\theta \mapsto v(s + \theta, j)$ is continuous
upon $[0, b]$, and the map $\theta \mapsto (x_c(s + \theta, j), z_c(s + \theta,
j))$ is a solution of Eq.~\eqref{eq:law} upon $[0, b]$ for it; since $s +
\theta \in \interior I_E^j$ for all $\theta \in (0, b]$, the motion so obtained
satisfies the condition upon $(0, b]$ stated in
Definition~\ref{def:compressive_law}. Hence, $\Phi(x_c(\tau, j), z_c(\tau, j),
v(\tau, j)) \leq 0$, the law being compressive. The inequality accordingly
holds upon $\interior I_E^j$, and, by continuity, upon $I_E^j$.
\end{proof}

\begin{lemma}[Compression during a contact]\label{lem:contact_compression}
Let $\mathcal{K}$ be either of $\mathcal{H}_f$ and $\mathcal{H}_r$, let $(\chi,
u)$ be a solution pair to $\mathcal{K}$ with $q_c(0, 0) = 0$, let $E \triangleq
\dom \chi$, and let $(t, j) \in E$ be such that $q_c(t, j) = 1$. If the contact
law of Eq.~\eqref{eq:law} is compressive, then $\Phi(x_c(\tau, j),
z_c(\tau, j), v(\tau, j)) \leq 0$ for all $\tau \in I_E^j$.
\end{lemma}

\begin{proof}
By Lemma~\ref{lem:restriction}, the pair $(\chi_c, v)$ formed of the components
$x$, $x_c$, $z_c$, $q_c$, and $v$ of $\chi$ is a solution pair to
$\mathcal{H}_c$ with $\dom \chi_c = E$ and $q_c(0, 0) = 0$. The assertion
accordingly follows by Lemma~\ref{lem:contact_compression_Hc}.
\end{proof}

\subsection{Continuity of the Contact Force}\label{sec:continuity}

Along a solution pair $(\chi_c, v)$ to $\mathcal{H}_c$, the contact force is
the negative of the output, the map
\[
F_c \triangleq -q_c \, \Phi(x_c, z_c, v)
\]
upon $\dom \chi_c$; along a solution pair $(\chi, u)$ to $\mathcal{H}_f$ or to
$\mathcal{H}_r$, it is the map given by the same expression upon $\dom \chi$.
In either case, it is a hybrid signal by the convention following
Definition~\ref{def:signal}, $\Phi$ being continuous. It is accordingly
continuous upon every interval of flow, and what remains to be settled is its
behavior at the transitions, which is settled for $\mathcal{H}_c$ first.

\begin{proposition}[The contact force at a transition]\label{prop:force_jump}
Let $(\chi_c, v)$ be a solution pair to $\mathcal{H}_c$ with $q_c(0, 0) = 0$,
let $E \triangleq \dom \chi_c$, and let $(t, j) \in E$ be such that $(t, j +
1) \in E$. Then, $F_c(t, j + 1) = 0$ and $F_c(t, j) \leq 0$; moreover, $F_c(t,
j) = 0$ unless $q_c(t, j) = 1$ and $v(t, j) = 0$.
\end{proposition}

\begin{proof}
By Lemmas~\ref{lem:invariance} and \ref{lem:invariant}, $q_c$ takes its values
in $\{ 0, 1 \}$ upon $E$, and $q_c(\cdot, k)$, being continuous, is accordingly
constant upon $I_E^k$ for all $k \in \mathbb{Z}_{\geq 0}$ with $I_E^k \neq
\emptyset$. By condition \ref{s:jump} and Eq. \eqref{eq:Hc_g}, $(\chi_c(t, j),
v(t, j)) \in D_c$, $x_c(t, j + 1) = 0$, $z_c(t, j + 1) = 0$, and $q_c(t, j + 1)
= 1 - q_c(t, j)$. If $q_c(t, j) = 0$, then $F_c(t, j) = 0$, and
\[
F_c(t, j + 1) = -\Phi(0, 0, v(t, j + 1)) = 0
\]
by Eq. \eqref{eq:vanishing}.
Let $q_c(t, j) = 1$. Then, $F_c(t, j + 1) = 0$, $q_c(t, j + 1)$ being $0$, and
Eq. \eqref{eq:Hc_D} gives $\Phi(x_c(t, j), z_c(t, j), v(t, j)) \geq 0$ and $v(t,
j) \geq 0$, whence $F_c(t, j) \leq 0$. Suppose that $v(t, j) > 0$, and let $s
\triangleq \min I_E^j$. If $s < t$, then, by continuity, $v(\cdot, j)$ is
positive upon $[\tau_0, t]$ for some $\tau_0 \in (s, t)$, and condition
\ref{s:flow} and Eq. \eqref{eq:Hc_C} give $\Phi(x_c(\tau, j), z_c(\tau, j),
v(\tau, j)) \leq 0$ for all $\tau \in (\tau_0, t)$, $q_c$ being $1$; by
continuity, the same holds at $t$, whence $F_c(t, j) = 0$. If $s = t$, then $j
\geq 1$, $q_c$ being $0$ upon $I_E^0$, and $(t, j - 1) \in E$ by
Lemmas~\ref{lem:flow_interval} and \ref{lem:truncation}, whence $x_c(t, j) = 0$
and $z_c(t, j) = 0$ by the foregoing, and 
\[
F_c(t, j) = -\Phi(0, 0, v(t, j)) = 0
\]
by Eq. \eqref{eq:vanishing}.
\end{proof}

\begin{corollary}[Continuity of the contact force for $\mathcal{H}_c$]%
\label{cor:force_continuous_Hc}
Let the contact law of Eq.~\eqref{eq:law} be compressive, let $(\chi_c,
v)$ and $E$ be as in Proposition~\ref{prop:force_jump}, and let
\[
T_E \triangleq \{ t \in \mathbb{R}_{\geq 0} : \exists j \in \mathbb{Z}_{\geq 0},
  (t, j) \in E \}
\]
Then, there exists a continuous map $\hat{F}_c : T_E \longrightarrow \mathbb{R}$
such that $F_c(t, j) = \hat{F}_c(t)$ for all $(t, j) \in E$.
\end{corollary}

\begin{proof}
Let $(t, j) \in E$ be such that $(t, j + 1) \in E$. By
Proposition~\ref{prop:force_jump}, $F_c(t, j + 1) = 0$ and $F_c(t, j) \leq 0$.
If $q_c(t, j) = 1$, then $F_c(t, j) \geq 0$ by
Lemma~\ref{lem:contact_compression_Hc}; otherwise, $F_c(t, j) = 0$, $q_c(t, j)$
being $0$ by Lemmas~\ref{lem:invariance} and \ref{lem:invariant}. In either
case, $F_c(t, j + 1) = F_c(t, j)$. If $(t, j), (t, j') \in
E$ and $j < j'$, then $(t, k) \in E$ for all $k \in \{ j, \dots, j' \}$, by
Lemmas~\ref{lem:flow_interval} and \ref{lem:truncation}, whence $F_c(t, j) =
F_c(t, j')$; the map $\hat{F}_c$ given by $\hat{F}_c(t) \triangleq F_c(t, j)$
for any $(t, j) \in E$ is accordingly well defined. Let $t \in T_E$, and let
$j_1$ be the least $j$ such that $(t, j) \in E$. If $t > 0$, then $t > \min
I_E^{j_1}$, as, otherwise, $(t, j_1 - 1)$ would belong to $E$; $\hat{F}_c$
coincides with $F_c(\cdot, j_1)$ upon $[\min I_E^{j_1}, t]$, and is accordingly
continuous from the left at $t$. If a greatest $j_2$ such that $(t, j) \in E$
exists and $t < \sup I_E^{j_2}$, then $\hat{F}_c$ coincides with $F_c(\cdot,
j_2)$ upon $[t, t'] \subseteq I_E^{j_2}$ for some $t' > t$, and is continuous
from the right at $t$; otherwise, no element of $T_E$ exceeds $t$. Hence,
$\hat{F}_c$ is continuous.
\end{proof}

\begin{corollary}[Continuity of the contact force]\label{cor:force_continuous}
Let the contact law of Eq.~\eqref{eq:law} be compressive, let
$\mathcal{K}$ be either of $\mathcal{H}_f$ and $\mathcal{H}_r$, let $(\chi, u)$
be a solution pair to $\mathcal{K}$ with $q_c(0, 0) = 0$, let $E \triangleq
\dom \chi$, and let $T_E$ be as in Corollary~\ref{cor:force_continuous_Hc}.
Then, there exists a continuous map $\hat{F}_c : T_E \longrightarrow \mathbb{R}$
such that $F_c(t, j) = \hat{F}_c(t)$ for all $(t, j) \in E$.
\end{corollary}

\begin{proof}
By Lemma~\ref{lem:restriction}, the pair $(\chi_c, v)$ formed of the components
$x$, $x_c$, $z_c$, $q_c$, and $v$ of $\chi$ is a solution pair to
$\mathcal{H}_c$ with $\dom \chi_c = E$ and $q_c(0, 0) = 0$, and the contact
force along it is that along $(\chi, u)$. The assertion accordingly follows by
Corollary~\ref{cor:force_continuous_Hc}.
\end{proof}

\subsection{Passivity}

\begin{definition}[Passive contact law]\label{def:passive_law}
A contact law of the form of Eq. \eqref{eq:law} is \emph{passive} if there
exists a continuously differentiable map $V_c : \mathbb{R} \times
\mathbb{R}^{n_z} \longrightarrow \mathbb{R}_{\geq 0}$ with $V_c(0, 0) = 0$ such
that, for all $(x_c, z_c, v_c) \in \mathbb{R} \times \mathbb{R}^{n_z} \times
\mathbb{R}$,
\begin{equation}\label{eq:law_passive}
\begin{split}
\left\langle \nabla V_c(x_c, z_c),
  \left( v_c, \Psi(x_c, z_c, v_c) \right) \right\rangle \leq \Phi(x_c, z_c, v_c) \, v_c
\end{split}
\end{equation}
in which $\mathbb{R} \times \mathbb{R}^{n_z}$ is identified with
$\mathbb{R}^{n_z + 1}$ as in Appendix~\ref{sec:NCF}, so that $\nabla V_c(x_c,
z_c)$ and $(v_c, \Psi(x_c, z_c, v_c))$ are elements of $\mathbb{R}^{n_z + 1}$.
The map $V_c$ is a \emph{storage function} for the law, and the map $d_c :
\mathbb{R} \times \mathbb{R}^{n_z} \times \mathbb{R} \longrightarrow
\mathbb{R}_{\geq 0}$ given by the difference of the two sides of Eq.
\eqref{eq:law_passive} is the \emph{dissipation rate} of the law.
\end{definition}

The dissipativity of the models of hysteresis upon which the laws of
Examples~\ref{ex:bwshc} and \ref{ex:bwm} are built has been investigated before
(e.g., see Refs. \cite{jayawardhana_sufficient_2009, jayawardhana_dissipativity_2011, jayawardhana_stability_2012, ouyang_characterization_2013} for the model of Duhem, of which that of Bouc and
Wen is an instance, and Refs. \cite{erlicher_thermodynamic_2004,
ikhouane_systems_2007} for the model of Bouc and Wen). The two examples that
follow establish the passivity of these laws in the sense of
Definition~\ref{def:passive_law}, and are included for completeness.

\begin{example}[Passivity of the BWSHCCL]\label{ex:bwshc_passive}
Let the contact law be that of Example~\ref{ex:bwshc}, and let $V_c :
\mathbb{R} \times \mathbb{R} \longrightarrow \mathbb{R}_{\geq 0}$ be given by
\begin{equation}\label{eq:bwshc_V}
V_c(x_c, z_c) \triangleq \frac{k}{p + 1} \left( \alpha \abs{x_c}^{p + 1}
  + \frac{1 - \alpha}{A} \abs{z_c}^{p + 1} \right)
\end{equation}
The map $V_c$ is non-negative and vanishes at $(0, 0)$. It is continuously
differentiable and its gradient is
\[
\nabla V_c(x_c, z_c) = \left( \alpha k \, x_c \abs{x_c}^{p - 1}, \;
  \frac{(1 - \alpha) k}{A} \, z_c \abs{z_c}^{p - 1} \right)
\]
Let $(x_c, z_c, v_c) \in \mathbb{R}^3$, and
write $\Psi$ and $\Phi$ for $\Psi(x_c, z_c, v_c)$ and $\Phi(x_c, z_c, v_c)$. By
the foregoing gradient and Eq. \eqref{eq:bwshc_psi}, the left-hand side of Eq.
\eqref{eq:law_passive} is
\[
\begin{split}
& \left\langle \nabla V_c(x_c, z_c), (v_c, \Psi) \right\rangle \\
& = \alpha k \, x_c \abs{x_c}^{p - 1} v_c + \frac{(1 - \alpha) k}{A} \, z_c
  \abs{z_c}^{p - 1} \Psi \\
& = \alpha k \, x_c \abs{x_c}^{p - 1} v_c + (1 - \alpha) k \, z_c
  \abs{z_c}^{p - 1} v_c \\
& \quad - \frac{(1 - \alpha) k}{A} (\beta \abs{z_c}^{p + n} \abs{v_c}
  + \gamma z_c \abs{z_c}^{p + n - 1} v_c)
\end{split}
\]
By Eq. \eqref{eq:bwshc_phi}, its right-hand side is
\[
\Phi v_c = \alpha k \, x_c \abs{x_c}^{p - 1} v_c + (1 - \alpha) k \, z_c
  \abs{z_c}^{p - 1} v_c  + c \abs{x_c}^{p} v_c^2
\]
Subtracting the left-hand side from the right-hand side yields the dissipation
rate
\[
\begin{split}
& d_c(x_c, z_c, v_c) \\
& = c \abs{x_c}^{p} v_c^2 + \frac{(1 - \alpha) k}{A} (\beta \abs{z_c}^{p + n} \abs{v_c}
  + \gamma z_c \abs{z_c}^{p + n - 1} v_c)
\end{split}
\]
Its first term is non-negative, $c$ being so. Its second term is so likewise:
the coefficient $(1 - \alpha) k / A$ is non-negative, and, the modulus of $z_c
\abs{z_c}^{p + n - 1} v_c$ being $\abs{z_c}^{p + n} \abs{v_c}$ and
$\abs{\gamma}$ being at most $\beta$, the bracketed factor is at least $(\beta -
\abs{\gamma}) \abs{z_c}^{p + n} \abs{v_c}$, which is non-negative. The law is
therefore passive with the storage function $V_c$, the two terms of $d_c$ being
the rates of dissipation by the rate-dependent term of Simon, Hunt, and Crossley
and by the hysteresis.

Nothing in the foregoing requires $\alpha < 1$; with $\alpha = 1$, the
hysteretic branch contributes neither to the storage nor to the dissipation.
For the SHCCL of Example~\ref{ex:bwshc}, accordingly, let $V_c : \mathbb{R}
\longrightarrow \mathbb{R}_{\geq 0}$ be given by $V_c(x_c) \triangleq k
\abs{x_c}^{p + 1} / (p + 1)$. Its gradient is $k \, x_c \abs{x_c}^{p - 1}$, the
left-hand side of Eq. \eqref{eq:law_passive} is $k \, x_c \abs{x_c}^{p - 1}
v_c$, its right-hand side is $k \, x_c \abs{x_c}^{p - 1} v_c + c \abs{x_c}^{p}
v_c^2$ by Eq. \eqref{eq:shc_phi}, and the dissipation rate is $c \abs{x_c}^{p}
v_c^2$, which is non-negative: the SHCCL is passive.
\end{example}

\begin{example}[Passivity of the BWMCL]\label{ex:bwm_passive}
Let the contact law be that of Example~\ref{ex:bwm}, and let $V_c : \mathbb{R}
\times \mathbb{R}^2 \longrightarrow \mathbb{R}_{\geq 0}$ be given by
\begin{equation}\label{eq:bwm_V}
V_c(x_c, z_c) \triangleq \frac{k}{p + 1} \left( \frac{1 - \alpha}{A}
  \abs{z_{c,1}}^{p + 1} + \alpha \abs{z_{c,2}}^{p + 1} \right)
\end{equation}
which does not depend upon $x_c$, the dashpot storing no energy. The map $V_c$
is non-negative and vanishes at $(0, 0)$. It is continuously differentiable and its gradient is
\[
\begin{split}
& \nabla V_c(x_c, z_c) \\
& = \left( 0, \; \frac{(1 - \alpha) k}{A} \, z_{c,1} \abs{z_{c,1}}^{p - 1},
  \; \alpha k \, z_{c,2} \abs{z_{c,2}}^{p - 1} \right)
\end{split}
\]
Let $(x_c, z_c, v_c) \in \mathbb{R} \times \mathbb{R}^2 \times \mathbb{R}$, and
write $\Psi_1$, $\Psi_2$, and $\Phi$ for their values at $(x_c, z_c, v_c)$. By
the foregoing gradient and Eq. \eqref{eq:bwm_psi1}, the left-hand side of Eq.
\eqref{eq:law_passive} is
\[
\begin{split}
& \left\langle \nabla V_c(x_c, z_c), (v_c, \Psi_1, \Psi_2) \right\rangle \\
& = \frac{(1 - \alpha) k}{A} \, z_{c,1} \abs{z_{c,1}}^{p - 1} \Psi_1
  + \alpha k \, z_{c,2} \abs{z_{c,2}}^{p - 1} \Psi_2 \\
& = (1 - \alpha) k \, z_{c,1} \abs{z_{c,1}}^{p - 1} \Psi_2
  + \alpha k \, z_{c,2} \abs{z_{c,2}}^{p - 1} \Psi_2 \\
& \quad - \frac{(1 - \alpha) k}{A} (\beta \abs{z_{c,1}}^{p + n} \abs{\Psi_2}
  + \gamma z_{c,1} \abs{z_{c,1}}^{p + n - 1} \Psi_2) \\
& = \Phi v_c - \frac{\Phi^2}{c} \\
& \quad - \frac{(1 - \alpha) k}{A} (\beta \abs{z_{c,1}}^{p + n} \abs{\Psi_2}
  + \gamma z_{c,1} \abs{z_{c,1}}^{p + n - 1} \Psi_2)
\end{split}
\]
The last equality holds since the sum of the
first two terms of the expression before it is $\Phi \Psi_2$ by Eq.
\eqref{eq:bwm_phi}, and $\Phi \Psi_2 = \Phi v_c - \Phi^2 / c$ by Eq.
\eqref{eq:bwm_psi2}. The right-hand side of Eq. \eqref{eq:law_passive} is $\Phi v_c$. 
Subtracting the left-hand side from $\Phi v_c$ yields the dissipation rate
\[
\begin{split}
& d_c(x_c, z_c, v_c) \\
& = \frac{\Phi^2}{c} + \frac{(1 - \alpha) k}{A} (\beta \abs{z_{c,1}}^{p + n} \abs{\Psi_2}
  + \gamma z_{c,1} \abs{z_{c,1}}^{p + n - 1} \Psi_2)
\end{split}
\]
Its first term is non-negative, $c$ being positive, and its second is so by the
argument of Example~\ref{ex:bwshc_passive} with $z_{c,1}$ and $\Psi_2$ in place
of $z_c$ and $v_c$. The law is therefore passive with the storage function $V_c$, the
two terms of $d_c$ being the rates of dissipation by the dashpot and by the
hysteresis.
\end{example}

\begin{proposition}[Passivity of the contact interface]\label{prop:pass_Hc}
Let the contact law of Eq. \eqref{eq:law} be passive with the storage function
$V_c$ and the dissipation rate $d_c$. Then, $\mathcal{H}_c$ is passive with the
storage function
\begin{equation}\label{eq:V_Hc}
V_{\mathcal{H}_c}(\chi_c) \triangleq V_c(x_c, z_c)
\end{equation}
\end{proposition}

\begin{proof}
Let $(\chi_c, v) \in C_c$. By Eq. \eqref{eq:V_Hc}, $\nabla
V_{\mathcal{H}_c}(\chi_c) = (0, \nabla V_c(x_c, z_c), 0)$, whence, by Eq.
\eqref{eq:Hc_f},
\[
\begin{split}
& \left\langle \nabla V_{\mathcal{H}_c}(\chi_c), f_c(\chi_c, v)
  \right\rangle \\
& \qquad = q_c \left\langle \nabla V_c(x_c, z_c),
  \left( v, \Psi(x_c, z_c, v) \right) \right\rangle \\
& \qquad = q_c \, \Phi(x_c, z_c, v) \, v - q_c \, d_c(x_c, z_c, v)
\end{split}
\]
the second equality holding by Definition~\ref{def:passive_law}. Since $q_c \geq 0$
and $d_c \geq 0$, and since $q_c \Phi(x_c, z_c, v) = h_c(\chi_c, v)$ by Eq.
\eqref{eq:Hc_h}, the right-hand side of the second equality is bounded above
by $v \, h_c(\chi_c, v)$, which is condition \ref{d:flow}.

Let $(\chi_c, v) \in D_c$. By Eqs. \eqref{eq:Hc_g} and \eqref{eq:V_Hc},
\[
V_{\mathcal{H}_c}\left( g_c(\chi_c, v) \right) = V_c(0, 0) = 0 \leq
  V_{\mathcal{H}_c}(\chi_c)
\]
which is condition \ref{d:jump}. The map $V_{\mathcal{H}_c}$ is continuously
differentiable and non-negative, $V_c$ being so.
\end{proof}

\begin{lemma}[Relative motion and interface]\label{lem:pass_rel}
Let the contact law of Eq. \eqref{eq:law} be passive with the storage function
$V_c$ and the dissipation rate $d_c$. Then, for all $m \in \mathbb{R}_{> 0}$, $v, x_c, F \in \mathbb{R}$, all $z_c \in \mathbb{R}^{n_z}$, and all $q_c \in \mathbb{R}$,
\begin{equation}\label{eq:pass_rel}
\begin{split}
& m v \cdot m^{-1} \left( -q_c \, \Phi(x_c, z_c, v) + F \right) \\
& \qquad + q_c \left\langle \nabla V_c(x_c, z_c),
  \left( v, \Psi(x_c, z_c, v) \right) \right\rangle \\
& \qquad = F v - q_c \, d_c(x_c, z_c, v)
\end{split}
\end{equation}
\end{lemma}

\begin{proof}
By Definition~\ref{def:passive_law}, the second line of Eq.
\eqref{eq:pass_rel} equals $q_c \Phi(x_c, z_c, v) v - q_c d_c(x_c, z_c, v)$,
which cancels the term in $\Phi$ of the first line.
\end{proof}

\begin{proposition}[Passivity of $\mathcal{H}_r$]\label{prop:pass_Hr}
Let the contact law of Eq. \eqref{eq:law} be passive with the storage function
$V_c$ and the dissipation rate $d_c$. Then, $\mathcal{H}_r$ is passive with the
storage function
\begin{equation}\label{eq:V_Hr}
V_{\mathcal{H}_r}(\chi) \triangleq \tfrac{1}{2} m v^2 + V_c(x_c, z_c)
\end{equation}
\end{proposition}

\begin{proof}
Let $(\chi, F_{e,r}) \in C$. By Eq. \eqref{eq:V_Hr}, $\nabla
V_{\mathcal{H}_r}(\chi) = (0, m v, \nabla V_c(x_c, z_c), 0)$. By this and Eq.
\eqref{eq:Hr_f}, the left-hand side of Eq. \eqref{eq:pass_rel} with $F =
F_{e,r}$ is $\left\langle \nabla V_{\mathcal{H}_r}(\chi), f(\chi, F_{e,r})
\right\rangle$, whence, by Lemma~\ref{lem:pass_rel},
\[
\left\langle \nabla V_{\mathcal{H}_r}(\chi), f(\chi, F_{e,r}) \right\rangle
  = F_{e,r} \, v - q_c \, d_c(x_c, z_c, v)
\]
Since $q_c \geq 0$ and $d_c \geq 0$, and since $h(\chi, F_{e,r}) = v$ by Eq.
\eqref{eq:Hr_h}, the right-hand side of this equality is bounded above by
$F_{e,r} \, h(\chi, F_{e,r})$, which is condition \ref{d:flow}.

Let $(\chi, F_{e,r}) \in D$. By Eqs. \eqref{eq:Hr_g} and \eqref{eq:V_Hr},
\[
V_{\mathcal{H}_r}\left( g(\chi, F_{e,r}) \right) = \tfrac{1}{2} m v^2 +
  V_c(0, 0) \leq V_{\mathcal{H}_r}(\chi)
\]
by $V_c(0, 0) = 0$ and $V_c \geq 0$, which is condition \ref{d:jump}. The map
$V_{\mathcal{H}_r}$ is continuously differentiable and non-negative, $V_c$
being so.
\end{proof}

\begin{proposition}[Passivity of $\mathcal{H}_f$]\label{prop:pass_Hf}
Let the contact law of Eq. \eqref{eq:law} be passive with the storage function
$V_c$ and the dissipation rate $d_c$. Then, $\mathcal{H}_f$ is passive with the
storage function
\begin{equation}\label{eq:V_H}
V_{\mathcal{H}_f}(\chi) \triangleq \tfrac{1}{2} m v^2 +
  \tfrac{1}{2} (m_1 + m_2) v_m^2 + V_c(x_c, z_c)
\end{equation}
\end{proposition}

\begin{proof}
Let $(\chi, (F_{e,1}, F_{e,2})) \in C$. By Eqs. \eqref{eq:H_f} and
\eqref{eq:V_H}, the quantity $\left\langle \nabla V_{\mathcal{H}_f}(\chi),
f(\chi, (F_{e,1}, F_{e,2})) \right\rangle$ is the sum of the left-hand side of
Eq. \eqref{eq:pass_rel} with $F = F_{e,r}$ and of the term
\[
(m_1 + m_2) v_m \cdot (m_1 + m_2)^{-1} (F_{e,1} + F_{e,2})
\]
the gradient of $V_{\mathcal{H}_f}$ being given by
\[
\nabla V_{\mathcal{H}_f}(\chi) = \left( 0, m v, 0, (m_1 + m_2) v_m,
  \nabla V_c(x_c, z_c), 0 \right)
\]
whence, by Lemma~\ref{lem:pass_rel},
\[
\begin{split}
& \left\langle \nabla V_{\mathcal{H}_f}(\chi),
  f(\chi, (F_{e,1}, F_{e,2})) \right\rangle \\
& \qquad = F_{e,r} v + (F_{e,1} + F_{e,2}) v_m - q_c \, d_c(x_c, z_c, v)
\end{split}
\]
By Eqs. \eqref{eq:u} and \eqref{eq:H_h},
\[
\begin{split}
& F_{e,r} v + (F_{e,1} + F_{e,2}) v_m \\
& \qquad = \left( \eta_c F_{e,1} - \eta F_{e,2} \right) v
  + (F_{e,1} + F_{e,2}) v_m \\
& \qquad = F_{e,1} \left( v_m + \eta_c v \right) +
  F_{e,2} \left( v_m - \eta v \right) \\
& \qquad = F_{e,1} h_1(\chi, (F_{e,1}, F_{e,2})) + F_{e,2} h_2(\chi, (F_{e,1}, F_{e,2}))
\end{split}
\]
Since $q_c \geq 0$ and $d_c \geq 0$, the two preceding identities give
condition \ref{d:flow}.

Let $(\chi, (F_{e,1}, F_{e,2})) \in D$. By Eqs. \eqref{eq:H_g} and
\eqref{eq:V_H},
\[
\begin{split}
& V_{\mathcal{H}_f}\left( g(\chi, (F_{e,1}, F_{e,2})) \right) \\
& \qquad = \tfrac{1}{2} m v^2 + \tfrac{1}{2} (m_1 + m_2) v_m^2 + V_c(0, 0) \\
& \qquad \leq V_{\mathcal{H}_f}(\chi)
\end{split}
\]
by $V_c(0, 0) = 0$ and $V_c \geq 0$, which is condition \ref{d:jump}. The map
$V_{\mathcal{H}_f}$ is continuously differentiable and non-negative, $V_c$
being so.
\end{proof}

\subsection{Confinement}

\begin{definition}[Confining contact law]\label{def:confining_law}
A contact law of the form of Eq. \eqref{eq:law} is \emph{confining} if it is
passive in the sense of Definition~\ref{def:passive_law} with a storage
function $V_c$ such that, for every compact $K \subset \mathbb{R}$ and every
$\varrho \in \mathbb{R}_{\geq 0}$, the set
\begin{equation}\label{eq:confining}
\left\{ z_c \in \mathbb{R}^{n_z} : \exists x_c \in K, \;
  V_c(x_c, z_c) \leq \varrho \right\}
\end{equation}
is bounded. Such a $V_c$ is a \emph{confining storage function} for the law.
\end{definition}

\begin{lemma}[Confinement of the internal state]\label{lem:confinement}
Let the contact law of Eq. \eqref{eq:law} be confining with the confining
storage function $V_c$, and let $S \subseteq \mathbb{R} \times
\mathbb{R}^{n_z}$ be such that the sets
\begin{equation}\label{eq:confinement_hyp}
\left\{ x_c \in \mathbb{R} : \exists z_c \in \mathbb{R}^{n_z}, \;
  (x_c, z_c) \in S \right\}
\end{equation}
\begin{equation}\label{eq:confinement_hyp_V}
\left\{ V_c(x_c, z_c) : (x_c, z_c) \in S \right\}
\end{equation}
are bounded. Then, the set
\begin{equation}\label{eq:confinement_set}
\left\{ z_c \in \mathbb{R}^{n_z} : \exists x_c \in \mathbb{R}, \;
  (x_c, z_c) \in S \right\}
\end{equation}
is bounded.
\end{lemma}

\begin{proof}
Let $\kappa \in \mathbb{R}_{\geq 0}$ bound the set of Eq.
\eqref{eq:confinement_hyp} and let $\varrho \in \mathbb{R}_{\geq 0}$ bound that
of Eq. \eqref{eq:confinement_hyp_V}. If $(x_c, z_c) \in S$, then $x_c \in
[-\kappa, \kappa]$ and $V_c(x_c, z_c) \leq \varrho$, whence $z_c$ belongs to
the set of Eq. \eqref{eq:confining} with $K \triangleq [-\kappa, \kappa]$,
which is bounded by Definition~\ref{def:confining_law}.
\end{proof}

\begin{proposition}[Sufficient condition for confinement]%
\label{prop:radial_confining}
Let the contact law of Eq. \eqref{eq:law} be passive with the storage function
$V_c$, let
\begin{equation}\label{eq:radial_index}
P(K, r) \triangleq \left\{ (x_c, z_c) \in \mathbb{R} \times \mathbb{R}^{n_z} :
  x_c \in K, \; \abs{z_c} \geq r \right\}
\end{equation}
for $K \subseteq \mathbb{R}$ and $r \in \mathbb{R}_{\geq 0}$, and suppose that,
for every compact $K \subset \mathbb{R}$,
\begin{equation}\label{eq:radial}
\lim_{r \rightarrow +\infty} \; \inf V_c\left[ P(K, r) \right] = +\infty
\end{equation}
Then, the law is
confining, with the confining storage function $V_c$.
\end{proposition}

\begin{proof}
Let $K \subset \mathbb{R}$ be compact, let $\varrho \in \mathbb{R}_{\geq 0}$,
and let $Z$ denote the set of Eq. \eqref{eq:confining}. By Eq.
\eqref{eq:radial}, there exists $r \in \mathbb{R}_{>0}$ such that $\inf
V_c\left[ P(K, r) \right] > \varrho$. Let $z_c \in Z$ and let $x_c \in K$ be
such that $V_c(x_c, z_c) \leq \varrho$. Were $\abs{z_c} \geq r$, the pair
$(x_c, z_c)$ would belong to $P(K, r)$, whence $V_c(x_c, z_c) > \varrho$ by the
choice of $r$, a contradiction. Therefore, $\abs{z_c} < r$, and $Z$ is bounded.
\end{proof}

\begin{proposition}[Uniform sufficient condition for confinement]%
\label{prop:uniform_confining}
Let the contact law of Eq. \eqref{eq:law} be passive with the storage function
$V_c$, and suppose that there exists $\alpha \in \mathcal{K}_\infty$ such that,
for all $(x_c, z_c) \in \mathbb{R} \times \mathbb{R}^{n_z}$,
\begin{equation}\label{eq:radial_uniform}
V_c(x_c, z_c) \geq \alpha\left( \abs{z_c} \right)
\end{equation}
Then, the law is confining, with the confining storage function $V_c$.
\end{proposition}

\begin{proof}
Let $K \subset \mathbb{R}$ be compact and let $r \in \mathbb{R}_{\geq 0}$. By
Eqs. \eqref{eq:radial_index} and \eqref{eq:radial_uniform} and the
monotonicity of $\alpha$, $\inf V_c\left[ P(K, r) \right] \geq \alpha(r)$.
Since $\alpha(r) \rightarrow +\infty$ as $r \rightarrow +\infty$, Eq.
\eqref{eq:radial} holds, and the assertion follows by
Proposition~\ref{prop:radial_confining}.
\end{proof}

\begin{example}[Confinement for the BWSHCCL]\label{ex:bwshc_confining}
Let the contact law be that of Example~\ref{ex:bwshc}, let $V_c$ be the
storage function of Eq. \eqref{eq:bwshc_V}, and let $\theta : \mathbb{R}_{\geq
0} \longrightarrow \mathbb{R}_{\geq 0}$ be given by
$\theta(r) \triangleq {(1 - \alpha) k \, r^{p + 1} / (A (p + 1))}$. By Eq.
\eqref{eq:bwshc_V}, whose first term is non-negative, $V_c(x_c, z_c) \geq
\theta(\abs{z_c})$ for all $(x_c, z_c) \in \mathbb{R}^2$. Since
$\alpha < 1$, the map $\theta$ is a positive multiple of the map
$r \mapsto r^{p + 1}$ upon $\mathbb{R}_{\geq 0}$, and belongs accordingly to
$\mathcal{K}_\infty$. The law is therefore
confining, with the confining storage function $V_c$, by
Proposition~\ref{prop:uniform_confining}.

The restriction $\alpha < 1$, which Example~\ref{ex:bwshc_passive} did not
require, is essential here: with $\alpha = 1$, the storage function would not
depend upon $z_c$ and would not confine it. The SHCCL of
Example~\ref{ex:bwshc}, whose internal state is absent, is nevertheless
confining, with the storage function given for it in
Example~\ref{ex:bwshc_passive}, the set of Eq. \eqref{eq:confining} being then
a subset of $\mathbb{R}^0$, which consists of the empty tuple alone and is
accordingly bounded.
\end{example}

\begin{example}[Confinement for the BWMCL]\label{ex:bwm_confining}
Let the contact law be that of Example~\ref{ex:bwm}, let $V_c$ be the storage
function of Eq. \eqref{eq:bwm_V}, let $\varepsilon \triangleq {\min \{ (1 -
\alpha) / A, \alpha \}}$, and let $\theta : \mathbb{R}_{\geq 0}
\longrightarrow \mathbb{R}_{\geq 0}$ be given by
$\theta(r) \triangleq {\varepsilon k (r / \sqrt{2})^{p + 1} / (p + 1)}$. Let
$(x_c, z_c) \in \mathbb{R} \times \mathbb{R}^2$. The coefficients of
$\abs{z_{c,1}}^{p + 1}$ and $\abs{z_{c,2}}^{p + 1}$ in Eq. \eqref{eq:bwm_V}
being at least $\varepsilon k / (p + 1)$, and $\abs{z_c}$ being at most
$\sqrt{2} \max \{ \abs{z_{c,1}}, \abs{z_{c,2}} \}$,
\[
\begin{split}
V_c(x_c, z_c) & \geq \frac{\varepsilon k}{p + 1} \left( \abs{z_{c,1}}^{p + 1}
  + \abs{z_{c,2}}^{p + 1} \right) \\
& \geq \frac{\varepsilon k}{p + 1} \max \left\{ \abs{z_{c,1}},
  \abs{z_{c,2}} \right\}^{p + 1} \\
& \geq \frac{\varepsilon k}{p + 1} \left( \frac{\abs{z_c}}{\sqrt{2}}
  \right)^{p + 1} = \theta(\abs{z_c})
\end{split}
\]
Since $\varepsilon > 0$, $\alpha$ belonging to $(0, 1)$, the map $\theta$ is a
positive multiple of the map $r \mapsto r^{p + 1}$ upon $\mathbb{R}_{\geq 0}$,
and belongs accordingly to $\mathcal{K}_\infty$. The law is therefore
confining, with the confining storage function $V_c$, by
Proposition~\ref{prop:uniform_confining}.
\end{example}

Lemma~\ref{lem:confinement} bounds the internal state upon a set of
configurations. It is applied below to the configurations attained by a
solution pair up to a given ordinary time, and thereby transferred to the time
domain. In each of the two propositions that follow, $\mathcal{K}$ denotes the
system named, $\chi$ its state, $V$ the storage function named alongside it,
$n_s$ its dimension, and $(\chi, u)$ a solution pair to it; $x_c$ and $z_c$
are the indentation and the internal state of the contact law, which are
$\chi_3$ and $\chi_{4:n_z+3}$ for $\mathcal{H}_r$ and $\chi_5$ and
$\chi_{6:n_z+5}$ for $\mathcal{H}_f$, as recorded after Eqs.
\eqref{eq:Hr_state} and \eqref{eq:H_state}, respectively. In each
case, there exists a map $R : \mathbb{R}^{n_s} \longrightarrow \mathbb{R}_{\geq
0}$ such that
\begin{equation}\label{eq:confinement_split}
V(\chi) = V_c(x_c, z_c) + R(\chi)
\end{equation}
for all $\chi \in \mathbb{R}^{n_s}$, namely $R(\chi) \triangleq \tfrac{1}{2} m
v^2$ for $\mathcal{H}_r$ by Eq. \eqref{eq:V_Hr} and $R(\chi) \triangleq
\tfrac{1}{2} m v^2 + \tfrac{1}{2} (m_1 + m_2) v_m^2$ for $\mathcal{H}_f$ by Eq.
\eqref{eq:V_H}.

The two propositions are proved alike. Let $T \in \mathbb{R}_{\geq 0}$ be as in
the statement and let
\[
S \triangleq \left\{ \left( x_c(t, j), z_c(t, j) \right) :
  (t, j) \in \dom \chi, \; t \leq T \right\}
\]
The set of Eq. \eqref{eq:confinement_hyp} is then the first of the two sets
supposed bounded in the statement. By Eq. \eqref{eq:confinement_split} and the
non-negativity of $V_c$ and of $R$,
\[
0 \leq V_c\left( x_c(t, j), z_c(t, j) \right) \leq V\left( \chi(t, j) \right)
\]
for all $(t, j) \in \dom \chi$ with $t \leq T$, whence the set of Eq.
\eqref{eq:confinement_hyp_V} is bounded, the second of the two sets supposed
bounded in the statement being so. Lemma~\ref{lem:confinement} applies, and the
set of Eq. \eqref{eq:confinement_set} is the set asserted to be bounded in the
statement.

\begin{proposition}[Confinement for $\mathcal{H}_r$]\label{prop:confinement_Hr}
Let the contact law of Eq. \eqref{eq:law} be confining with the confining
storage function $V_c$, let $V_{\mathcal{H}_r}$ be given by Eq.
\eqref{eq:V_Hr}, let $(\chi, F_{e,r})$ be a solution pair to $\mathcal{H}_r$,
let $T \in \mathbb{R}_{\geq 0}$, and let the sets
\begin{equation}\label{eq:confinement_Hr_x}
\left\{ x_c(t, j) : (t, j) \in \dom \chi, \; t \leq T \right\}
\end{equation}
\begin{equation}\label{eq:confinement_Hr_V}
\left\{ V_{\mathcal{H}_r}\left( \chi(t, j) \right) :
  (t, j) \in \dom \chi, \; t \leq T \right\}
\end{equation}
be bounded. Then, the set
\begin{equation}\label{eq:confinement_Hr_z}
\left\{ z_c(t, j) : (t, j) \in \dom \chi, \; t \leq T \right\}
\end{equation}
is bounded.
\end{proposition}

\begin{proposition}[Confinement for $\mathcal{H}_f$]\label{prop:confinement_Hf}
Let the contact law of Eq. \eqref{eq:law} be confining with the confining
storage function $V_c$, let $V_{\mathcal{H}_f}$ be given by Eq. \eqref{eq:V_H},
let $(\chi, u)$ be a solution pair to $\mathcal{H}_f$, let $T \in
\mathbb{R}_{\geq 0}$, and let the sets
\begin{equation}\label{eq:confinement_Hf_x}
\left\{ x_c(t, j) : (t, j) \in \dom \chi, \; t \leq T \right\}
\end{equation}
\begin{equation}\label{eq:confinement_Hf_V}
\left\{ V_{\mathcal{H}_f}\left( \chi(t, j) \right) :
  (t, j) \in \dom \chi, \; t \leq T \right\}
\end{equation}
be bounded. Then, the set
\begin{equation}\label{eq:confinement_Hf_z}
\left\{ z_c(t, j) : (t, j) \in \dom \chi, \; t \leq T \right\}
\end{equation}
is bounded.
\end{proposition}

\subsection{Existence}

\begin{lemma}[Covering]\label{lem:covering}
Let $\mathcal{K}$ be either of $\mathcal{H}_f$ and $\mathcal{H}_r$, with the
flow set $C$, the jump set $D$, and the dimensions $n_s$ and $n_i$, and let
$\mathcal{S}$ be given by Eq. \eqref{eq:S_regimes}. Then,
\begin{equation}\label{eq:covering_H}
C \cup D = \mathcal{S} \times \mathbb{R}^{n_i}
\end{equation}
\end{lemma}

\begin{proof}
Every element of $C \cup D$ has its component $q_c$ in $\{ 0, 1 \}$ by Eqs.
\eqref{eq:H_C} and \eqref{eq:H_D}, whence $C \cup D \subseteq \mathcal{S}
\times \mathbb{R}^{n_i}$. Conversely, let $(\chi, u) \in \mathcal{S} \times
\mathbb{R}^{n_i}$ and suppose that $(\chi, u) \notin C$. If $q_c = 0$, then $x
< 0$ and $v < 0$ by Eq. \eqref{eq:H_C}, whence $(\chi, u) \in D$ by Eq.
\eqref{eq:H_D}; and if $q_c = 1$, then $\Phi(x_c, z_c, v) > 0$ and $v > 0$,
whence likewise. The two cases are exhaustive, $\chi$ belonging to
$\mathcal{S}$, whence $\mathcal{S} \times \mathbb{R}^{n_i} \subseteq C \cup D$.
\end{proof}

\begin{lemma}[Output bound for $\mathcal{H}_f$]\label{lem:output_Hf}
Let the contact law of Eq. \eqref{eq:law} be passive with the storage function
$V_c$, let $V_{\mathcal{H}_f}$ be given by Eq. \eqref{eq:V_H}, and let
\begin{equation}\label{eq:gamma_H}
\gamma \triangleq \sqrt{2 / \min \{ m_1, m_2 \}}
\end{equation}
Then, for all $(\chi, u) \in \mathbb{R}^{n_s} \times \mathbb{R}^{n_i}$, in
which $n_s$ and $n_i$ are the dimensions of $\mathcal{H}_f$,
\begin{equation}\label{eq:output_H}
\abs{h(\chi, u)} \leq \gamma \sqrt{V_{\mathcal{H}_f}(\chi)}
\end{equation}
\end{lemma}

\begin{proof}
By Eq. \eqref{eq:H_h} and the definitions of $m$, $\eta$, and $\eta_c$,
\[
\begin{split}
& \tfrac{1}{2} m_1 h_1(\chi, u)^2 + \tfrac{1}{2} m_2 h_2(\chi, u)^2 \\
& \qquad = \tfrac{1}{2} m v^2 + \tfrac{1}{2} (m_1 + m_2) v_m^2
\end{split}
\]
The left-hand side of this identity is bounded below by $\tfrac{1}{2} \min \{ m_1, m_2 \}
\abs{h(\chi, u)}^2$, and the right-hand side is bounded above by
$V_{\mathcal{H}_f}(\chi)$ by Eq. \eqref{eq:V_H}, the map $V_c$ being
non-negative. The proof follows.
\end{proof}

\begin{lemma}[Output bound for $\mathcal{H}_r$]\label{lem:output_Hr}
Let the contact law of Eq. \eqref{eq:law} be passive with the storage function
$V_c$, let $V_{\mathcal{H}_r}$ be given by Eq. \eqref{eq:V_Hr}, let $n_s$ and
$n_i$ be the dimensions of $\mathcal{H}_r$, and let
\begin{equation}\label{eq:gamma_Hr}
\gamma_r \triangleq \sqrt{2 / m}
\end{equation}
Then, for all $(\chi, F_{e,r}) \in \mathbb{R}^{n_s} \times \mathbb{R}^{n_i}$,
\begin{equation}\label{eq:output_Hr}
\abs{h(\chi, F_{e,r})} \leq \gamma_r \sqrt{V_{\mathcal{H}_r}(\chi)}
\end{equation}
\end{lemma}

\begin{proof}
By Eq. \eqref{eq:Hr_h}, $h(\chi, F_{e,r}) = v$, and, by Eq. \eqref{eq:V_Hr}, the
quantity $\tfrac{1}{2} m v^2$ is bounded above by $V_{\mathcal{H}_r}(\chi)$,
the map $V_c$ being non-negative.
\end{proof}

\begin{lemma}[Bound upon an integrated component]\label{lem:integrated}
Let $\mathcal{K} = (C, f, D, g, h)$ be a hybrid system with inputs and outputs
with the dimensions $n_s$, $n_i$, and $n_o$, let $\iota_1, \iota_2 \in \{ 1,
\dots, n_s \}$, and let $\rho : \mathbb{R}^{n_s} \times \mathbb{R}^{n_i}
\longrightarrow \mathbb{R}$ be such that, for all $(\chi, u) \in
\mathbb{R}^{n_s} \times \mathbb{R}^{n_i}$,
\begin{equation}\label{eq:integrated_f}
f_{\iota_1}(\chi, u) = \rho(\chi, u) \, \chi_{\iota_2}
\end{equation}
\begin{equation}\label{eq:integrated_g}
\abs{g_{\iota_1}(\chi, u)} \leq \abs{\chi_{\iota_1}}
\end{equation}
Let $(\chi, u)$ be a solution pair to $\mathcal{K}$, let $T \in
\mathbb{R}_{\geq 0}$, and let $\kappa \in \mathbb{R}_{\geq 0}$ be such that,
for all $(t, j) \in \dom \chi$ with $t \leq T$,
\begin{equation}\label{eq:integrated_hyp_rho}
\abs{\rho\left( \chi(t, j), u(t, j) \right)} \leq 1
\end{equation}
\begin{equation}\label{eq:integrated_hyp_k}
\abs{\chi_{\iota_2}(t, j)} \leq \kappa
\end{equation}
Then,
\begin{equation}\label{eq:integrated_bound}
\abs{\chi_{\iota_1}(t, j)} \leq \abs{\chi_{\iota_1}(0, 0)} + \kappa T
\end{equation}
for all $(t, j) \in \dom \chi$ with $t \leq T$.
\end{lemma}

\begin{proof}
Let $(T', J) \in \dom \chi$ with $T' \leq T$ and let $\{ t_j \}_{j = 0}^{J +
1}$ generate $E_{(T', J)}$. For $j \in \{ 0, 1, \dots, J \}$ and $t \in
\interior I_{\dom \chi}^{j}$ with $t \leq T$, condition \ref{s:flow} and Eqs.
\eqref{eq:integrated_f}, \eqref{eq:integrated_hyp_rho} and
\eqref{eq:integrated_hyp_k} give $\abs{\dot{\chi}_{\iota_1}(t, j)} \leq
\kappa$; the map $\dot{\chi}_{\iota_1}(\cdot, j)$ being continuous upon
$I_{\dom \chi}^{j}$ by Definition~\ref{def:arc}, the same bound holds upon
$[t_j, t_{j + 1}]$, whence
\[
\abs{\chi_{\iota_1}(t_{j + 1}, j)} \leq \abs{\chi_{\iota_1}(t_j, j)}
  + \kappa \left( t_{j + 1} - t_j \right)
\]
By condition \ref{s:jump} and Eq. \eqref{eq:integrated_g},
\[
\abs{\chi_{\iota_1}(t_{j + 1}, j + 1)} \leq
  \abs{\chi_{\iota_1}(t_{j + 1}, j)}
\]
for $j \in \{ 0, 1, \dots, J - 1 \}$. Induction upon $j$ by means of these two
inequalities gives
$\abs{\chi_{\iota_1}(T', J)} \leq \abs{\chi_{\iota_1}(0, 0)} + \kappa T'$,
whence Eq. \eqref{eq:integrated_bound}.
\end{proof}

\begin{theorem}[Completeness of $\mathcal{H}_f$]\label{thm:completeness}
Let the contact law of Eq. \eqref{eq:law} be confining with the confining
storage function $V_c$, let $V_{\mathcal{H}_f}$ be given by Eq. \eqref{eq:V_H},
let $n_s$ and $n_i$ be the dimensions of $\mathcal{H}_f$, let $w :
\mathbb{R}_{\geq 0} \longrightarrow \mathbb{R}^2$ be continuous and bounded,
and let $(\chi, u)$ be a solution pair to $\mathcal{H}_f$ that is maximal under
$w$ in the sense of Definition~\ref{def:maximal}. Then, $(\chi, u)$ is complete.
\end{theorem}

\begin{proof}
Let $\mu \in \mathbb{R}_{\geq 0}$ be such that $\abs{w(t)} \leq \mu$ for all $t
\in \mathbb{R}_{\geq 0}$, let $T \in \mathbb{R}_{\geq 0}$, and let $\mathcal{S}$
be given by Eq. \eqref{eq:S_regimes}. Throughout the proof, $(t, j)$ denotes an
arbitrary element of $\dom \chi$ with $t \leq T$, every assertion made of it
holding accordingly for all such elements. By Definition~\ref{def:solution} and
Lemma~\ref{lem:hbc_Hf}, $(\chi(0, 0), u(0, 0)) \in C \cup D$, whence $\chi(0,
0) \in \mathcal{S}$ by Lemma~\ref{lem:covering}, and $q_c(t, j) \in \{ 0, 1
\}$ by Lemmas~\ref{lem:invariance} and \ref{lem:invariant}.

By Proposition~\ref{prop:pass_Hf} and Definition~\ref{def:passivity}, the
system $\mathcal{H}_f$ is dissipative with respect to the supply rate $s(u, y)
= \left\langle u, y \right\rangle$ with the storage function
$V_{\mathcal{H}_f}$, and, by Lemma~\ref{lem:output_Hf} and the Cauchy--Schwarz
inequality, $\left\langle u, h(\chi, u) \right\rangle \leq \gamma \abs{u}
\sqrt{V_{\mathcal{H}_f}(\chi)}$ for all $(\chi, u) \in \mathbb{R}^{n_s} \times
\mathbb{R}^{n_i}$, with $\gamma$ given by Eq. \eqref{eq:gamma_H}, which is the
hypothesis of Proposition~\ref{prop:growth}. That proposition and
Proposition~\ref{prop:induced_integral} give
\[
\sqrt{V_{\mathcal{H}_f}\left( \chi(t, j) \right)} \leq
  \sqrt{V_{\mathcal{H}_f}\left( \chi(0, 0) \right)} + \gamma \mu T
\]
Let
\[
\begin{gathered}
\varrho_T \triangleq \left( \sqrt{V_{\mathcal{H}_f}\left( \chi(0, 0) \right)}
  + \gamma \mu T \right)^2 \\
\kappa_v \triangleq \sqrt{2 \varrho_T / m} \qquad
  \kappa_m \triangleq \sqrt{2 \varrho_T / (m_1 + m_2)}
\end{gathered}
\]
so that $V_{\mathcal{H}_f}\left( \chi(t, j) \right) \leq \varrho_T$. By Eq.
\eqref{eq:V_H}, the quantity $\tfrac{1}{2} m v(t, j)^2$ is bounded above
by $V_{\mathcal{H}_f}\left( \chi(t, j) \right)$, the remaining summands
$\tfrac{1}{2} (m_1 + m_2) v_m(t, j)^2$ and $V_c\left( x_c(t, j), z_c(t, j)
\right)$ being non-negative, whence $\abs{v(t, j)} \leq \kappa_v$, and, likewise,
$\tfrac{1}{2} (m_1 + m_2) v_m(t, j)^2$ is bounded above by
$V_{\mathcal{H}_f}\left( \chi(t, j) \right)$, whence $\abs{v_m(t, j)} \leq
\kappa_m$.

As recorded after Eq. \eqref{eq:H_state}, the components $x$, $v$, $x_m$, $v_m$,
and $x_c$ of the state are $\chi_1$, $\chi_2$, $\chi_3$, $\chi_4$, and $\chi_5$,
respectively, while $z_c = \chi_{6:n_z+5}$ and $q_c = \chi_{n_s}$, $n_s$ being
$n_z + 6$. Lemma~\ref{lem:integrated} is applied three times, the applications being
indexed by $\ell \in \{ 1, 2, 3 \}$. For each such $\ell$, let $\iota_1^\ell,
\iota_2^\ell \in \{ 1, \dots, n_s \}$ be given by $\left( \iota_1^1, \iota_2^1
\right) \triangleq (1, 2)$, $\left( \iota_1^2, \iota_2^2 \right) \triangleq (3,
4)$, and $\left( \iota_1^3, \iota_2^3 \right) \triangleq (5, 2)$, and let
$\rho^\ell : \mathbb{R}^{n_s} \times \mathbb{R}^{n_i} \longrightarrow
\mathbb{R}$ be given by $\rho^1(\chi, u) \triangleq 1$, $\rho^2(\chi, u)
\triangleq 1$, and $\rho^3(\chi, u) \triangleq \chi_{n_s}$. In the $\ell$th
application, the indices $\iota_1^\ell$ and $\iota_2^\ell$ and
the map $\rho^\ell$ are those of Eqs. \eqref{eq:integrated_f} to
\eqref{eq:integrated_hyp_k}.

By Eq. \eqref{eq:H_f}, the hypothesis of Eq. \eqref{eq:integrated_f} holds for
each $\ell \in \{ 1, 2, 3 \}$, since $f_1(\chi, u) = v$, $f_3(\chi, u) = v_m$,
and $f_5(\chi, u) = q_c v$ for all $(\chi, u) \in \mathbb{R}^{n_s} \times
\mathbb{R}^{n_i}$. By Eq. \eqref{eq:H_g}, the hypothesis of Eq.
\eqref{eq:integrated_g} holds for each $\ell \in \{ 1, 2, 3 \}$, since
$g_1(\chi, u) = x$, $g_3(\chi, u) = x_m$, and $g_5(\chi, u) = 0$ for all
$(\chi, u) \in \mathbb{R}^{n_s} \times \mathbb{R}^{n_i}$, with equality
for $\ell \in \{ 1, 2 \}$ and by the non-negativity of $\abs{x_c}$ for $\ell =
3$.

The hypothesis of Eq. \eqref{eq:integrated_hyp_rho} holds for $\ell \in \{ 1, 2
\}$, $\rho^1$ and $\rho^2$ being identically $1$, and for $\ell = 3$ by $q_c(t,
j) \in \{ 0, 1 \}$, $\rho^3(\chi, u)$ being the component $\chi_{n_s}$ of the
state, which is $q_c$. The hypothesis of Eq.
\eqref{eq:integrated_hyp_k} holds with $\kappa \triangleq \kappa_v$ for $\ell
\in \{ 1, 3 \}$ and with $\kappa \triangleq \kappa_m$ for $\ell = 2$, by the
bounds upon $\abs{v(t, j)}$ and $\abs{v_m(t, j)}$ obtained above, the component
$\iota_2^\ell$ being $v$ for $\ell \in \{ 1, 3 \}$ and $v_m$ for $\ell = 2$.

Lemma~\ref{lem:integrated} accordingly gives $\abs{x(t, j)} \leq \abs{x(0, 0)}
+ \kappa_v T$, $\abs{x_m(t, j)} \leq \abs{x_m(0, 0)} + \kappa_m T$, and
$\abs{x_c(t, j)} \leq \abs{x_c(0, 0)} + \kappa_v T$.

The sets of Eqs. \eqref{eq:confinement_Hf_x} and \eqref{eq:confinement_Hf_V}
are bounded, by the bound upon $\abs{x_c(t, j)}$ and by $V_{\mathcal{H}_f}\left(
\chi(t, j) \right) \leq \varrho_T$, respectively, whence the set of Eq.
\eqref{eq:confinement_Hf_z} is bounded by
Proposition~\ref{prop:confinement_Hf}. The components $\chi_1(t, j)$ to
$\chi_5(t, j)$ are therefore bounded independently of $(t, j)$, by the bounds
upon $\abs{x(t, j)}$, $\abs{v(t, j)}$, $\abs{x_m(t, j)}$, $\abs{v_m(t, j)}$,
and $\abs{x_c(t, j)}$ obtained above, respectively; so is $\chi_{6:n_z+5}(t, j) = z_c(t,
j)$, by the boundedness of the set of Eq. \eqref{eq:confinement_Hf_z}; and so is
$\chi_{n_s}(t, j) = q_c(t, j)$, by $q_c(t, j) \in \{ 0, 1 \}$. By the
identification of Appendix~\ref{sec:NCF}, the set
\[
\left\{ \chi(t, j) : (t, j) \in \dom \chi, \; t \leq T \right\}
\]
is accordingly bounded.

The system $\mathcal{H}_f$ satisfies the hybrid basic conditions by
Lemma~\ref{lem:hbc_Hf}, and the set $\mathcal{S}$ is invariant for
$\mathcal{H}_f$ by Lemma~\ref{lem:invariant} and satisfies Eq.
\eqref{eq:covering_S} by Lemma~\ref{lem:covering}. Since $T$ is arbitrary,
Proposition~\ref{prop:completeness} applies.
\end{proof}

\begin{theorem}[Completeness of $\mathcal{H}_r$]\label{thm:completeness_Hr}
Let the contact law of Eq. \eqref{eq:law} be confining with the confining
storage function $V_c$, let $V_{\mathcal{H}_r}$ be given by Eq.
\eqref{eq:V_Hr}, let $n_s$ and $n_i$ be the dimensions of $\mathcal{H}_r$, let
$w : \mathbb{R}_{\geq 0} \longrightarrow \mathbb{R}$ be continuous and bounded,
and let $(\chi, u)$ be a solution pair to $\mathcal{H}_r$ that is maximal under
$w$ in the sense of Definition~\ref{def:maximal}. Then, $(\chi, u)$ is complete.
\end{theorem}

\begin{proof}
Let $\mu \in \mathbb{R}_{\geq 0}$ be such that $\abs{w(t)} \leq \mu$ for all $t
\in \mathbb{R}_{\geq 0}$, let $T \in \mathbb{R}_{\geq 0}$, and let $\mathcal{S}$
be given by Eq. \eqref{eq:S_regimes}. Throughout the proof, $(t, j)$ denotes an
arbitrary element of $\dom \chi$ with $t \leq T$, every assertion made of it
holding accordingly for all such elements. By Definition~\ref{def:solution} and
Lemma~\ref{lem:hbc_Hr}, $(\chi(0, 0), u(0, 0)) \in C \cup D$, whence $\chi(0,
0) \in \mathcal{S}$ by Lemma~\ref{lem:covering}, and $q_c(t, j) \in \{ 0, 1
\}$ by Lemmas~\ref{lem:invariance} and \ref{lem:invariant}.

By Proposition~\ref{prop:pass_Hr} and Definition~\ref{def:passivity}, the
system $\mathcal{H}_r$ is dissipative with respect to the supply rate $s(u, y)
= \left\langle u, y \right\rangle$ with the storage function
$V_{\mathcal{H}_r}$, and, by Lemma~\ref{lem:output_Hr} and the Cauchy--Schwarz
inequality, $\left\langle u, h(\chi, u) \right\rangle \leq \gamma_r \abs{u}
\sqrt{V_{\mathcal{H}_r}(\chi)}$ for all $(\chi, u) \in \mathbb{R}^{n_s} \times
\mathbb{R}^{n_i}$, with $\gamma_r$ given by Eq. \eqref{eq:gamma_Hr}, which is
the hypothesis of Proposition~\ref{prop:growth}. That proposition and
Proposition~\ref{prop:induced_integral} give
\[
\sqrt{V_{\mathcal{H}_r}\left( \chi(t, j) \right)} \leq
  \sqrt{V_{\mathcal{H}_r}\left( \chi(0, 0) \right)} + \gamma_r \mu T
\]
Let
\[
\varrho_T \triangleq \left( \sqrt{V_{\mathcal{H}_r}\left( \chi(0, 0) \right)}
  + \gamma_r \mu T \right)^2 \qquad \kappa_v \triangleq \sqrt{2 \varrho_T / m}
\]
so that $V_{\mathcal{H}_r}\left( \chi(t, j) \right) \leq \varrho_T$. By Eq.
\eqref{eq:V_Hr}, the quantity
$\tfrac{1}{2} m v(t, j)^2$ is bounded
above by $V_{\mathcal{H}_r}\left( \chi(t, j) \right)$, the summand $V_c\left(
x_c(t, j), z_c(t, j) \right)$ being non-negative, whence $\abs{v(t, j)} \leq
\kappa_v$.

As recorded after Eq. \eqref{eq:Hr_state}, the components $x$, $v$, and $x_c$
of the state are $\chi_1$, $\chi_2$, and $\chi_3$, respectively, while $z_c =
\chi_{4:n_z+3}$ and $q_c = \chi_{n_s}$, $n_s$ being $n_z + 4$.
Lemma~\ref{lem:integrated} is applied twice,
the applications being indexed by $\ell \in \{ 1, 2 \}$. For each such $\ell$,
let $\iota_1^\ell, \iota_2^\ell \in \{ 1, \dots, n_s \}$ be given by $\left(
\iota_1^1, \iota_2^1 \right) \triangleq (1, 2)$ and $\left( \iota_1^2,
\iota_2^2 \right) \triangleq (3, 2)$, and let $\rho^\ell : \mathbb{R}^{n_s}
\times \mathbb{R}^{n_i} \longrightarrow \mathbb{R}$ be given by $\rho^1(\chi,
u) \triangleq 1$ and $\rho^2(\chi, u) \triangleq \chi_{n_s}$. In the $\ell$th
application, the indices $\iota_1^\ell$ and $\iota_2^\ell$ and
the map $\rho^\ell$ are those of Eqs. \eqref{eq:integrated_f} to
\eqref{eq:integrated_hyp_k}.

By Eq. \eqref{eq:Hr_f}, the hypothesis of Eq. \eqref{eq:integrated_f} holds for
each $\ell \in \{ 1, 2 \}$, since $f_1(\chi, u) = v$ and $f_3(\chi, u) = q_c v$
for all $(\chi, u) \in \mathbb{R}^{n_s} \times \mathbb{R}^{n_i}$. By Eq.
\eqref{eq:Hr_g}, the hypothesis of Eq. \eqref{eq:integrated_g} holds for each
$\ell \in \{ 1, 2 \}$, since $g_1(\chi, u) = x$ and $g_3(\chi, u) = 0$ for all
$(\chi, u) \in \mathbb{R}^{n_s} \times \mathbb{R}^{n_i}$, with equality
for $\ell = 1$ and by the non-negativity of $\abs{x_c}$ for $\ell = 2$.

The hypothesis of Eq. \eqref{eq:integrated_hyp_rho} holds for $\ell = 1$,
$\rho^1$ being identically $1$, and for $\ell = 2$ by $q_c(t, j) \in \{ 0, 1
\}$, $\rho^2(\chi, u)$ being the component $\chi_{n_s}$ of the state, which is
$q_c$. The hypothesis of Eq.
\eqref{eq:integrated_hyp_k} holds with $\kappa \triangleq \kappa_v$ for each
$\ell \in \{ 1, 2 \}$, by the bound upon $\abs{v(t, j)}$ obtained above, the
component $\iota_2^\ell$ being $v$ in both cases.

Lemma~\ref{lem:integrated} accordingly gives $\abs{x(t, j)} \leq \abs{x(0, 0)}
+ \kappa_v T$ and $\abs{x_c(t, j)} \leq \abs{x_c(0, 0)} + \kappa_v T$.

The sets of Eqs. \eqref{eq:confinement_Hr_x} and \eqref{eq:confinement_Hr_V}
are bounded, by the bound upon $\abs{x_c(t, j)}$ and by $V_{\mathcal{H}_r}\left(
\chi(t, j) \right) \leq \varrho_T$, respectively, whence the set of Eq.
\eqref{eq:confinement_Hr_z} is bounded by
Proposition~\ref{prop:confinement_Hr}. The components $\chi_1(t, j)$ to
$\chi_3(t, j)$ are therefore bounded independently of $(t, j)$, by the bounds
upon $\abs{x(t, j)}$, $\abs{v(t, j)}$, and $\abs{x_c(t, j)}$ obtained above,
respectively; so is $\chi_{4:n_z+3}(t, j) =
z_c(t, j)$, by the boundedness of the set of Eq. \eqref{eq:confinement_Hr_z};
and so is $\chi_{n_s}(t, j) = q_c(t, j)$, by $q_c(t, j) \in \{ 0, 1 \}$. By the
identification of Appendix~\ref{sec:NCF}, the set
\[
\left\{ \chi(t, j) : (t, j) \in \dom \chi, \; t \leq T \right\}
\]
is accordingly bounded.

The system $\mathcal{H}_r$ satisfies the hybrid basic conditions by
Lemma~\ref{lem:hbc_Hr}, and the set $\mathcal{S}$ is invariant for
$\mathcal{H}_r$ by Lemma~\ref{lem:invariant} and satisfies Eq.
\eqref{eq:covering_S} by Lemma~\ref{lem:covering}. Since $T$ is arbitrary,
Proposition~\ref{prop:completeness} applies.
\end{proof}

\begin{remark}
Theorems \ref{thm:completeness} and \ref{thm:completeness_Hr} exclude the escape of the state in finite ordinary time; they do not exclude the Zeno behavior of
Definition~\ref{def:classification}, and the ordinary time attained by a
maximal solution pair may be finite.
\end{remark}

\begin{corollary}[Existence of complete solution pairs]\label{cor:existence}
Let the contact law of Eq. \eqref{eq:law} be confining, let $\mathcal{K}$ be
either of $\mathcal{H}_f$ and $\mathcal{H}_r$, with the dimensions $n_s$ and
$n_i$, let $\mathcal{S}$ be given by Eq. \eqref{eq:S_regimes}, let $w :
\mathbb{R}_{\geq 0} \longrightarrow \mathbb{R}^{n_i}$ be continuous and
bounded, and let $\chi_0 \in \mathcal{S}$. Then, there exists a complete
solution pair $(\chi, u)$ to $\mathcal{K}$ under $w$ such that $\chi(0, 0) =
\chi_0$.
\end{corollary}

\begin{proof}
Let $E_0 \triangleq \{ (0, 0) \}$, which is the compact hybrid time domain
generated by $\{ t_j \}_{j = 0}^{1}$ with $t_0 = t_1 = 0$, and, hence, a hybrid
time domain, let $\chi_e : E_0 \longrightarrow \mathbb{R}^{n_s}$ be
given by $\chi_e(0, 0) \triangleq \chi_0$, and let $u_e$ be the hybrid signal
induced by $w$ upon $E_0$. The set $\mathcal{J}_{E_0}$ being empty, $\chi_e$ is
a hybrid arc, and condition \ref{s:flow} holds vacuously, as does condition
\ref{s:jump}, $E_0$ containing no point of the form $(t, 1)$. Moreover,
$(\chi_e(0, 0), u_e(0, 0)) = (\chi_0, w(0)) \in C \cup D \subseteq \closure C
\cup D$ by Lemma~\ref{lem:covering}, where $C$ and $D$ are the flow set and the
jump set of $\mathcal{K}$. Hence, $(\chi_e, u_e)$ is a solution pair to
$\mathcal{K}$ under $w$, and, by Lemma~\ref{lem:maximal}, there exists a
solution pair $(\chi, u)$
to $\mathcal{K}$ that is maximal under $w$ and extends it, whence $\chi(0, 0) =
\chi_0$. Let $V_c$ be a confining storage function for the law. Since $n_i = 2$
if $\mathcal{K} = \mathcal{H}_f$, and $n_i = 1$ if $\mathcal{K} =
\mathcal{H}_r$, the pair $(\chi, u)$ is complete by
Theorem~\ref{thm:completeness} in the first case, and by
Theorem~\ref{thm:completeness_Hr} in the second.
\end{proof}

\begin{example}[Existence and Completeness under the BWSHCCL, BWMCL, and SHCCL]
\label{ex:complete}
Let the contact law be the BWSHCCL or the SHCCL of Example~\ref{ex:bwshc}, or
the BWMCL of Example~\ref{ex:bwm}. The law is confining by
Examples~\ref{ex:bwshc_confining} and \ref{ex:bwm_confining}, whence, by
Theorems~\ref{thm:completeness} and \ref{thm:completeness_Hr}, every solution
pair to $\mathcal{H}_f$ that is maximal under a continuous and bounded $w :
\mathbb{R}_{\geq 0} \longrightarrow \mathbb{R}^2$, and every solution pair to
$\mathcal{H}_r$ that is maximal under a continuous and bounded $w :
\mathbb{R}_{\geq 0} \longrightarrow \mathbb{R}$, is complete. By
Corollary~\ref{cor:existence}, moreover, each of $\mathcal{H}_f$ and
$\mathcal{H}_r$ has, under every such $w$ of the dimension of its input, a
complete solution pair from every element of the set $\mathcal{S}$ of Eq.
\eqref{eq:S_regimes} for that system.
\end{example}

\subsection{Non-Uniqueness}\label{sec:nu}

The solutions of $\mathcal{H}_f$ and $\mathcal{H}_r$ are not necessarily
unique. The examples below are of $\mathcal{H}_r$; they carry over to
$\mathcal{H}_f$, with $w$ of constant value $(0, 0)$, once the components
$x_m = x_{m,0} + v_{m,0} t$ and $v_m = v_{m,0}$ are added to the state.

\begin{example}[Non-uniqueness at a resting contact]\label{ex:nonuniq_rest}
Let the contact law of Eq. \eqref{eq:law} be arbitrary, subject to Eq.
\eqref{eq:vanishing}, and let $w : \mathbb{R}_{\geq 0} \longrightarrow
\mathbb{R}$ be the constant map of value $0$, which is continuous and bounded.
Let the initial state be that of Eq. \eqref{eq:Hr_init} with
\begin{equation}\label{eq:nonuniq_init}
x_0 \triangleq 0
\end{equation}
\begin{equation}\label{eq:nonuniq_v0}
v_0 \triangleq 0
\end{equation}
which is admissible, the restrictions of Section~\ref{sec:bodies} requiring
only $x_0 \in \mathbb{R}_{\geq 0}$ and $x_0 + l_1 + l_2 \in \mathbb{R}_{>0}$.
Let
\begin{equation}\label{eq:nonuniq_chi0}
\chi_0 \triangleq (0, 0, 0, 0, 0)
\end{equation}
denote the resulting initial state and let $u$ denote in either case below the
hybrid signal induced by $w$.

The first solution pair is $(\chi^{\mathrm{a}}, u)$ with
\begin{equation}\label{eq:nonuniq_dom_a}
\dom \chi^{\mathrm{a}} \triangleq \mathbb{R}_{\geq 0} \times \{ 0 \}
\end{equation}
\begin{equation}\label{eq:nonuniq_arc_a}
\chi^{\mathrm{a}}(t, 0) \triangleq (0, 0, 0, 0, 0)
\end{equation}
for all $t \in \mathbb{R}_{\geq 0}$: the bodies remain at rest and touching at a
single point, the interface disengaged, for all ordinary time.

The second solution pair is $(\chi^{\mathrm{b}}, u)$ with
\begin{equation}\label{eq:nonuniq_dom_b}
\dom \chi^{\mathrm{b}} \triangleq \{ 0 \} \times \mathbb{Z}_{\geq 0}
\end{equation}
\begin{equation}\label{eq:nonuniq_arc_b}
\chi^{\mathrm{b}}(0, j) \triangleq \left( 0, 0, 0, 0, \tfrac{1}{2}
  \left( 1 - (-1)^j \right) \right)
\end{equation}
for all $j \in \mathbb{Z}_{\geq 0}$: the interface is engaged and disengaged
without end, at the initial instant, no ordinary time elapsing.
\end{example}

\begin{proposition}[The first solution pair]\label{prop:nonuniq_a}
The pair $(\chi^{\mathrm{a}}, u)$ of Eqs. \eqref{eq:nonuniq_dom_a} and
\eqref{eq:nonuniq_arc_a} is a complete solution pair to $\mathcal{H}_r$ with
$\chi^{\mathrm{a}}(0, 0) = \chi_0$.
\end{proposition}

\begin{proof}
The set of Eq. \eqref{eq:nonuniq_dom_a} is a hybrid time domain, and
$\chi^{\mathrm{a}}$ is a hybrid arc in the sense of Definition~\ref{def:arc},
being constant. Condition \ref{s:jump} is vacuous, no $(t, j)$ of the domain
having $(t, j + 1)$ in it.

Let $t \in \mathbb{R}_{\geq 0}$. Since $q_c = 0$ and $x = 0$, the pair
$(\chi^{\mathrm{a}}(t, 0), 0)$ belongs to $C$ by Eq. \eqref{eq:H_C}. By Eq.
\eqref{eq:Hr_f} and $q_c = 0$, $v = 0$, and $w = 0$, every component of
$f(\chi^{\mathrm{a}}(t, 0), 0)$ vanishes, as does
$\dot{\chi}^{\mathrm{a}}(t, 0)$, whence condition \ref{s:flow}.

The pair is complete by Definition~\ref{def:classification}, the ordinary time
attained upon its domain being unbounded.
\end{proof}

\begin{proposition}[The second solution pair]\label{prop:nonuniq_b}
The pair $(\chi^{\mathrm{b}}, u)$ of Eqs. \eqref{eq:nonuniq_dom_b} and
\eqref{eq:nonuniq_arc_b} is a complete solution pair to $\mathcal{H}_r$ with
$\chi^{\mathrm{b}}(0, 0) = \chi_0$.
\end{proposition}

\begin{proof}
The set of Eq. \eqref{eq:nonuniq_dom_b} is a hybrid time domain, being the
union of the sets $\{ 0 \} \times \{ 0, 1, \dots, i \}$ with $i \in
\mathbb{Z}_{\geq 0}$, which increase with $i$ and each of which is a compact
hybrid time domain, generated by the sequence $\{ t_j \}_{j = 0}^{i + 1}$ with
$t_j = 0$ for all $j \in \{ 0, 1, \dots, i + 1 \}$. The map
$\chi^{\mathrm{b}}$ is a hybrid arc, the set
$\mathcal{J}_{\dom \chi^{\mathrm{b}}}$ being empty, whence condition
\ref{s:flow} is vacuous.

Let $j \in \mathbb{Z}_{\geq 0}$ and write $q \triangleq \tfrac{1}{2}(1 -
(-1)^j)$, which is $0$ for $j$ even and $1$ for $j$ odd. If $q = 0$, then $x =
0$ and $v = 0$, whence $(\chi^{\mathrm{b}}(0, j), 0) \in D$ by Eq.
\eqref{eq:H_D}. If $q = 1$, then $\Phi(0, 0, 0) = 0$ by Eq.
\eqref{eq:vanishing} and $v = 0$, whence likewise. By Eq. \eqref{eq:Hr_g},
$g(\chi^{\mathrm{b}}(0, j), 0) = (0, 0, 0, 0, 1 - q)$, which is
$\chi^{\mathrm{b}}(0, j + 1)$, whence condition \ref{s:jump}.

The pair is complete by Definition~\ref{def:classification}, the number of the
transitions attained upon its domain being unbounded.
\end{proof}

\begin{remark}
The pair $(\chi^{\mathrm{a}}, u)$ is eventually continuous and the pair
$(\chi^{\mathrm{b}}, u)$ is Zeno, in the sense of
Definition~\ref{def:classification}; both are complete; the two agree at $(0, 0)$.
\end{remark}

\begin{example}[Non-uniqueness at a separation]\label{ex:nonuniq_sep}
Let $n_z \triangleq 0$, let $k \in \mathbb{R}_{>0}$, and let the contact law of
Eq. \eqref{eq:law} be given by
\begin{equation}\label{eq:nonuniq_law}
\Phi(x_c, v_c) \triangleq k \min \{ x_c, 0 \}
\end{equation}
which is continuous and satisfies Eq. \eqref{eq:vanishing}. Let $w :
\mathbb{R}_{\geq 0} \longrightarrow \mathbb{R}$ be the constant map of value
$0$, let $x_0 \in \mathbb{R}_{>0}$ and $v_0 \in \mathbb{R}_{<0}$ be the
parameters of Eq. \eqref{eq:Hr_init}, and let
\begin{equation}\label{eq:nonuniq_omega}
\omega \triangleq \sqrt{k / m}
\end{equation}
\begin{equation}\label{eq:nonuniq_t1}
t_1 \triangleq - x_0 / v_0
\end{equation}
\begin{equation}\label{eq:nonuniq_t2}
t_2 \triangleq t_1 + \pi / \omega
\end{equation}
Let $u$ denote in either case below the hybrid signal induced by $w$.

Up to $(t_2, 1)$ the two solution pairs agree. Their common part is the pair
$(\xi, u)$ with
\begin{equation}\label{eq:nonuniq_dom_c}
\dom \xi \triangleq \left( [0, t_1] \times \{ 0 \} \right) \cup
  \left( [t_1, t_2] \times \{ 1 \} \right)
\end{equation}
\begin{equation}\label{eq:nonuniq_arc_c0}
\xi(t, 0) \triangleq \left( x_0 + v_0 t, \; v_0, \; 0, \; 0 \right)
\end{equation}
\begin{equation}\label{eq:nonuniq_arc_c1}
\xi(t, 1) \triangleq \left( \sigma(t), \; v_0 \cos \omega (t - t_1), \;
  \sigma(t), \; 1 \right)
\end{equation}
in which
\begin{equation}\label{eq:nonuniq_sigma}
\sigma(t) \triangleq \omega^{-1} v_0 \sin \omega (t - t_1)
\end{equation}
the components being $x$, $v$, $x_c$, and $q_c$, the internal state being
absent. The bodies approach at a constant relative velocity, meet at $t_1$,
and are in contact throughout $[t_1, t_2]$, the indentation being harmonic and
returning to zero at $t_2$, where the relative velocity is $-v_0$.

The third solution pair is $(\chi^{\mathrm{c}}, u)$, which extends $(\xi, u)$
by a transition at $(t_2, 1)$, with
\begin{equation}\label{eq:nonuniq_dom_d}
\dom \chi^{\mathrm{c}} \triangleq \dom \xi \cup
  \left( [t_2, +\infty) \times \{ 2 \} \right)
\end{equation}
\begin{equation}\label{eq:nonuniq_arc_d}
\chi^{\mathrm{c}}(t, 2) \triangleq \left( -v_0 (t - t_2), \; -v_0, \; 0, \;
  0 \right)
\end{equation}
for all $t \in [t_2, +\infty)$: the interface is disengaged at $t_2$ and the
bodies recede.

The fourth solution pair is $(\chi^{\mathrm{d}}, u)$, which extends $(\xi, u)$
by a continued evolution, with
\begin{equation}\label{eq:nonuniq_dom_e}
\dom \chi^{\mathrm{d}} \triangleq \left( [0, t_1] \times \{ 0 \} \right) \cup
  \left( [t_1, +\infty) \times \{ 1 \} \right)
\end{equation}
\begin{equation}\label{eq:nonuniq_arc_e}
\chi^{\mathrm{d}}(t, 1) \triangleq \left( -v_0 (t - t_2), \; -v_0, \;
  -v_0 (t - t_2), \; 1 \right)
\end{equation}
for all $t \in [t_2, +\infty)$: the interface remains engaged, transmitting no
force, and the bodies recede.
\end{example}

\begin{proposition}[The common part]\label{prop:nonuniq_c}
Let the conditions of Example~\ref{ex:nonuniq_sep} hold. Then, the pair $(\xi,
u)$ of Eqs. \eqref{eq:nonuniq_dom_c} to \eqref{eq:nonuniq_arc_c1} is a solution
pair to $\mathcal{H}_r$ with $\xi(0, 0) = (x_0, v_0, 0, 0)$.
\end{proposition}

\begin{proof}
The set of Eq. \eqref{eq:nonuniq_dom_c} is a compact hybrid time domain,
generated by $(0, t_1, t_2)$, and $\xi$ is a hybrid arc, the maps $\xi(\cdot,
0)$ and $\xi(\cdot, 1)$ being continuously differentiable upon $[0, t_1]$ and
$[t_1, t_2]$, respectively.

Let $t \in [0, t_1)$. Then, $q_c = 0$ and $x = x_0 + v_0 t > 0$, whence $(\xi(t,
0), 0) \in C$ by Eq. \eqref{eq:H_C}, and, in particular, $(\xi(0, 0), 0) \in
\closure C \cup D$. By Eq. \eqref{eq:Hr_f} and $q_c = 0$ and
$w = 0$, the flow map takes the value $(v_0, 0, 0, 0)$, which is
$\dot{\xi}(t, 0)$ by Eq. \eqref{eq:nonuniq_arc_c0}.

At $(t_1, 0)$, one has $x = 0$ and $v = v_0 < 0$, whence $(\xi(t_1, 0), 0) \in
D$ by Eq. \eqref{eq:H_D}; and, by Eq. \eqref{eq:Hr_g}, the jump map takes the
value $(0, v_0, 0, 1)$, which is $\xi(t_1, 1)$ by Eqs. \eqref{eq:nonuniq_arc_c1}
and \eqref{eq:nonuniq_sigma}, $\sigma(t_1)$ vanishing.

Let $t \in (t_1, t_2)$. Then, $q_c = 1$ and $\sigma(t) < 0$, the argument of the
sine of Eq. \eqref{eq:nonuniq_sigma} lying in $(0, \pi)$ and $v_0$ being
negative, whence $\Phi = k \sigma(t) < 0$ by Eq. \eqref{eq:nonuniq_law} and
$(\xi(t, 1), 0) \in C$ by Eq. \eqref{eq:H_C}. By Eq. \eqref{eq:Hr_f} and $w =
0$, the flow map takes the value
\[
\left( v_0 \cos \omega (t - t_1), \; -m^{-1} k \sigma(t), \;
  v_0 \cos \omega (t - t_1), \; 0 \right)
\]
and differentiation of Eqs. \eqref{eq:nonuniq_arc_c1} and
\eqref{eq:nonuniq_sigma} gives the same, $\omega^2$ being $k / m$ by Eq.
\eqref{eq:nonuniq_omega}. Conditions \ref{s:flow} and \ref{s:jump} are thereby
satisfied.
\end{proof}

\begin{proposition}[The third solution pair]\label{prop:nonuniq_d}
Let the conditions of Example~\ref{ex:nonuniq_sep} hold. Then, the pair
$(\chi^{\mathrm{c}}, u)$ of Eqs. \eqref{eq:nonuniq_dom_d} and
\eqref{eq:nonuniq_arc_d} is a complete solution pair to $\mathcal{H}_r$ with
$\chi^{\mathrm{c}}(0, 0) = (x_0, v_0, 0, 0)$.
\end{proposition}

\begin{proof}
The pair $(\chi^{\mathrm{c}}(0, 0), 0)$ belongs to $\closure C \cup D$, and the
conditions \ref{s:flow} and \ref{s:jump} are satisfied upon $\dom \xi$, by
Proposition~\ref{prop:nonuniq_c}; the set of Eq. \eqref{eq:nonuniq_dom_d}
is a hybrid time domain.

At $(t_2, 1)$, one has $q_c = 1$, $\sigma(t_2) = 0$, whence $\Phi = 0$ by Eq.
\eqref{eq:nonuniq_law}, and $v = v_0 \cos \pi = -v_0 > 0$; thus,
$(\chi^{\mathrm{c}}(t_2, 1), 0) \in D$ by Eq. \eqref{eq:H_D}. By Eq.
\eqref{eq:Hr_g}, the jump map takes the value $(0, -v_0, 0, 0)$, which is
$\chi^{\mathrm{c}}(t_2, 2)$ by Eq. \eqref{eq:nonuniq_arc_d}.

Let $t \in (t_2, +\infty)$. Then, $q_c = 0$ and $x = -v_0 (t - t_2) > 0$, whence
$(\chi^{\mathrm{c}}(t, 2), 0) \in C$ by Eq. \eqref{eq:H_C}; and, by Eq.
\eqref{eq:Hr_f}, the flow map takes the value $(-v_0, 0, 0, 0)$, which is
$\dot{\chi}^{\mathrm{c}}(t, 2)$. The pair is complete by
Definition~\ref{def:classification}, the ordinary time attained upon its domain
being unbounded.
\end{proof}

\begin{proposition}[The fourth solution pair]\label{prop:nonuniq_e}
Let the conditions of Example~\ref{ex:nonuniq_sep} hold. Then, the pair
$(\chi^{\mathrm{d}}, u)$ of Eqs. \eqref{eq:nonuniq_dom_e} and
\eqref{eq:nonuniq_arc_e} is a complete solution pair to $\mathcal{H}_r$ with
$\chi^{\mathrm{d}}(0, 0) = (x_0, v_0, 0, 0)$.
\end{proposition}

\begin{proof}
The set of Eq. \eqref{eq:nonuniq_dom_e} is a hybrid time domain. The pair
$(\chi^{\mathrm{d}}(0, 0), 0)$ belongs to $\closure C \cup D$, and the
conditions \ref{s:flow} and \ref{s:jump} are satisfied upon $\dom \xi$, by
Proposition~\ref{prop:nonuniq_c}. The map $\chi^{\mathrm{d}}(\cdot, 1)$ is
continuously differentiable upon $[t_1, +\infty)$, the derivatives of Eqs.
\eqref{eq:nonuniq_arc_c1} and \eqref{eq:nonuniq_arc_e} agreeing at $t_2$, where
each is $(-v_0, 0, -v_0, 0)$.

Let $t \in [t_2, +\infty)$. Then, $q_c = 1$ and $x_c = -v_0 (t - t_2) \geq 0$,
whence $\Phi = 0$ by Eq. \eqref{eq:nonuniq_law} and $(\chi^{\mathrm{d}}(t, 1),
0) \in C$ by Eq. \eqref{eq:H_C}. By Eq. \eqref{eq:Hr_f} and $w = 0$, the flow
map takes the value $(-v_0, 0, -v_0, 0)$, which is
$\dot{\chi}^{\mathrm{d}}(t, 1)$ by Eq. \eqref{eq:nonuniq_arc_e}. Condition
\ref{s:jump} is satisfied, no further transition occurring. The pair is
complete by Definition~\ref{def:classification}, the ordinary time attained
upon its domain being unbounded.
\end{proof}

The models of this article admit solutions that are not physically meaningful,
and such solutions may coexist with those that are. Two
rules of practice are advanced.
\begin{enumerate}[label=(R\arabic*), ref=R\arabic*, leftmargin=*, nosep]
\item\label{r:separation} Among solution pairs that coincide up to a hybrid
time $(t, j)$ at which the interface is engaged, the contact force vanishes,
and the bodies recede, with $v > 0$, any that does not take a transition at $(t, j)$ is
excluded. At $(t_2, 1)$ in Example~\ref{ex:nonuniq_sep}, this selects the
solution pair of Eq. \eqref{eq:nonuniq_arc_d} and excludes that of Eq.
\eqref{eq:nonuniq_arc_e}, in which the interface remains engaged for all
subsequent time while transmitting no force.
\item\label{r:flow} Among the remaining solution pairs, any that takes a
transition at a hybrid time at which another of them, coinciding with it up to
that time, continues by flowing is excluded. At the configuration of Eq.
\eqref{eq:nonuniq_chi0}, this selects the solution pair of Eq.
\eqref{eq:nonuniq_arc_a} and excludes that of Eq. \eqref{eq:nonuniq_arc_b}, in
which no ordinary time elapses.
\end{enumerate}

Rule \ref{r:flow} presupposes that a flow cannot be continued beyond the
instant at which the contact force is extinguished with the bodies receding.
Where it can, as at $(t_2, 1)$ in Example~\ref{ex:nonuniq_sep}, $\Phi$
vanishing for positive indentations, rule \ref{r:flow} taken alone selects the
solution pair of Eq. \eqref{eq:nonuniq_arc_e}, which rule \ref{r:separation}
excludes; rule \ref{r:separation} accordingly takes precedence. Where the
contact force merely touches zero, the bodies receding, that rule retains only
the solution pairs in which a contact that the law would continue is
terminated; this is consistent with the delimitation of a contact in
Section~\ref{sec:interface} and with the treatment given there of returns
within an overlap.

Neither rule is imposed by the model, and a computation that observes them
selects one solution pair among those admitted rather than computing the
solution; the selection is a modeling decision. The two examples exhibit the
phenomenon and do not exhaust it: the flow set and the jump set meet at
configurations that neither example attains, and the branching there has not
been examined. Two questions are accordingly left open: which of those
configurations are reachable from the initial conditions of
Section~\ref{sec:bodies}, and whether the two rules, or others, suffice to
determine the solution uniquely wherever the solutions branch.

These are, however, difficulties of a limited kind. The configurations at
which the two sets meet are described by equalities, $x$, $v$, or $\Phi$
vanishing there, and, along the solutions computed in
Section~\ref{sec:simulation}, they are encountered only at the instants at which
a contact commences or terminates.

\begin{remark}
A computation of the solutions of $\mathcal{H}_f$ or $\mathcal{H}_r$ must decide
whether to flow or to jump wherever the computed state lies in both the flow set
and the jump set. For hybrid systems in general, the HyEQ Toolbox decides by a
priority given, at the choice of the user, to jumps, to flows, or to neither
(see Refs. \cite{sanfelice_toolbox_2013} and \cite{copp_zero-crossing_2016}); 
rules \ref{r:separation} and \ref{r:flow}, which the computations of 
Section~\ref{sec:simulation} observe, are specific to the models of this article,
giving priority to jumps at a separation and to flows elsewhere.
\end{remark}

\section{Simulation}\label{sec:simulation}

\subsection{Background}\label{sec:sim_background}

Three studies are reported. The first is of the reduced system, a ball falling
upon a fixed obstacle under gravity, with the BWMCL; the second is
of the full system, two balls upon a common axis under prescribed forces, with
the BWSHCCL; the third is again of the reduced system, and compares the two
with the SHCCL of Example~\ref{ex:bwshc} at a collision velocity so small
that the passage to permanent contact occupies the greater part of the motion.
Together, they exercise the whole of the model of Section~\ref{sec:model} and
contact laws carrying none, one, and two internal states. The elements that the
studies share are set out in this subsection: the numerical method, the
determination of the parameters, the presentation of the results, and the
observables of the collisions; what is particular to each study is given with
it.

Each flow interval is integrated by the routine \texttt{solve\_ivp} of SciPy
\cite{virtanen_scipy_2020} with the method \texttt{DOP853}
\cite{hairer_solving_1993}, the relative tolerance $10^{-11}$, the absolute
tolerance $10^{-14}$, and the step size bounded above. A jump is taken when a continuous
function of the state whose sign marks membership of the flow set changes sign within a step; the state at the located event is first advanced along the flow, by the least interval that a
derivative-free search finds, into the jump set with a small margin. The test
suite of the code checks every computed solution against the definition of a
solution pair: the hybrid time domain is ordered; the state lies in the flow
set at every sample but the last before a jump, and in the jump set at that
sample; and each state after a jump is the image of the one before it under the
jump map, exactly. Independently, the forward Euler method, started from the
initial state of each flow interval, is required to converge to its final state
at an estimated order of at least $0.9$ under grid refinement. The arithmetic
is that of the binary64 format of IEEE 754 in the default rounding mode
\cite{ieee_ieee_2019}; the implementation is in Python 3.12 and rests upon
NumPy 2.3.3 \cite{harris_array_2020}, SciPy 1.16.2 \cite{virtanen_scipy_2020}, and
Matplotlib 3.10.6 \cite{hunter_matplotlib_2007}. It is available from the repository of the corresponding author.\footnote{\url{https://gitlab.com/user9716869/VEPBB}}

It should be remarked that the solutions of $\mathcal{H}_f$ and $\mathcal{H}_r$ need not be unique,
and each study reports one of them, which observes rules \ref{r:separation} and
\ref{r:flow} of Section~\ref{sec:nu} to within the numerical tolerances: each
transition is taken at the earliest instant at which the search described above
locates the state in the jump set. The three laws employed are compressive, by
Examples~\ref{ex:shc_compressive}, \ref{ex:bwshc_compressive}, and
\ref{ex:bwm_compressive}, and the contact force is continuous under each, by
Corollary~\ref{cor:force_continuous}.

Some of the parameters follow from the properties of the materials. Writing
$Y_i$, $\nu_i$, $\rho_i$, and $R_i$ for the Young modulus, the Poisson ratio,
the density, and the radius of the body $\mathcal{B}_i$, let
\begin{equation}
Y^{*} \triangleq \left( \frac{1 - \nu_1^2}{Y_1}
  + \frac{1 - \nu_2^2}{Y_2} \right)^{-1}
\end{equation}
and
\begin{equation}
R_e \triangleq \left( R_1^{-1} + R_2^{-1} \right)^{-1}
\end{equation}
The mass of a solid sphere is
\begin{equation}
m_i = \tfrac{4}{3} \pi R_i^3 \rho_i
\end{equation}
and the exponent $p$ of Eqs. \eqref{eq:bwshc_phi} and \eqref{eq:bwm_phi} is that
of the Hertzian relation,
\begin{equation}
p = \tfrac{3}{2}
\end{equation}
Upon loading from the undeformed state, while the hysteretic terms are negligible, the hysteretic displacement is $z_c = A x_c$ under the BWSHCCL and
$z_{c,1} = A z_{c,2}$ under the BWMCL, so
that, by Eqs. \eqref{eq:bwshc_phi} and \eqref{eq:bwm_phi}, the elastic and
hysteretic branches together
constitute a Hertzian spring of coefficient $\left( \alpha + (1 - \alpha) A^{p}
\right) k$, of which the deflection is $x_c$ under the former law and $z_{c,2}$
under the latter. The coefficient $k$ is chosen so that this is the Hertzian
coefficient:
\begin{equation}\label{eq:k_hertz}
k = \frac{k_H}{\alpha + (1 - \alpha) A^{p}}
\end{equation}
where 
\begin{equation}
k_H \triangleq \tfrac{4}{3} Y^{*} \sqrt{R_e}
\end{equation}
For the SHCCL of Example~\ref{ex:bwshc}, Eq. \eqref{eq:k_hertz} with $\alpha
= 1$ gives $k = k_H$.

The parameters $A$, $\alpha$, $\beta$, $\gamma$, $n$, and $c$ are not determined
by the material data, and were chosen in each study to display the features with
which this article is concerned. No claim is made that the values chosen
identify the response of the materials named; their identification against
measurement lies outside the scope of this article, and the results below are
to be read as a demonstration of the formulation rather than as a prediction.

Each solution pair is defined upon a hybrid time domain, and the quantities
plotted against the ordinary time are therefore
composed of the restrictions of its hybrid arc to the successive intervals of
flow. In the first two studies, part (a) of the figure shows the evolution of
the state and part (b) the observables of the successive collisions. Plotted in
part (a) are the applied forces; the relative displacement $x$ and
the relative velocity $v$ of Eq. \eqref{eq:main}; the components of the
internal state $z_c$ of the contact law; and the contact force
$F_c$. The intervals of flow upon which $q_c$ is unity are shaded.

The observables plotted for the successive collisions are the following. Let $E$
be the hybrid time domain of a solution pair reported below. Since $q_c$ is
constant upon each
interval of flow, and is exchanged for $1 - q_c$ at every transition, the
intervals alternate between the two regimes; the bodies being apart at the
initial instant in each study, the collision of index $\ell \in \mathbb{Z}_{\geq 0}$ occupies
\begin{equation}\label{eq:obs_interval}
I_E^{2\ell + 1} = [t_e^{\ell}, t_s^{\ell}]
\end{equation}
(the collisions being the contacts that are terminated by a separation) and the flight which follows it occupies $I_E^{2\ell + 2}$. Every quantity below is
evaluated at the hybrid time $(t, 2\ell + 1)$, that is, within the collision and not
after the transition by which it is terminated. The duration of the contact is
\begin{equation}\label{eq:t_d}
t_d^{\ell} \triangleq t_s^{\ell} - t_e^{\ell}
\end{equation}
and that of the flight is
\begin{equation}\label{eq:t_f}
t_f^{\ell} \triangleq t_e^{\ell + 1} - t_s^{\ell}
\end{equation}
The approach velocity is
\begin{equation}\label{eq:v_i}
v_i^{\ell} \triangleq v(t_e^{\ell}, 2\ell + 1)
\end{equation}
and the coefficient of restitution is that of Newton,
\begin{equation}\label{eq:e_i}
e^{\ell} \triangleq -v(t_s^{\ell}, 2\ell + 1) / v_i^{\ell}
\end{equation}
The greatest indentation and the greatest contact force are
\begin{equation}\label{eq:xc_max}
x_c^{\max, \ell} \triangleq \max \left\{ \abs{x_c(t, 2\ell + 1)} :
  t \in I_E^{2\ell + 1} \right\}
\end{equation}
\begin{equation}\label{eq:F_max}
F_c^{\max, \ell} \triangleq \max \left\{ F_c(t, 2\ell + 1) :
  t \in I_E^{2\ell + 1} \right\}
\end{equation}
and the residual indentation is
\begin{equation}\label{eq:xc_r}
x_c^{r, \ell} \triangleq \abs{x_c(t_s^{\ell}, 2\ell + 1)}
\end{equation}
this being that which the interface retains at the separation and which the jump
map discards. The energy surrendered by the relative motion to the interface is
\begin{equation}\label{eq:W}
W^{\ell} \triangleq -\int_{t_e^{\ell}}^{t_s^{\ell}} F_c \, v \, \mathrm{d} t
\end{equation}
Lastly, let
\begin{equation}\label{eq:J}
J_{F_c}^{\ell} \triangleq \int_{t_e^{\ell}}^{t_s^{\ell}} F_c \, \mathrm{d} t, \quad
J_{F_e}^{\ell} \triangleq \int_{t_e^{\ell}}^{t_s^{\ell}} F_{e,r} \, \mathrm{d} t
\end{equation}
be the impulses of the contact force and of the applied force over the contact.
Writing $\Delta v^{\ell} \triangleq v(t_s^{\ell}, 2\ell + 1) - v_i^{\ell}$ for the velocity
imparted by the collision, Eq. \eqref{eq:main} gives
\begin{equation}\label{eq:momentum}
m \, \Delta v^{\ell} = J_{F_c}^{\ell} + J_{F_e}^{\ell}
\end{equation}
and the extent to which the applied force intrudes upon the collision is
measured by
\begin{equation}\label{eq:Lambda}
\Lambda^{\ell} \triangleq J_{F_e}^{\ell} / J_{F_c}^{\ell}
\end{equation}
This is negative while the applied force opposes the rebound, it is small in
magnitude when the collision may be treated as impulsive, and, by Eq.
\eqref{eq:momentum}, it approaches $-1$ as $m \, \Delta v^{\ell}$ becomes
negligible beside $J_{F_c}^{\ell}$, the contact force then serving chiefly to
balance the applied force rather than to reverse the relative motion.

\subsection{Bouncing Ball}

\begin{figure*}[!p]
\centering

\begin{subfigure}[t]{0.485\textwidth}
\centering
\input{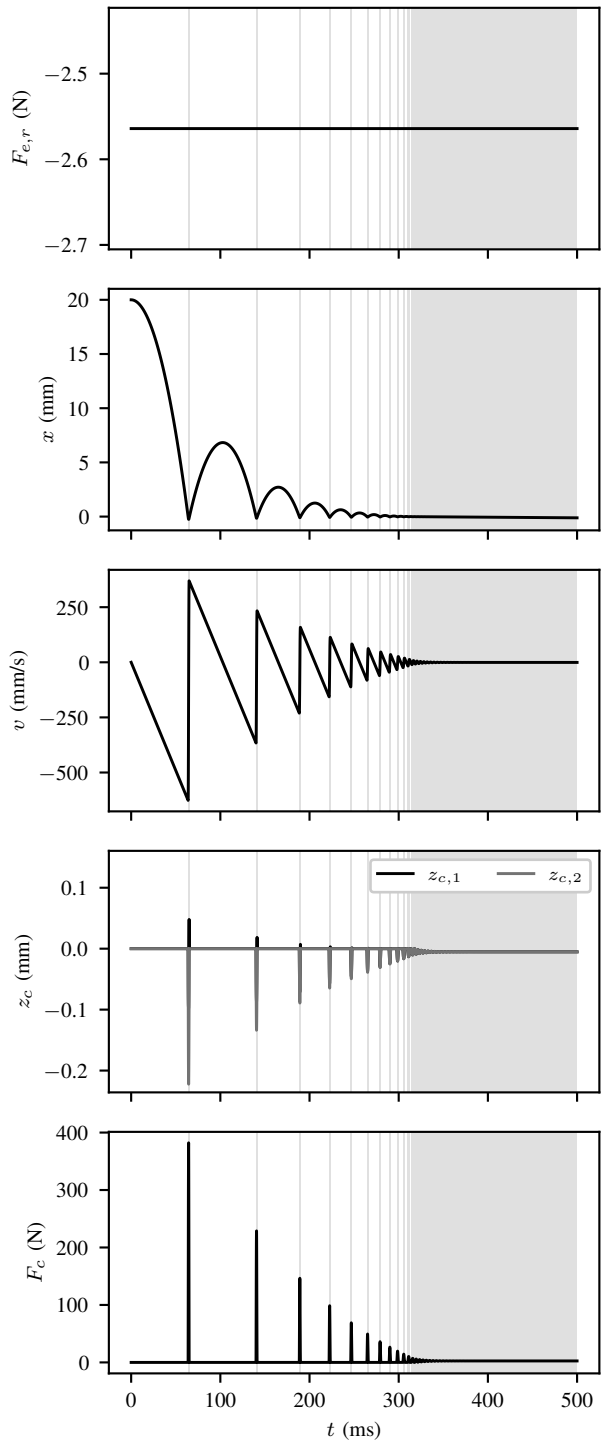}
\caption{Evolution of the state.}
\end{subfigure}
\hfill
\begin{subfigure}[t]{0.485\textwidth}
\centering
\input{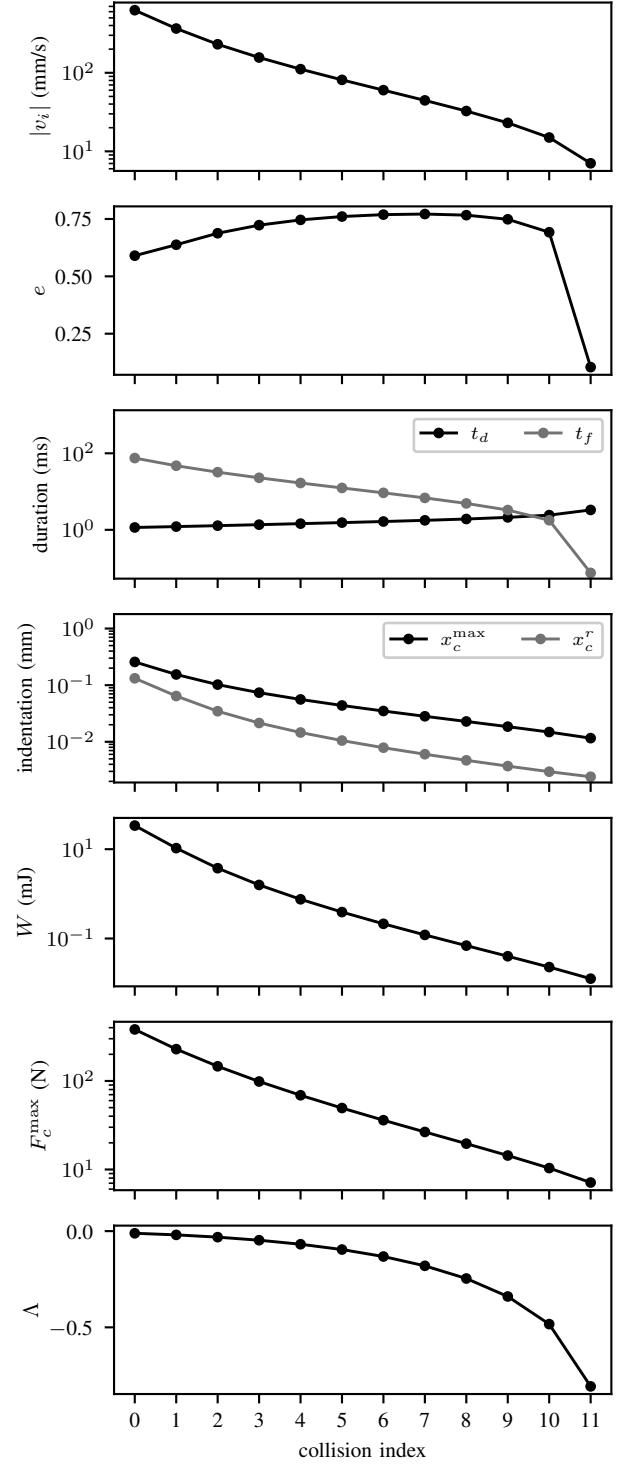}
\caption{Observables of the successive collisions.}
\end{subfigure}

\caption{Repeated impact of a hardened steel ball upon a fixed bituminous
substrate under gravity, computed with the BWMCL. The
shaded intervals of (a) are those upon which the bodies are in
contact.}\label{fig:one}

\end{figure*}

\begin{table}[!t]
\caption{Parameters of the first study.}\label{tab:one}
\centering
\begin{tabular}{llll}
\hline
& quantity & value & unit \\
\hline
\multicolumn{4}{l}{\emph{Materials and geometry}} \\
$\mathcal{B}_1$ & hardened steel & & \\
& $Y_1$ & $200$ & GPa \\
& $\nu_1$ & $0.30$ & \\
& $\rho_1$ & $7800$ & kg\,m$^{-3}$ \\
& $R_1$ & $20$ & mm \\
$\mathcal{B}_2$ & bitumen & & \\
& $Y_2$ & $1.0$ & GPa \\
& $\nu_2$ & $0.35$ & \\
& $R_2$ & $\infty$ & \\
\hline
\multicolumn{4}{l}{\emph{Inferred}} \\
& $Y^{*}$ & $1.134 \times 10^{9}$ & Pa \\
& $R_e$ & $20$ & mm \\
& $m$ & $2.614 \times 10^{-1}$ & kg \\
& $p$ & $3/2$ & \\
& $k_H = k$ & $2.138 \times 10^{8}$ & N\,m$^{-3/2}$ \\
\hline
\multicolumn{4}{l}{\emph{Chosen}} \\
& $A$ & $1$ & \\
& $\alpha$ & $0.35$ & \\
& $\beta = \gamma$ & $5.000 \times 10^{7}$ & m$^{-2}$ \\
& $n$ & $2$ & \\
& $c$ & $4.400 \times 10^{3}$ & N\,s\,m$^{-1}$ \\
\hline
\multicolumn{4}{l}{\emph{Initial conditions and applied force}} \\
& $x_0$ & $20$ & mm \\
& $v_0$ & $0$ & m\,s$^{-1}$ \\
& $g$ & $9.81$ & m\,s$^{-2}$ \\
& $F_{e,r}$ & $-mg$ & N \\
\hline
\end{tabular}
\end{table}

The first study is of a ball of hardened steel of $40$ mm diameter, released
from rest at a height of $20$ mm above a fixed substrate of bitumen and falling
under gravity. The system is that of Eqs. \eqref{eq:Hr_f} to
\eqref{eq:Hr_init}, the substrate being of infinite mass, the contact law is
the BWMCL, and the interval simulated is $0.5$ s. The pairing
is chosen for the sake of the substrate: bitumen possesses no equilibrium
modulus, and is accordingly a viscoelastic fluid, of which the element of
Maxwell borne by the contact law is the simplest representation. The ball is
hard enough to rebound from it, and yet sinks into it without bound once in
permanent contact, and the study exhibits both.

The parameters are inferred and chosen as set out in
Section~\ref{sec:sim_background}. The modulus $Y_2$ is that which the substrate exhibits over the duration of a
contact, and not a modulus measured at rest, which the substrate does not possess.
The parameters $A$, $\alpha$, $\beta$, $\gamma$, $n$, and $c$ were chosen so
that the study exhibits the following features: the first collisions leave a residual
indentation that is an appreciable fraction of the greatest indentation, that
fraction diminishes over the sequence, and the passage to permanent contact occurs
within the interval simulated and is followed by sinking. The sinking is
governed by $c$, which sets both the rate $mg/c$ at which the interface flows
under the weight of the ball and the relaxation of the element of Maxwell within
a contact, the latter being comparable with the duration of a contact. The
values are recorded in Table \ref{tab:one}.

The evolution of the state is shown in Fig. \ref{fig:one}(a), in which the
applied force is $F_{e,r}$ and the internal state of the contact law has the two
components $z_{c,1}$ and $z_{c,2}$. Over the sequence of collisions, the contacts lengthen and the flights shorten
until the two are comparable and the sequence terminates. The ball thereafter
remains in contact and sinks steadily at the rate $mg/c$, the interface
possessing no configuration at which it supports the weight of the ball without
flowing.

The observables of the successive collisions are shown in Fig. \ref{fig:one}(b).
The approach velocity, the greatest indentation, the residual indentation, the
greatest contact force, and the energy surrendered to the interface diminish
from one collision to the next. The duration of the contact grows over the
sequence while that of the flight falls, until the flights are shorter than the
contacts and the sequence terminates. The coefficient of restitution rises over
the first collisions and is then nearly constant; it falls over the last
collisions, steeply at the last. The ratio $\Lambda$ is small in magnitude over
the first collisions, which admit of treatment as impulsive, and
approaches unity in magnitude at the last. The final decline of the coefficient
of restitution is of this origin: the impulse of the weight is no longer
negligible beside that of the contact.

\subsection{Repeated Binary Collisions}

\begin{figure*}[!p]
\centering

\begin{subfigure}[t]{0.485\textwidth}
\centering
\input{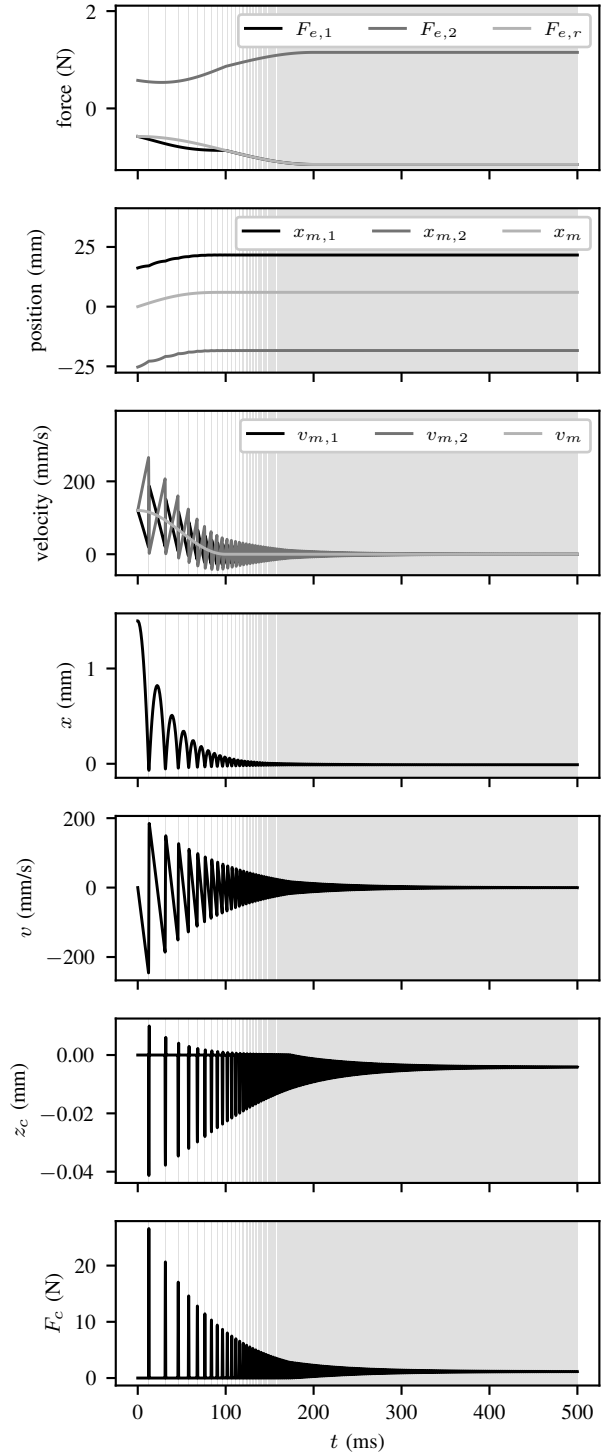}
\caption{Evolution of the state.}
\end{subfigure}
\hfill
\begin{subfigure}[t]{0.485\textwidth}
\centering
\input{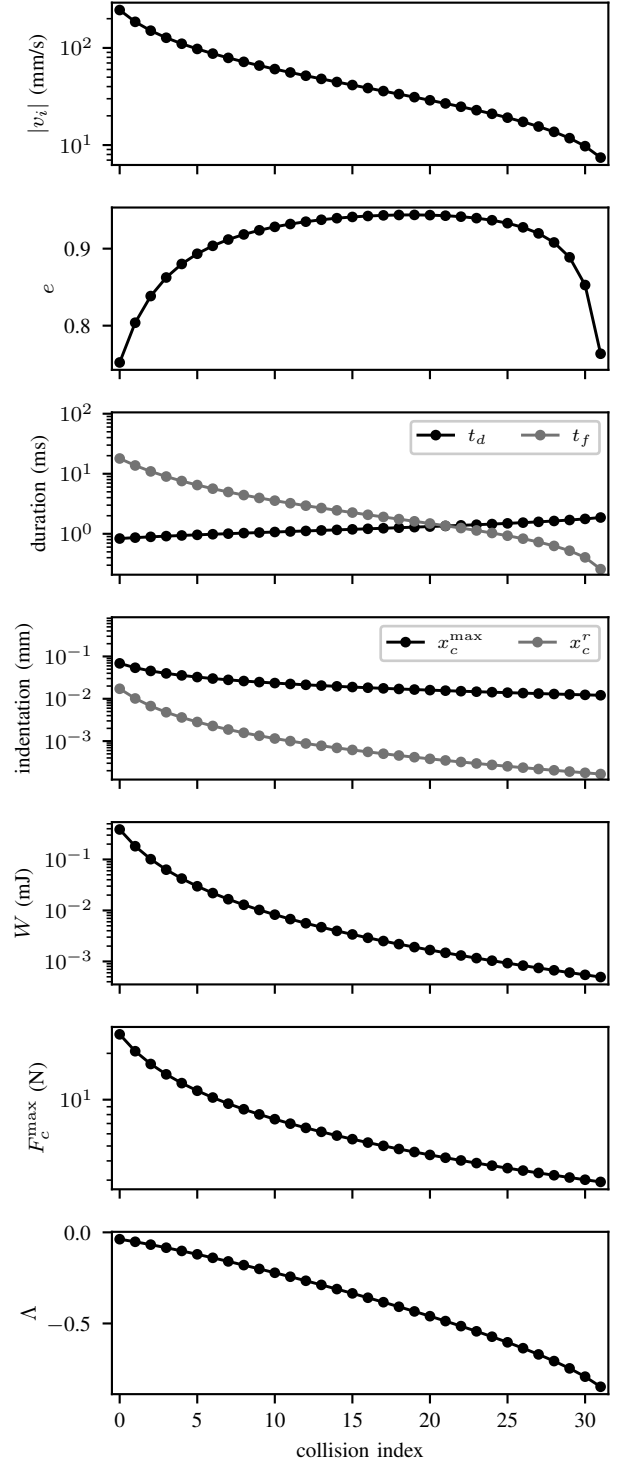}
\caption{Observables of the successive collisions.}
\end{subfigure}

\caption{Repeated collision of a ball of poly(tetrafluoroethylene) with a ball
of polyoxymethylene, of equal diameter and constrained to move upon a common
axis in the absence of gravity, computed with the
BWSHCCL. The applied forces bring the center
of mass of the pair to rest while drawing the balls together. The shaded
intervals of (a) are those upon which the bodies are in
contact.}\label{fig:two}

\end{figure*}

\begin{table}[!t]
\caption{Parameters of the second study.}\label{tab:two}
\centering
\begin{tabular}{llll}
\hline
& quantity & value & unit \\
\hline
\multicolumn{4}{l}{\emph{Materials and geometry}} \\
$\mathcal{B}_1$ & poly(tetrafluoroethylene) & & \\
& $Y_1$ & $0.50$ & GPa \\
& $\nu_1$ & $0.46$ & \\
& $\rho_1$ & $2200$ & kg\,m$^{-3}$ \\
& $R_1$ & $20$ & mm \\
$\mathcal{B}_2$ & polyoxymethylene & & \\
& $Y_2$ & $2.90$ & GPa \\
& $\nu_2$ & $0.35$ & \\
& $\rho_2$ & $1410$ & kg\,m$^{-3}$ \\
& $R_2$ & $20$ & mm \\
\hline
\multicolumn{4}{l}{\emph{Inferred}} \\
& $Y^{*}$ & $5.321 \times 10^{8}$ & Pa \\
& $R_e$ & $10$ & mm \\
& $m_1$ & $7.372 \times 10^{-2}$ & kg \\
& $m_2$ & $4.725 \times 10^{-2}$ & kg \\
& $m$ & $2.879 \times 10^{-2}$ & kg \\
& $\eta$ & $0.6094$ & \\
& $p$ & $3/2$ & \\
& $k_H = k$ & $7.095 \times 10^{7}$ & N\,m$^{-3/2}$ \\
\hline
\multicolumn{4}{l}{\emph{Chosen}} \\
& $A$ & $1$ & \\
& $\alpha$ & $0.35$ & \\
& $\beta = \gamma$ & $2.409 \times 10^{8}$ & m$^{-2}$ \\
& $n$ & $2$ & \\
& $c$ & $2.500 \times 10^{7}$ & N\,s\,m$^{-5/2}$ \\
\hline
\multicolumn{4}{l}{\emph{Initial conditions and applied forces}} \\
& $x_0$ & $1.5$ & mm \\
& $v_0$ & $0$ & m\,s$^{-1}$ \\
& $x_{m,0}$ & $0$ & mm \\
& $v_{m,0}$ & $120$ & mm\,s$^{-1}$ \\
& $F_{e,0}$ & $0.5759$ & N \\
& $T_e$ & $0.20$ & s \\
& $T_s$ & $0.10$ & s \\
\hline
\end{tabular}
\end{table}

The second study is of two balls of $40$ mm diameter, one of
poly(tetrafluoroethylene) and one of polyoxymethylene, constrained to move upon
a common axis in the absence of gravity and driven by prescribed axial forces.
The system is that of Eqs. \eqref{eq:H_f} to \eqref{eq:H_init}, the contact
law is the BWSHCCL, and the interval simulated is $0.5$ s.

Let $S : \mathbb{R} \longrightarrow \mathbb{R}$ be given by
\begin{equation}\label{eq:ramp}
S(w) \triangleq
\begin{cases}
0 & w < 0 \\
\tfrac{1}{2} \left( 1 - \cos \pi w \right) & 0 \leq w \leq 1 \\
1 & w > 1
\end{cases}
\end{equation}
which is continuously differentiable, its derivative being
\begin{equation}\label{eq:ramp_d}
S'(w) =
\begin{cases}
0 & w < 0 \\
\tfrac{1}{2} \pi \sin \pi w & 0 \leq w \leq 1 \\
0 & w > 1
\end{cases}
\end{equation}
and, with $F_{e,0}, v_{m,0}, T_e, T_s \in \mathbb{R}_{>0}$, let $a :
\mathbb{R}_{\geq 0} \longrightarrow \mathbb{R}$ be given by
\begin{equation}\label{eq:a_of_t}
a(t) \triangleq \frac{v_{m,0}}{T_s} \, S'(t / T_s)
\end{equation}
and $F_{e,1}, F_{e,2} : \mathbb{R}_{\geq 0} \longrightarrow \mathbb{R}$ given by
\begin{equation}\label{eq:u_1}
F_{e,1}(t) \triangleq -F_{e,0} \left( 1 + S(t / T_e) \right) - m_1 a(t)
\end{equation}
\begin{equation}\label{eq:u_2}
F_{e,2}(t) \triangleq F_{e,0} \left( 1 + S(t / T_e) \right) - m_2 a(t)
\end{equation}
The first term draws the balls together, growing in magnitude from $F_{e,0}$ to
$2 F_{e,0}$ and constant for $t \geq T_e$; the second decelerates both bodies
alike, whence 
\begin{equation}
v_m(t) = v_{m,0} \left( 1 - S(t / T_s) \right)
\end{equation}
which vanishes for $t \geq T_s$. By Eq. \eqref{eq:u}, the reduced force is 
\begin{equation}
F_{e,r}(t) = -F_{e,0} \left( 1 + S(t/ T_e) \right)
\end{equation}
the deceleration contributing nothing to it.

The parameters are inferred and chosen as set out in
Section~\ref{sec:sim_background}, the coefficient
$k$ being given by Eq. \eqref{eq:k_hertz}. The values are
recorded in Table \ref{tab:two}.

The evolution of the state is shown in Fig. \ref{fig:two}(a), the quantities
plotted against the ordinary time as set out in
Section~\ref{sec:sim_background}. The contacts lengthen
and the flights shorten until the two are comparable and the sequence
terminates. The center of mass is at rest well before the sequence terminates;
the applied forces are thereafter equal and opposite, and $F_{e,1}$ and
$F_{e,r}$ are indistinguishable in the figure. Upon the termination of the
sequence, the balls remain in contact, settling towards an indentation at which
the elastic and hysteretic branches sustain the reduced force.

The observables of the successive collisions are shown in Fig. \ref{fig:two}(b).
The greatest indentation, the residual indentation, the greatest contact force,
and the energy surrendered to the interface diminish from one collision to the
next. The coefficient of restitution rises as the approach velocity falls, until the
impulse of the applied force is no longer negligible beside that of the contact;
over the collisions that follow, it falls.

\subsection{The Approach to Permanent Contact}

\begin{figure*}[!t]
\centering
\input{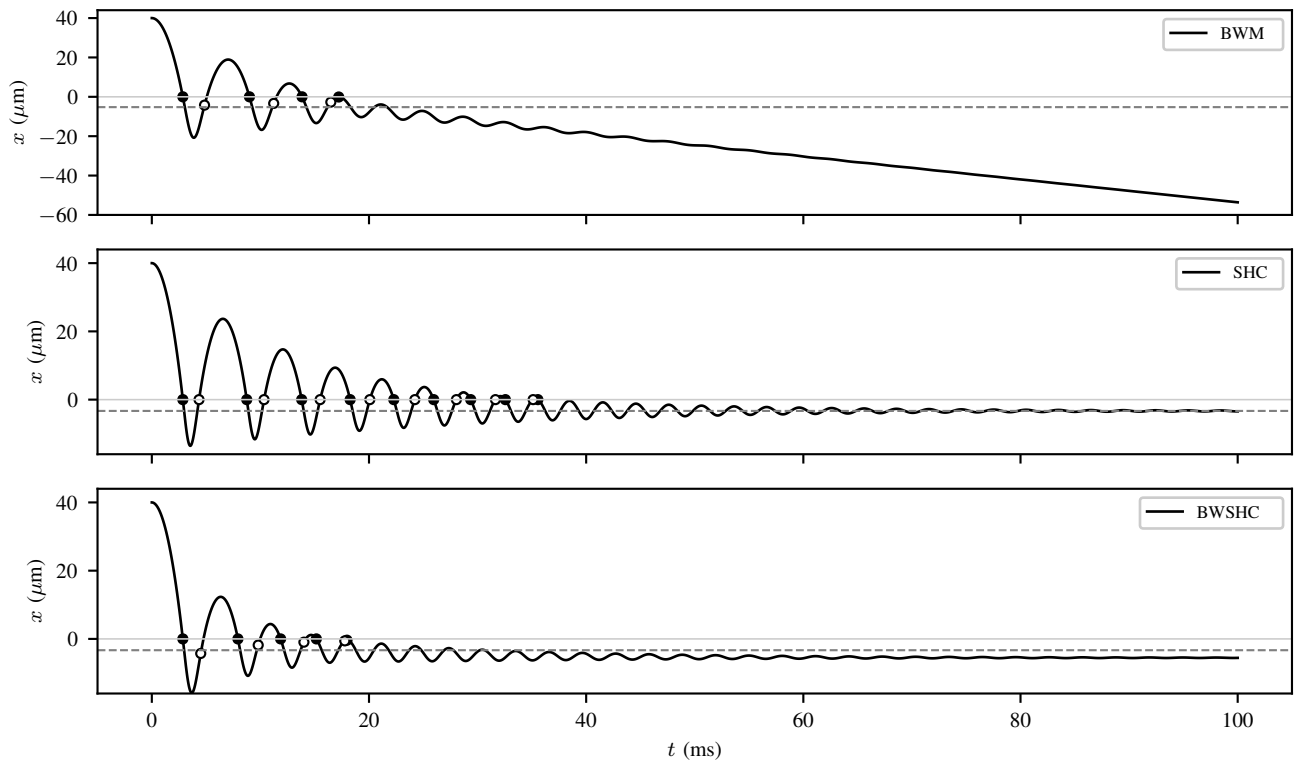}
\caption{Relative displacement of a hardened steel ball released from a height
of $40$ \textmu m above a fixed bituminous substrate, computed with three
contact laws. The filled circles are the commencements of the contacts and the
open circles the separations; the broken line is the indentation at which a
Hertzian spring of the coefficient proper to the substrate of that panel
supports the weight of the ball.}\label{fig:slow}
\end{figure*}

\begin{table}[!t]
\caption{Parameters of the third study.}\label{tab:slow}
\centering
\begin{tabular}{llll}
\hline
& quantity & value & unit \\
\hline
\multicolumn{4}{l}{\emph{Materials and geometry}} \\
$\mathcal{B}_1$ & hardened steel & & \\
& $Y_1$ & $200$ & GPa \\
& $\nu_1$ & $0.30$ & \\
& $\rho_1$ & $7800$ & kg\,m$^{-3}$ \\
& $R_1$ & $20$ & mm \\
$\mathcal{B}_2$ & glassy bitumen & & \\
& $Y_2$ & $2.0$ & GPa \\
& $\nu_2$ & $0.35$ & \\
& $R_2$ & $\infty$ & \\
\hline
\multicolumn{4}{l}{\emph{Inferred}} \\
& $Y^{*}$ & $2.256 \times 10^{9}$ & Pa \\
& $R_e$ & $20$ & mm \\
& $m$ & $2.614 \times 10^{-1}$ & kg \\
& $p$ & $3/2$ & \\
& $k_H = k$ & $4.254 \times 10^{8}$ & N\,m$^{-3/2}$ \\
\hline
\multicolumn{4}{l}{\emph{Chosen}} \\
& $A$ & $1$ & \\
& $n$ & $2$ & \\
SHCCL & $\alpha$ & $1$ & \\
& $\beta = \gamma$ & $0$ & m$^{-2}$ \\
& $c$ & $4.300 \times 10^{9}$ & N\,s\,m$^{-5/2}$ \\
BWSHCCL & $\alpha$ & $0.35$ & \\
& $\beta = \gamma$ & $5.000 \times 10^{9}$ & m$^{-2}$ \\
& $c$ & $2.150 \times 10^{9}$ & N\,s\,m$^{-5/2}$ \\
\hline
\multicolumn{4}{l}{\emph{Initial conditions and applied force}} \\
& $x_0$ & $40$ & \textmu m \\
& $v_0$ & $0$ & m\,s$^{-1}$ \\
& $g$ & $9.81$ & m\,s$^{-2}$ \\
& $F_{e,r}$ & $-mg$ & N \\
\hline
\end{tabular}
\end{table}

The third study is of the ball of the first, released from a height of $40$
\textmu m, whence the first collision occurs at $28.0$ mm\,s$^{-1}$. The system
is that of Eqs. \eqref{eq:Hr_f} to \eqref{eq:Hr_init} and the interval simulated
is $0.1$ s. Its purpose is to exhibit the passage to permanent contact. 
Three contact laws are compared: SHCCL, BWSHCCL, and BWMCL. The first two are borne by a substrate of bitumen below its
transition to the glassy state, of which the parameters are recorded in Table
\ref{tab:slow}; the third is borne by the substrate of the first study, with the
parameters of Table \ref{tab:one}, the element of Maxwell being appropriate to
the material only above that transition.

The SHCCL is recorded in Table \ref{tab:slow} as the BWSHCCL with $\alpha = 1$,
as in Example~\ref{ex:bwshc}. The indentation at which a Hertzian spring
supports the weight of the ball,
\begin{equation}\label{eq:x_static}
x_c^{s} \triangleq \left( \frac{m g}{k_H} \right)^{1/p}
\end{equation}
is the configuration of static equilibrium of the ball upon an elastic
interface, and serves as the term of comparison. It is $3.312$ \textmu m for the
glassy substrate, which bears SHCCL and BWSHCCL, and $5.240$ \textmu m for the
substrate of the first study, which bears BWMCL.

The three laws are compared in Fig. \ref{fig:slow}. In each, the contacts lengthen
and the flights shorten, and the collisions are finite in number.
What follows the last separation distinguishes the three.
Under SHCCL, the ball approaches $-x_c^{s}$ by a decaying oscillation: the contact
force is extinguished at a vanishing indentation, nothing is retained at a
separation, and the configuration of rest is that of the elastic problem,
bearing no trace of the collisions that preceded it. Under BWSHCCL, 
it approaches a configuration appreciably deeper than $-x_c^{s}$.
Under BWMCL, it comes to no rest at all, the interface
possessing no configuration in which it sustains a load without flowing; the
ball passes $-x_c^{s}$ and descends thereafter at the rate $mg/c$, without
bound.

\section{Formalization}\label{sec:formalization}

The definitions and results of Sections~\ref{sec:hybrid} to~\ref{sec:properties}
are formalized in Lean~4 \cite{de_moura_lean_2021}, upon Mathlib
\cite{the_mathlib_community_lean_2020}, both at version 4.35.0, in two
libraries: \texttt{HS}, on hybrid systems with inputs and outputs in general,
and \texttt{VEPBB}, which builds upon \texttt{HS} and is specific to this
article. Every definition, lemma, proposition, theorem, corollary, and example
of these sections has a formal counterpart; most of the remarks, and the
physical reasoning from which the model of Section~\ref{sec:model} is derived,
lie outside the formalization. The Lean code was written predominantly by the
model Claude Opus 5.5, through the agentic coding tool Claude Code, under the
close supervision of the authors.

The library \texttt{HS} formalizes Section~\ref{sec:hybrid}, and some generic
auxiliary results of Section~\ref{sec:properties}, in somewhat greater
generality: the states of a hybrid system take their values in a real normed
space in place of $\mathbb{R}^{n_s}$, and its inputs and outputs in spaces
endowed with only the structure that each result requires, such as an inner
product for passivity; likewise, assumptions such as the continuity of a supply
rate are left to the results that need them. Hybrid arcs and hybrid signals are
represented by maps upon $\mathbb{R} \times \mathbb{Z}_{\geq 0}$ whose values
outside the hybrid time domain are immaterial. The one substantive departure
concerns the existence theorem of Peano, which the proof of
Proposition~\ref{prop:completeness} invokes and which Mathlib does not provide:
the formal counterpart of the proposition assumes instead that the differential
equation of the flow has local solutions under the given input. It assumes,
moreover, that bounded sets of states have compact closures, as they do in
$\mathbb{R}^{n_s}$.

The library \texttt{VEPBB} formalizes Sections~\ref{sec:model}
and~\ref{sec:properties}: the contact laws, the systems $\mathcal{H}_c$,
$\mathcal{H}_f$, and $\mathcal{H}_r$, and the results about them. The internal
state of a law takes its values in a real normed space in place of
$\mathbb{R}^{n_z}$, subject to the same assumption upon bounded sets where
compactness is needed. The standing assumptions upon a law in
Section~\ref{sec:interface} are collected in a single hypothesis, which states
the local existence of the solutions of the initial value problem of Eq.
\eqref{eq:law} explicitly, together with their uniqueness. The hypothesis of
local solutions of the flow enters Theorem~\ref{thm:completeness},
Theorem~\ref{thm:completeness_Hr}, and Corollary~\ref{cor:existence} as it
enters Proposition~\ref{prop:completeness}. For the laws of
Examples~\ref{ex:bwshc} and~\ref{ex:bwm}, with the parameters admitted there,
all of these hypotheses are established formally, so that
Example~\ref{ex:complete} is formalized without additional hypotheses. Some of
the formal proofs take routes different from those of this article, and some of
the formal results require fewer hypotheses.

The formalization contains no unproved statement, and its results depend upon
no axioms beyond the standard axioms of Lean
\cite{lean_fro_lean_2026,carneiro_type_2019}. A separate
repository\footnote{\url{https://gitlab.com/user9716869/LEAN-CHECK}} holds
independent checks of both libraries: the tool \texttt{comparator}
\cite{lean_fro_comparator_2026,lean_fro_lean_2026} rechecks, with independent
kernels \cite{carneiro_lean4lean_2026,de_moura_lean_2026}, the proofs of the
published statements, and a semantic review compares every formal definition
and exported result with this article. The review admits any number of
reviewers, each a session of Claude Code with a model and an effort level of
its own; a library passes only if no reviewer reports a discrepancy, and
deliberately planted errors test the reviewers. The review of the libraries
reported here was conducted by two reviewers: Claude Opus 5.5, at the effort
level max, and Claude Fable 5.2, at the effort level xhigh. This semantic
review, an informal procedure devised for this work, reduces but cannot exclude
the risk of an undetected discrepancy. The libraries are available from the
repositories of the corresponding
author.\footnote{\url{https://gitlab.com/user9716869/LEAN-HS}}\footnote{\url{https://gitlab.com/user9716869/LEAN-VEPBB}}

\section{Conclusions and Future Work} \label{sec:conclusions}

The repeated direct collinear impact of two rigid bodies, each subject to a
continuous, bounded applied force, has been formulated as a hybrid system, for a
class of contact laws admitting a residual indentation. The contact interface
carries a state of its own and is coupled to the bodies through the contact
force and the relative velocity alone; the equations of motion of the bodies
accordingly contain neither switching nor resets, and the motion of the center
of mass of the system is untouched by the contact altogether. The conditions
under which a contact commences and is terminated, and the disposition of the
internal state of the interface across a separation, are constituents of
the specification rather than features of an implementation. 

Little is asked of the contact law for the formulation to be defined: its two
maps are to be continuous, and its initial value problem is to possess a unique
solution, the internal state being taken to be at rest at the commencement of
every contact. That an undeformed interface transmit no force is required in
addition, so that the contact force be continuous at the commencement of a
contact. The analytical properties ask
more: the passivity of the system rests upon that of the law, the continuity of
the contact force at the termination of a contact upon the compressiveness of
the law, and the completeness of the maximal solutions, Zeno solutions not
excluded, upon the confinement of the internal state by the storage function of
the law, the applied forces being continuous and bounded. The class so delimited is not confined
to laws that are rate-independent, nor to any one arrangement of rheological
elements, and a law enters the hybrid system only
through the two maps that constitute it, the remaining data being the same for
every law admitted.

The numerical studies exhibit two features that a treatment of impact as an
instantaneous event cannot express. The duration of the contacts grows while
that of the flights diminishes, until the bodies are out of contact for less
time than they are in contact; and the sequence of collisions terminates not by
accumulation but after finitely many collisions, the last contact persisting to
the end of the interval simulated.

The scope of the foregoing should be stated plainly. The bodies are convex and
axisymmetric, their motion is confined to the axis, and they meet at a single
point, so that no tangential force arises; the indentation is small in
comparison with the dimensions of the bodies; and every contact begins from an
undeformed interface, which is exact for a solid that recovers between contacts
and otherwise rests upon the assumption that no contact falls upon material
indented before. That assumption fails where the bodies return to contact within the overlap
left by the preceding contact, as they may near the termination of a sequence or
where an applied force draws them apart during a contact and then reverses; the
configuration at which the contact commences is then in error by at most that
overlap. The error is typically small in the first case and need not be in the second,
but it is in either case available from the computed motion, from which its
significance in a given application may be judged. More generally, the account is no more precise than the
idealization of the bodies as rigid permits: discrepancies on the scale of the
indentation lie beyond its resolution, and an application that cannot tolerate
them calls for deformable bodies. The parameters employed in the simulations are
illustrative, and were not identified against measurement. Nor does the model
guarantee that a sequence of collisions terminates.

The model of the interface is phenomenological, and its fidelity is
least assured where the incident velocities are smallest. 
Those are also the collisions least accessible
to measurement. However, little is lost by this. It is the contact law that governs the response,
and the framework alters nothing within it: wherever a collision is brief in
comparison with the interval that separates it from the next, and the contact
force great in comparison with the applied force, the account given here and the
customary one agree. They part where either condition fails: in the neighborhood of the
termination of a sequence, which is at once the regime in which the customary
account ceases to apply and the regime in which measurement is least able to
adjudicate between them, and where an applied force draws the bodies apart
during a contact, which the customary account does not address.

The natural continuation is the embedding of the contact interface in the
setting for which it is ultimately intended, namely that of a system of rigid
bodies in three dimensions, with contacts at several points and with the
tangential forces admitted. The structure of the interface commends itself to
this: it is coupled to the bodies through the contact force and the relative
velocity alone, so that it may be attached at a point of contact without
disturbing the equations of motion, and an interface carries its own state, so
that several may be carried at once, each with its own transitions, the map that
effects them acting upon one alone. The equations of motion of the bodies would
retain the form they have here, the whole of the hybrid structure residing, as
it does here, in the contact.

A second continuation is experimental. The response of an interface in the
limit of small loads, at the rates that arise in collisions, bears directly upon
the contacts near the termination of a sequence, and upon the returns within an
overlap that they may entail; in particular, whether there exist material pairs
such that a residual deformation persists in that limit is a question for
measurement.

\appendices

\section{Notation, Conventions, Foundations}\label{sec:NCF}

Essentially all of the definitions and results that are employed in this article are standard in the fields of set theory, general topology, analysis, ordinary differential equations, and nonlinear systems/control. They can be found in a number of textbooks and monographs on these subjects (e.g., see Ref. \cite{takeuti_introduction_1982}, Refs. \cite{kelley_general_2017, morris_topology_2020}, Refs. \cite{bloch_real_2010, shurman_calculus_2016, ziemer_modern_2017}, Refs. \cite{chicone_ordinary_1999, schaeffer_ordinary_2016}, Refs. \cite{lasalle_extensions_1960, yoshizawa_stability_1966, yoshizawa_stability_1975, sontag_mathematical_1998, sastry_nonlinear_1999, khalil_nonlinear_2002, haddad_nonlinear_2011}, respectively). 

\begin{definition}
$\in$ denotes the set membership relation, $\subseteq$ denotes the subset relation, $\subset$ denotes the proper subset relation, $\cup$ denotes the binary set union operation, $\cap$ denotes the binary set intersection operation, $\setminus$ denotes the binary set difference operation, $\mathcal{P}$ denotes the power set operation, $\emptyset$ denotes the empty set, $(a_1, \ldots, a_n)$ denotes an $n$-tuple, $\{ a_1, \ldots, a_n \}$ denotes an unordered collection of elements.\footnote{It should be noted that some of the syntactic constructions may carry different semantics depending on the context. For example, $(a, b)$ may be used as a pair or as an interval. It is hoped that the context of the discussion will always make the meaning of a given syntactic construction apparent.}
\end{definition}

\begin{definition}
By convention, a topological space cannot be empty. Suppose $X \neq \emptyset$ and $\tau \subseteq \mathcal{P} X$ is a topology on $X$. $\closure A$ denotes the closure of $A \subseteq X$ and $\interior A$ denotes the interior of $A \subseteq X$; if $Y \subseteq X$ and $Y \neq \emptyset$, then $\tau | Y$ will denote the subspace topology of $\tau$ on $Y$; the sets $A \subseteq X$ and $B \subseteq X$ are separated if and only if $\closure A \cap B = A \cap \closure B = \emptyset$; a set $C \subseteq X$ is clopen if and only if it is open and closed; $A \subseteq X$ is connected if and only if it is not a union of two nonempty separated sets; $(X, \tau)$ is a connected topological space if and only if $X$ is a connected set.  
\end{definition}

\begin{definition}
$\mathbb{Z}$ is the set of all integers; $\mathbb{R}$ is the set of all real numbers; an interval of real numbers $I \subseteq \mathbb{R}$ is non-degenerate if it has a non-empty interior; $\mathbb{K}_{>a} \triangleq (a, +\infty) \cap \mathbb{K}$, $\mathbb{K}_{<a} \triangleq (-\infty, a) \cap \mathbb{K}$, $\mathbb{K}_{\geq a} \triangleq [a, +\infty) \cap \mathbb{K}$, and $\mathbb{K}_{\leq a} \triangleq (-\infty, a] \cap \mathbb{K}$ for any $a \in \mathbb{R}$ with $\mathbb{K} \subseteq \mathbb{R}$; $\mathbb{R}^n$ with $n \in \mathbb{Z}_{\geq 0}$ is the set of $n$-tuples of real numbers (augmented with the structure of the Euclidean space), the set $\mathbb{R}^0$ consisting of the empty tuple $()$ alone; $\mathbb{R}$ is identified with $\mathbb{R}^1$, the real number $a$ being identified with the $1$-tuple $(a)$; for all $n \in \mathbb{Z}_{\geq 0}$, $0$ also denotes the zero element of $\mathbb{R}^n$; given $m \in \mathbb{Z}_{\geq 1}$ and $n_1, \ldots, n_m \in \mathbb{Z}_{\geq 0}$, the set $\mathbb{R}^{n_1} \times \cdots \times \mathbb{R}^{n_m}$ is identified with $\mathbb{R}^{n_1 + \cdots + n_m}$ by the concatenation of tuples, the element $(a_1, \ldots, a_m)$ with $a_k = (a_{k, 1}, \ldots, a_{k, n_k}) \in \mathbb{R}^{n_k}$ for all $k \in \{ 1, \ldots, m \}$ being identified with $(a_{1, 1}, \ldots, a_{1, n_1}, \ldots, a_{m, 1}, \ldots, a_{m, n_m})$; if $X = (x_1, \ldots, x_n) \in \mathbb{R}^n$ with $n \in \mathbb{Z}_{\geq 1}$, then $X_i \triangleq x_i$ for all $i \in \{ 1, \ldots, n \}$, and $X_{i:j} \triangleq (x_i, \ldots, x_j) \in \mathbb{R}^{j - i + 1}$ for all $i \in \{ 1, \ldots, n + 1 \}$ and $j \in \{ i - 1, \ldots, n \}$, the range $i{:}j$ including both of its bounds and $X_{i:(i - 1)}$ being the empty tuple; $f : X \longrightarrow Y$ denotes a function with the domain $X$ and the codomain $Y$; $\dom f$ denotes the domain of $f$; given $f : X \longrightarrow Y$ and $A \subseteq X$, $f[A]$ denotes the image of $A$ under $f$; if $f : X \longrightarrow \mathbb{R}^n$ with $n \in \mathbb{Z}_{\geq 1}$, then $f_i : X \longrightarrow \mathbb{R}$ is given by $f_i(x) \triangleq (f(x))_i$ for all $x \in X$ and $i \in \{ 1, \ldots, n \}$, and $f_{i:j} : X \longrightarrow \mathbb{R}^{j - i + 1}$ is given by $f_{i:j}(x) \triangleq (f(x))_{i:j}$ for all $x \in X$, $i \in \{ 1, \ldots, n + 1 \}$, and $j \in \{ i - 1, \ldots, n \}$; unless stated otherwise, the topology of a subset of $\mathbb{R}^n$ with $n \in \mathbb{Z}_{\geq 1}$ is always the subspace topology of the standard topology on $\mathbb{R}^n$; given $A \subseteq \mathbb{R}$, $\inf A \in \mathbb{R} \cup \{ -\infty, +\infty \}$ denotes the infimum of $A$ and $\sup A \in \mathbb{R} \cup \{ -\infty, +\infty \}$ denotes the supremum of $A$; given a sequence $\{ x_i \in \mathbb{R}^n \}_{i \in \mathbb{Z}_{\geq 1}}$ with $n \in \mathbb{Z}_{\geq 1}$, $\lim_{i \rightarrow +\infty} x_i$ denotes the limit of the sequence, provided that it exists; 
given $n \in \mathbb{Z}_{\geq 1}$, $X \subseteq \mathbb{R}$, and
$f : X \longrightarrow \mathbb{R}^n$, $\lim_{x \rightarrow a} f(x)$ denotes the
limit of $f$ at a limit point $a \in \mathbb{R}$ of $X$,
$\lim_{x \rightarrow a^{-}} f(x)$ and $\lim_{x \rightarrow a^{+}} f(x)$ denote
the limits at $a$ of the restrictions of $f$ to $X_{<a}$ and to $X_{>a}$,
respectively, and $\lim_{x \rightarrow +\infty} f(x)$ and
$\lim_{x \rightarrow -\infty} f(x)$ denote the limits of $f$ as $x$ tends to
$+\infty$ and to $-\infty$, respectively, each provided that it exists, an
equality of any of these limits to $+\infty$ or to $-\infty$ indicating, for
$n = 1$, that $f$ diverges to $+\infty$ or to $-\infty$, respectively, in the
corresponding sense, and $f(x) \rightarrow L$ as $x \rightarrow a$ meaning that
$\lim_{x \rightarrow a} f(x) = L$, and likewise with $a^{-}$, $a^{+}$,
$+\infty$, or $-\infty$ in place of $a$;
given $a, b \in \mathbb{R}$ with $a \leq b$ and an integrable function $f : [a, b] \longrightarrow \mathbb{R}$, $\int_a^b f(t) \, \mathrm{d}t$ denotes the integral of $f$ upon $[a, b]$; $\langle \cdot, \cdot \rangle : \mathbb{R}^n \times \mathbb{R}^n \longrightarrow \mathbb{R}$ with $n \in \mathbb{Z}_{\geq 0}$ is the canonical inner product on $\mathbb{R}^n$; $\abs{\cdot} : \mathbb{R}^n \longrightarrow \mathbb{R}_{\geq 0}$ with $n \in \mathbb{Z}_{\geq 0}$ is the Euclidean norm on $\mathbb{R}^n$, which coincides with the absolute value for $n = 1$; the inner product and the norm upon $\mathbb{R}^0$ vanish identically; under the identification described above, $\langle (a_1, \ldots, a_m), (b_1, \ldots, b_m) \rangle = \sum_{k = 1}^{m} \langle a_k, b_k \rangle$ and $\abs{(a_1, \ldots, a_m)}^2 = \sum_{k = 1}^{m} \abs{a_k}^2$ for all $m \in \mathbb{Z}_{\geq 1}$, $n_1, \ldots, n_m \in \mathbb{Z}_{\geq 0}$, and $a_k, b_k \in \mathbb{R}^{n_k}$ with $k \in \{ 1, \ldots, m \}$; given a differentiable function $f : X \longrightarrow Y$ such that $X \subseteq \mathbb{R}$ and $Y \subseteq \mathbb{R}^n$ with $n \in \mathbb{Z}_{\geq 1}$, $\mathrm{d}f/\mathrm{d}x$ and $\partial f$ may be used to denote the derivative of the function; the overdot notation $\dot{x} \triangleq (\mathrm{d}x/\mathrm{d}t)$ may be used to represent the derivative of a differentiable function $x : X \longrightarrow \mathbb{R}^n$ with $n \in \mathbb{Z}_{\geq 1}$ and $X \subseteq \mathbb{R}$ with respect to the time variable in the context of mechanics; given a differentiable function $f : X \longrightarrow Y$ such that $X \subseteq \mathbb{R}^n$ and $Y \subseteq \mathbb{R}$ with $n \in \mathbb{Z}_{\geq 1}$, $\partial_i f$ denotes the $i$-th partial derivative of the function for $i \in \{ 1, \ldots, n \}$, and $\nabla f : X \longrightarrow \mathbb{R}^n$ denotes its gradient, given by $(\nabla f(x))_i \triangleq \partial_i f(x)$ for all $x \in X$ and $i \in \{ 1, \ldots, n \}$; $\mathcal{K}_\infty$ denotes the set of all continuous strictly increasing maps $\alpha : \mathbb{R}_{\geq 0} \longrightarrow \mathbb{R}_{\geq 0}$ such that $\alpha(0) = 0$ and $\alpha(r) \rightarrow +\infty$ as $r \rightarrow +\infty$.
\end{definition}

\section*{Use of Generative Artificial Intelligence (AI)}

During the preparation of this article, the authors utilized generative artificial
intelligence to assist in the theoretical, computational, and textual
development of the research. The initial version of the manuscript and the code
was developed primarily with Anthropic Claude Opus 5.0 and, to a lesser extent, Google Gemini 3.1 Pro; the present revision, with Anthropic Claude Opus 5.5 and Anthropic Claude Fable 5.2. 
The utilization functioned as a continuous,
human-in-the-loop collaborative process rather than a passive delegation of
tasks; the final content emerged through iterative cycles of prompting,
critical evaluation, and refinement by the authors. In the theoretical domain,
the AI was employed to search and summarize literature, aid in the translation
and contextual interpretation of historical, non-English texts, formulate
intermediate mathematical definitions and lemmas, and assist in proof
development. In the computational domain, Claude Opus 5 helped draft software
architecture specifications and generate code based on them; in the revision,
Claude Opus 5.5 helped specify unit and integration tests and review their
outcomes, and Claude Fable 5.1 implemented them together with the resulting
amendments to the code. Furthermore, the AI assisted in manuscript preparation.
This included drafting specific paragraphs from author-provided specifications,
generating and refining complex LaTeX environments, performing general
copyediting to elevate the academic prose, and, in the revision, reviewing the
manuscript for errors and inconsistencies. The utilization extended further, to
conceptual development: several principal ideas arose in exchanges with the AI,
the authors being unable to attribute each to a prompt or a response. Because
all conceptual, mathematical, computational, and textual elements were actively
co-created and continuously verified by the authors throughout the
aforedescribed collaborative workflow, the authors assume full responsibility
for the content of the publication.

\section*{Acknowledgment}

The authors would like to acknowledge their families, colleagues, and friends. Special thanks go to the members of staff of Auburn University Libraries for their assistance in finding rare and out-of-print research articles and research monographs. The authors would also like to acknowledge the professional online communities, instructional websites, and various online service providers, especially \url{https://www.adobe.com/acrobat/online/pdf-to-word.html}, \url{https://archive.org/}, \url{https://capitalizemytitle.com}, \url{https://www.overleaf.com}, \url{https://pgfplots.net}, \url{https://proofwiki.org/}, \url{https://www.reddit.com}, \url{https://scholar.google.com}, \url{https://stackexchange.com}, \url{https://stringtranslate.com}, \url{https://www.wikipedia.org}. We also note that some of the calculations that are presented in this article were performed with the assistance of the software Wolfram Mathematica \cite{wolfram_research_inc_mathematica_2023}. Other software that was used to produce this article included Adobe Acrobat Reader, Adobe Digital Editions, Anthropic Claude, DiffMerge, Git, GitLab, Google Chrome, Grammarly, Jupyter Notebook, LibreOffice, macOS Tahoe, Mamba, Microsoft Outlook, Preview, Safari, TeX Live/MacTeX, Texmaker, and Zotero.

\bibliography{template.bib}

@article{bouc_modemathematique_1971,
	title = {Mod{\`e}le {Math{\'e}matique} d{\textquoteright}{Hyst{\'e}r{\'e}sis}},
	volume = {24},
	number = {1},
	journal = {Acustica},
	author = {Bouc, R.},
	year = {1971},
	pages = {16--25},
}

@article{milehins_boucwen_2025,
	title = {The {Bouc}{\textendash}{Wen} {Model} for {Binary} {Direct} {Collinear} {Collisions} of {Convex} {Viscoplastic} {Bodies}},
	volume = {20},
	doi = {10.1115/1.4068158},
	number = {6},
	journal = {ASME Journal of Computational and Nonlinear Dynamics},
	author = {Milehins, Mihails and Marghitu, Dan B.},
	year = {2025},
	pages = {061005},
}

@book{hairer_solving_1993,
	address = {Berlin, The Federal Republic of Germany},
	edition = {2},
	series = {Springer {Series} in {Computational} {Mathematics}},
	title = {Solving {Ordinary} {Differential} {Equations} {I}: {Nonstiff} {Problems}},
	volume = {8},
	isbn = {978-3-540-56670-0},
	publisher = {Springer},
	author = {Hairer, Ernst and N{\o}rsett, Syvert P. and Wanner, Gerhard},
	year = {1993},
}

@book{pfeiffer_multibody_2004,
	address = {Weinheim, The Federal Republic of Germany},
	series = {Wiley {Series} in {Nonlinear} {Science}},
	title = {Multibody {Dynamics} {With} {Unilateral} {Contacts}},
	isbn = {978-0-471-15565-2},
	publisher = {WILEY-VCH Verlag GmbH \& Co. KGaA},
	author = {Pfeiffer, Friedrich and Glocker, Christoph},
	year = {2004},
}

@book{ziemer_modern_2017,
	address = {Cham, The Swiss Confederation},
	edition = {2},
	series = {Graduate {Texts} in {Mathematics}},
	title = {Modern {Real} {Analysis}},
	volume = {278},
	isbn = {978-3-319-64628-2},
	publisher = {Springer},
	author = {Ziemer, William P. and Torres, Monica},
	year = {2017},
}

@book{shurman_calculus_2016,
	address = {Cham, The Swiss Confederation},
	series = {Undergraduate {Texts} in {Mathematics}},
	title = {Calculus and {Analysis} in {Euclidean} {Space}},
	isbn = {978-3-319-49314-5},
	publisher = {Springer},
	author = {Shurman, Jerry},
	year = {2016},
}

@book{tabor_hardness_1951,
	address = {Oxford, The United Kingdom of Great Britain and Northern Ireland},
	series = {Monographs on the {Physics} and {Chemistry} of {Materials}},
	title = {The {Hardness} of {Metals}},
	publisher = {Oxford University Press},
	author = {Tabor, David},
	year = {1951},
}

@book{stronge_impact_2018,
	address = {Cambridge, The United Kingdom of Great Britain and Northern Ireland},
	edition = {2},
	title = {Impact {Mechanics}},
	isbn = {978-0-521-84188-7},
	publisher = {Cambridge University Press},
	author = {Stronge, William James},
	year = {2018},
}

@book{roithmayr_dynamics_2016,
	address = {New York, NY},
	title = {Dynamics: {Theory} and {Application} of {Kane}{\textquoteright}s {Method}},
	isbn = {978-1-107-00569-3},
	publisher = {Cambridge University Press},
	author = {Roithmayr, Carlos M. and Hodges, Dewey H.},
	year = {2016},
}

@book{ikhouane_systems_2007,
	address = {Chichester, The United Kingdom of Great Britain and Northern Ireland},
	title = {Systems {With} {Hysteresis}: {Analysis}, {Identification} and {Control} {Using} the {Bouc}-{Wen} {Model}},
	isbn = {978-0-470-03236-7},
	publisher = {John Wiley \& Sons},
	author = {Ikhouane, Fay{\c c}al and Rodellar, Jos{\'e}},
	year = {2007},
}

@book{goldsmith_impact_1960,
	address = {London, The United Kingdom of Great Britain and Northern Ireland},
	title = {Impact: {The} {Theory} and {Physical} {Behaviour} of {Colliding} {Solids}},
	publisher = {Edward Arnold},
	author = {Goldsmith, Werner},
	year = {1960},
}

@inproceedings{ning_elastic-plastic_1993,
	address = {Rotterdam, The Netherlands},
	title = {Elastic-{Plastic} {Impact} of {Fine} {Particles} {With} a {Surface}},
	booktitle = {Proceedings of the {Second} {International} {Conference} on {Micromechanics} of {Granular} {Media}: {Powders} \& {Grains} 93},
	publisher = {A. A. Balkema},
	author = {Ning, Z and Thornton, C},
	editor = {Thornton, C.},
	year = {1993},
	pages = {33--38},
}

@book{yoshizawa_stability_1975,
	address = {New York, NY},
	series = {Applied {Mathematical} {Sciences}},
	title = {Stability {Theory} and the {Existence} of {Periodic} {Solutions} and {Almost} {Periodic} {Solutions}},
	volume = {14},
	isbn = {978-1-4612-6376-0},
	publisher = {Springer},
	author = {Yoshizawa, T.},
	year = {1975},
}

@article{barnhart_stresses_1957,
	title = {Stresses in {Beams} {During} {Transverse} {Impact}},
	volume = {24},
	doi = {10.1115/1.4011560},
	number = {3},
	journal = {ASME Journal of Applied Mechanics},
	author = {Barnhart, K. E. and Goldsmith, Werner},
	year = {1957},
	pages = {440--446},
}

@article{hunt_coefficient_1975,
	title = {Coefficient of {Restitution} {Interpreted} as {Damping} in {Vibroimpact}},
	volume = {42},
	doi = {10.1115/1.3423596},
	number = {2},
	journal = {ASME Journal of Applied Mechanics},
	author = {Hunt, K. H. and Crossley, F. R. E.},
	year = {1975},
	pages = {440--445},
}

@article{chatterjee_two_1998,
	title = {Two {Interpretations} of {Rigidity} in {Rigid}-{Body} {Collisions}},
	volume = {65},
	doi = {10.1115/1.2791929},
	number = {4},
	journal = {ASME Journal of Applied Mechanics},
	author = {Chatterjee, A. and Ruina, A.},
	year = {1998},
	pages = {894--900},
}

@article{butcher_characterizing_2000,
	title = {Characterizing {Damping} and {Restitution} in {Compliant} {Impacts} via {Modified} {K}-{V} and {Higher}-{Order} {Linear} {Viscoelastic} {Models}},
	volume = {67},
	doi = {10.1115/1.1308578},
	number = {4},
	journal = {ASME Journal of Applied Mechanics},
	author = {Butcher, E. A. and Segalman, D. J.},
	year = {2000},
	pages = {831--834},
}

@article{ismail_impact_2008,
	title = {Impact of {Viscoplastic} {Bodies}: {Dissipation} and {Restitution}},
	volume = {75},
	doi = {10.1115/1.2965371},
	number = {6},
	journal = {ASME Journal of Applied Mechanics},
	author = {Ismail, K. A. and Stronge, W. J.},
	year = {2008},
	pages = {061011},
}

@article{xiong_contact_2014,
	title = {A {Contact} {Force} {Model} {With} {Nonlinear} {Compliance} and {Residual} {Indentation}},
	volume = {81},
	doi = {10.1115/1.4024403},
	number = {2},
	journal = {ASME Journal of Applied Mechanics},
	author = {Xiong, Xiaogang and Kikuuwe, Ryo and Yamamoto, Motoji},
	year = {2014},
	pages = {021003},
}

@inproceedings{kikuuwe_incorporating_2007,
	address = {Los Alamitos, CA},
	title = {Incorporating {Geometric} {Algorithms} in {Impedance}- and {Admittance}-{Type} {Haptic} {Rendering}},
	isbn = {978-0-7695-2738-3},
	doi = {10.1109/WHC.2007.75},
	booktitle = {Proceedings of the {Second} {Joint} {EuroHaptics} {Conference} and {Symposium} on {Haptic} {Interfaces} for {Virtual} {Environment} and {Teleoperator} {Systems} ({WHC}'07)},
	publisher = {IEEE Computer Society},
	author = {Kikuuwe, Ryo and Fujimoto, Hideo},
	year = {2007},
	pages = {249--254},
}

@article{harris_array_2020,
	title = {Array {Programming} {With} {NumPy}},
	volume = {585},
	doi = {10.1038/s41586-020-2649-2},
	number = {7825},
	journal = {Nature},
	author = {Harris, Charles R. and Millman, K. Jarrod and van der Walt, St{\'e}fan J. and Gommers, Ralf and Virtanen, Pauli and Cournapeau, David and Wieser, Eric and Taylor, Julian and Berg, Sebastian and Smith, Nathaniel J. and Kern, Robert and Picus, Matti and Hoyer, Stephan and van Kerkwijk, Marten H. and Brett, Matthew and Haldane, Allan and Fern{\'a}ndez del R{\'i}o, Jaime and Wiebe, Mark and Peterson, Pearu and G{\'e}rard-Marchant, Pierre and Sheppard, Kevin and Reddy, Tyler and Weckesser, Warren and Abbasi, Hameer and Gohlke, Christoph and Oliphant, Travis E.},
	year = {2020},
	pages = {357--362},
}

@article{nikravesh_determination_2023,
	title = {Determination of {Effective} {Mass} for {Continuous} {Contact} {Models} in {Multibody} {Dynamics}},
	volume = {58},
	doi = {10.1007/s11044-022-09859-4},
	number = {3-4},
	journal = {Multibody System Dynamics},
	author = {Nikravesh, Parviz E. and Poursina, Mohammad},
	year = {2023},
	pages = {253--273},
}

@article{zhang_continuous_2024,
	title = {A {Continuous} {Contact}-{Force} {Model} for the {Impact} {Analysis} of {Viscoelastic} {Materials} {With} {Elastic} {Aftereffect}},
	volume = {61},
	doi = {10.1007/s11044-023-09954-0},
	number = {3},
	journal = {Multibody System Dynamics},
	author = {Zhang, Yifei and Ding, Yong and Xu, Guoshan},
	year = {2024},
	pages = {435--451},
}

@book{morris_topology_2020,
	title = {Topology {Without} {Tears}},
	publisher = {Sidney A. Morris},
	author = {Morris, Sidney A.},
	year = {2020},
}

@article{virtanen_scipy_2020,
	title = {Scipy 1.0: {Fundamental} {Algorithms} for {Scientific} {Computing} in {Python}},
	volume = {17},
	doi = {10.1038/s41592-019-0686-2},
	number = {3},
	journal = {Nature Methods},
	author = {Virtanen, Pauli and Gommers, Ralf and Oliphant, Travis E. and Haberland, Matt and Reddy, Tyler and Cournapeau, David and Burovski, Evgeni and Peterson, Pearu and Weckesser, Warren and Bright, Jonathan and van der Walt, St{\'e}fan J. and Brett, Matthew and Wilson, Joshua and Millman, K. Jarrod and Mayorov, Nikolay and Nelson, Andrew R. J. and Jones, Eric and Kern, Robert and Larson, Eric and Carey, C J and Polat, Ilhan and Feng, Yu and Moore, Eric W. and VanderPlas, Jake and Laxalde, Denis and Perktold, Josef and Cimrman, Robert and Henriksen, Ian and Quintero, E. A. and Harris, Charles R. and Archibald, Anne M. and Ribeiro, Ant{\^o}nio H. and Pedregosa, Fabian and van Mulbregt, Paul and {SciPy 1.0 Contributors}},
	year = {2020},
	pages = {261--272},
}

@article{zhang_continuous_2022,
	title = {A {Continuous} {Contact} {Force} {Model} for the {Impact} {Analysis} of {Hard} and {Soft} {Materials}},
	volume = {177},
	doi = {10.1016/j.mechmachtheory.2022.105065},
	journal = {Mechanism and Machine Theory},
	author = {Zhang, Jie and Fang, Mingyang and Zhao, Lei and Zhao, Quanliang and Liang, Xu and He, Guangping},
	year = {2022},
	pages = {105065},
}

@article{wen_method_1976,
	title = {Method for {Random} {Vibration} of {Hysteretic} {Systems}},
	volume = {102},
	doi = {10.1061/JMCEA3.0002106},
	number = {2},
	journal = {Journal of the Engineering Mechanics Division},
	author = {Wen, Yi-Kwei},
	year = {1976},
	pages = {249--263},
}

@article{erlicher_thermodynamic_2004,
	title = {Thermodynamic {Admissibility} of {Bouc}{\textendash}{Wen} {Type} {Hysteresis} {Models}},
	volume = {332},
	doi = {10.1016/j.crme.2003.10.009},
	number = {1},
	journal = {Comptes Rendus M{\'e}canique},
	author = {Erlicher, Silvano and Point, Nelly},
	year = {2004},
	pages = {51--57},
}

@book{haupt_continuum_2000,
	address = {Berlin, The Federal Republic of Germany},
	series = {Advanced {Texts} in {Physics}},
	title = {Continuum {Mechanics} and {Theory} of {Materials}},
	isbn = {978-3-662-04109-3},
	publisher = {Springer},
	author = {Haupt, Peter},
	translator = {Kurth, Joan A.},
	year = {2000},
}

@misc{wolfram_research_inc_mathematica_2023,
	address = {Champaign, IL},
	title = {Mathematica, {Version} 13.3},
	url = {https://www.wolfram.com/mathematica},
	author = {{Wolfram Research Inc}},
	year = {2023},
}

@book{khalil_nonlinear_2002,
	address = {Upper Saddle River, NJ},
	title = {Nonlinear {Systems}},
	isbn = {978-0-13-228024-2},
	publisher = {Prentice Hall},
	author = {Khalil, Hassan K.},
	year = {2002},
}

@article{khulief_continuous_1987,
	title = {A {Continuous} {Force} {Model} for the {Impact} {Analysis} of {Flexible} {Multibody} {Systems}},
	volume = {22},
	doi = {10.1016/0094-114X(87)90004-8},
	number = {3},
	journal = {Mechanism and Machine Theory},
	author = {Khulief, Y. A and Shabana, A. A},
	year = {1987},
	pages = {213--224},
}

@article{crook_study_1952,
	title = {A {Study} of {Some} {Impacts} {Between} {Metal} {Bodies} by a {Piezo}-{Electric} {Method}},
	volume = {212},
	doi = {10.1098/rspa.1952.0088},
	number = {1110},
	journal = {Proceedings of the Royal Society of London. Series A. Mathematical and Physical Sciences},
	publisher = {Royal Society},
	author = {Crook, A. W.},
	year = {1952},
	pages = {377--390},
}

@article{lasalle_extensions_1960,
	title = {Some {Extensions} of {Liapunov}'s {Second} {Method}},
	volume = {7},
	doi = {10.1109/TCT.1960.1086720},
	number = {4},
	journal = {IRE Transactions on Circuit Theory},
	author = {LaSalle, J. P.},
	year = {1960},
	pages = {520--527},
}

@article{sadd_contact_1993,
	title = {Contact {Law} {Effects} on {Wave} {Propagation} in {Particulate} {Materials} {Using} {Distinct} {Element} {Modeling}},
	volume = {28},
	doi = {10.1016/0020-7462(93)90061-O},
	number = {2},
	journal = {International Journal of Non-Linear Mechanics},
	author = {Sadd, Martin H. and Tai, QiMing and Shukla, Arun},
	year = {1993},
	pages = {251--265},
}

@article{xiong_differential_2013,
	title = {A {Differential} {Algebraic} {Method} to {Approximate} {Nonsmooth} {Mechanical} {Systems} by {Ordinary} {Differential} {Equations}},
	volume = {2013},
	doi = {10.1155/2013/320276},
	journal = {Journal of Applied Mathematics},
	author = {Xiong, Xiaogang and Kikuuwe, Ryo and Yamamoto, Motoji},
	year = {2013},
	pages = {1--13},
}

@article{walton_viscosity_1986,
	title = {Viscosity, {Granular}-{Temperature}, and {Stress} {Calculations} for {Shearing} {Assemblies} of {Inelastic}, {Frictional} {Disks}},
	volume = {30},
	doi = {10.1122/1.549893},
	number = {5},
	journal = {Journal of Rheology},
	author = {Walton, Otis R. and Braun, Robert L.},
	year = {1986},
	pages = {949--980},
}

@article{vu-quoc_elastoplastic_1999,
	title = {An {Elastoplastic} {Contact} {Force}{\textendash}{Displacement} {Model} in the {Normal} {Direction}: {Displacement}{\textendash}{Driven} {Version}},
	volume = {455},
	doi = {10.1098/rspa.1999.0488},
	number = {1991},
	journal = {Proceedings of the Royal Society of London. Series A: Mathematical, Physical and Engineering Sciences},
	author = {Vu-Quoc, Loc and Zhang, Xiang},
	year = {1999},
	pages = {4013--4044},
}

@article{luding_cohesive_2008,
	title = {Cohesive, {Frictional} {Powders}: {Contact} {Models} for {Tension}},
	volume = {10},
	doi = {10.1007/s10035-008-0099-x},
	number = {4},
	journal = {Granular Matter},
	author = {Luding, Stefan},
	year = {2008},
	pages = {235--246},
}

@article{schwager_coefficient_2008,
	title = {Coefficient of {Restitution} for {Viscoelastic} {Spheres}: {The} {Effect} of {Delayed} {Recovery}},
	volume = {78},
	doi = {10.1103/PhysRevE.78.051304},
	number = {5},
	journal = {Physical Review E},
	author = {Schwager, Thomas and P{\"o}schel, Thorsten},
	year = {2008},
	pages = {051304},
}

@article{machado_compliant_2012,
	title = {Compliant {Contact} {Force} {Models} in {Multibody} {Dynamics}: {Evolution} of the {Hertz} {Contact} {Theory}},
	volume = {53},
	doi = {10.1016/j.mechmachtheory.2012.02.010},
	journal = {Mechanism and Machine Theory},
	author = {Machado, Margarida and Moreira, Pedro and Flores, Paulo and Lankarani, Hamid M.},
	year = {2012},
	pages = {99--121},
}

@article{wang_advanced_2017,
	title = {Advanced {Impact} {Force} {Model} for {Low}-{Speed} {Pounding} {Between} {Viscoelastic} {Materials} and {Steel}},
	volume = {143},
	doi = {10.1061/(ASCE)EM.1943-7889.0001372},
	number = {12},
	journal = {Journal of Engineering Mechanics},
	author = {Wang, Wenxi and Hua, Xugang and Wang, Xiuyong and Chen, Zhengqing and Song, Gangbing},
	year = {2017},
	pages = {04017139},
}

@book{schaeffer_ordinary_2016,
	address = {New York, NY},
	series = {Texts in {Applied} {Mathematics}},
	title = {Ordinary {Differential} {Equations}: {Basics} and {Beyond}},
	volume = {65},
	isbn = {978-1-4939-6389-8},
	publisher = {Springer},
	author = {Schaeffer, David G. and Cain, John W.},
	year = {2016},
}

@book{sastry_nonlinear_1999,
	address = {New York, NY},
	series = {Interdisciplinary {Applied} {Mathematics}},
	title = {Nonlinear {Systems}: {Analysis}, {Stability}, and {Control}},
	volume = {10},
	isbn = {978-1-4419-3132-0},
	publisher = {Springer},
	author = {Sastry, Shankar},
	year = {1999},
}

@book{takeuti_introduction_1982,
	address = {New York, NY},
	edition = {2},
	series = {Graduate {Texts} in {Mathematics}},
	title = {Introduction to {Axiomatic} {Set} {Theory}},
	volume = {1},
	isbn = {978-1-4613-8168-6},
	publisher = {Springer},
	author = {Takeuti, Gaisi and Zaring, Wilson M.},
	year = {1982},
}

@book{yoshizawa_stability_1966,
	address = {Tokyo, Japan},
	series = {Publications of the {Mathematical} {Society} of {Japan}},
	title = {Stability {Theory} by {Liapunov}'s {Second} {Method}},
	number = {9},
	publisher = {The Mathematical Society of Japan},
	author = {Yoshizawa, T.},
	year = {1966},
}

@misc{ieee_ieee_2019,
	address = {New York, NY},
	title = {{IEEE} {Standard} for {Floating}-{Point} {Arithmetic}. {IEEE} {Std} {754TM}-2019 ({Revision} of {IEEE} {Std} 754-2008)},
	doi = {10.1109/IEEESTD.2019.8766229},
	publisher = {IEEE},
	author = {{IEEE}},
	year = {2019},
}

@book{newton_mathematical_1729,
	address = {London, The Kingdom of Great Britain},
	title = {The {Mathematical} {Principles} of {Natural} {Philosophy}},
	publisher = {Benjamin Motte},
	author = {Newton, Isaac},
	translator = {Motte, Andrew},
	year = {1729},
	note = {[Translated by A. Motte, originally published as Philosophi{\ae} Naturalis Principia Mathematica (Joseph Streater, London, The Kingdom of England, 1687)]},
}

@book{kelley_general_2017,
	address = {Mineola, NY},
	title = {General {Topology}},
	publisher = {Dover Publications Inc},
	author = {Kelley, John L.},
	year = {2017},
	note = {[Originally published as General Topology, D. Van Nostrand Company, New York, NY, 1955]},
}

@book{panagiotopoulos_inequality_1985,
	address = {Boston, MA},
	title = {Inequality {Problems} in {Mechanics} and {Applications}: {Convex} and {Nonconvex} {Energy} {Functions}},
	isbn = {978-3-7643-3094-5},
	publisher = {Birkh{\"a}user Boston},
	author = {Panagiotopoulos, P. D.},
	year = {1985},
}

@book{sontag_mathematical_1998,
	address = {New York, NY},
	edition = {2},
	series = {Texts in {Applied} {Mathematics}},
	title = {Mathematical {Control} {Theory}: {Deterministic} {Finite} {Dimensional} {Systems}},
	volume = {6},
	isbn = {978-0-387-98489-6},
	publisher = {Springer},
	author = {Sontag, Eduardo D.},
	year = {1998},
}

@book{bloch_real_2010,
	address = {New York, NY},
	title = {The {Real} {Numbers} and {Real} {Analysis}},
	isbn = {978-0-387-72176-7},
	publisher = {Springer},
	author = {Bloch, Ethan D.},
	year = {2010},
}

@book{haddad_nonlinear_2011,
	address = {Princeton, NJ},
	title = {Nonlinear {Dynamical} {Systems} and {Control}: {A} {Lyapunov}-{Based} {Approach}},
	isbn = {978-1-4008-4104-2},
	publisher = {Princeton University Press},
	author = {Haddad, Wassim M. and Chellaboina, VijaySekhar},
	year = {2011},
}

@book{stewart_dynamics_2011,
	address = {Philadelphia, PA},
	title = {Dynamics {With} {Inequalities}: {Impacts} and {Hard} {Constraints}},
	isbn = {978-1-61197-070-8},
	publisher = {Society for Industrial and Applied Mathematics},
	author = {Stewart, David E.},
	year = {2011},
}

@book{goebel_hybrid_2012,
	address = {Princeton, NJ},
	title = {Hybrid {Dynamical} {Systems}: {Modeling}, {Stability}, and {Robustness}},
	isbn = {978-1-4008-4263-6},
	publisher = {Princeton University Press},
	author = {Goebel, Rafal and Sanfelice, Ricardo G. and Teel, Andrew R.},
	year = {2012},
}

@book{sanfelice_hybrid_2021,
	address = {Princeton, NJ},
	title = {Hybrid {Feedback} {Control}},
	isbn = {978-0-691-18022-9},
	publisher = {Princeton University Press},
	author = {Sanfelice, Ricardo G.},
	year = {2021},
}

@inproceedings{bouc_forced_1968,
	address = {Prague, Czechoslovakia},
	title = {Forced {Vibration} of {Mechanical} {Systems} {With} {Hysteresis}},
	booktitle = {Proceedings of the 4th {Conference} on {Nonlinear} {Oscillations}},
	publisher = {Academia Publishing House of the Czechoslovak Academy of Sciences},
	author = {Bouc, R.},
	editor = {Gonda, J{\'a}n and Jel{\'i}nek, Franti{\v s}ek},
	year = {1968},
	pages = {315},
}

@article{andrews_theory_1930,
	title = {Theory of {Collision} of {Spheres} of {Soft} {Metals}},
	volume = {9},
	issn = {1941-5982},
	doi = {10.1080/14786443008565033},
	number = {58},
	journal = {The London, Edinburgh, and Dublin Philosophical Magazine and Journal of Science},
	author = {Andrews, J. P.},
	year = {1930},
	pages = {593--610},
}

@phdthesis{barnhart_transverse_1955,
	address = {Berkeley, CA},
	type = {Ph. {D}. thesis},
	title = {Transverse {Impact} on {Elastically} {Supported} {Beams}},
	school = {University of California, Berkeley},
	author = {Barnhart, Kenneth Edwin},
	year = {1955},
}

@article{cundall_discrete_1979,
	title = {A {Discrete} {Numerical} {Model} for {Granular} {Assemblies}},
	volume = {29},
	doi = {10.1680/geot.1979.29.1.47},
	number = {1},
	journal = {G{\'e}otechnique},
	author = {Cundall, P. A. and Strack, O. D. L.},
	year = {1979},
	pages = {47--65},
}

@article{tabor_simple_1948,
	title = {A {Simple} {Theory} of {Static} and {Dynamic} {Hardness}},
	volume = {192},
	doi = {10.1098/rspa.1948.0008},
	number = {1029},
	journal = {Proceedings of the Royal Society of London. Series A. Mathematical and Physical Sciences},
	author = {Tabor, David},
	year = {1948},
	pages = {247--274},
}

@phdthesis{chatterjee_rigid_1997,
	address = {Ithaca, NY},
	type = {Ph. {D}. thesis},
	title = {Rigid {Body} {Collisions}: {Some} {General} {Considerations}, {New} {Collision} {Laws}, and {Some} {Experimental} {Data}},
	school = {Cornell University},
	author = {Chatterjee, Anindya},
	year = {1997},
}

@book{brogliato_nonsmooth_2016,
	address = {Cham, The Swiss Confederation},
	edition = {3},
	series = {Communications and {Control} {Engineering}},
	title = {Nonsmooth {Mechanics}: {Models}, {Dynamics} and {Control}},
	isbn = {978-3-319-28664-8},
	publisher = {Springer},
	author = {Brogliato, Bernard},
	year = {2016},
}

@book{chicone_ordinary_1999,
	address = {New York, NY},
	series = {Texts in {Applied} {Mathematics}},
	title = {Ordinary {Differential} {Equations} {With} {Applications}},
	volume = {34},
	isbn = {978-0-387-98535-2},
	publisher = {Springer},
	author = {Chicone, Carmen},
	year = {1999},
}

@phdthesis{movahedi-lankarani_canonical_1988,
	address = {Tucson, AZ},
	type = {Ph. {D}. thesis},
	title = {Canonical {Equations} of {Motion} and {Estimation} of {Parameters} in the {Analysis} of {Impact} {Problems}},
	school = {The University of Arizona},
	author = {Movahedi-Lankarani, Hamid},
	year = {1988},
}

@article{poursina_new_2025,
	title = {A {New} {Model} {With} {Uniform} {Damping} {Force} for {Frictionless} {Impacts} {With} {Non}-{Permanent} {Deformation} at the {Time} of {Separation}},
	volume = {63},
	doi = {10.1007/s11044-024-10003-7},
	number = {1-2},
	journal = {Multibody System Dynamics},
	author = {Poursina, Mohammad and Nikravesh, Parviz E.},
	year = {2025},
	pages = {235--253},
}

@article{ouyang_characterization_2013,
	title = {On the {Characterization} of the {Duhem} {Hysteresis} {Operator} {With} {Clockwise} {Input}{\textendash}{Output} {Dynamics}},
	volume = {62},
	doi = {10.1016/j.sysconle.2012.11.022},
	number = {3},
	journal = {Systems \& Control Letters},
	author = {Ouyang, Ruiyue and Andrieu, Vincent and Jayawardhana, Bayu},
	year = {2013},
	pages = {286--293},
}

@article{jayawardhana_stability_2012,
	title = {Stability of {Systems} {With} the {Duhem} {Hysteresis} {Operator}: {The} {Dissipativity} {Approach}},
	volume = {48},
	doi = {10.1016/j.automatica.2012.06.069},
	number = {10},
	journal = {Automatica},
	author = {Jayawardhana, Bayu and Ouyang, Ruiyue and Andrieu, Vincent},
	year = {2012},
	pages = {2657--2662},
}

@inproceedings{jayawardhana_sufficient_2009,
	address = {New York, NY},
	title = {Sufficient {Conditions} for {Dissipativity} on {Duhem} {Hysteresis} {Model}},
	doi = {10.1109/CDC.2009.5400790},
	booktitle = {Proceedings of the 48h {IEEE} {Conference} on {Decision} and {Control} ({CDC}) held jointly with 2009 28th {Chinese} {Control} {Conference}},
	publisher = {IEEE},
	author = {Jayawardhana, Bayu and Andrieu, Vincent},
	year = {2009},
	pages = {4378--4383},
}

@inproceedings{jayawardhana_dissipativity_2011,
	address = {New York, NY},
	title = {Dissipativity of {General} {Duhem} {Hysteresis} {Models}},
	doi = {10.1109/CDC.2011.6160799},
	booktitle = {Proceedings of the 50th {IEEE} {Conference} on {Decision} and {Control} and {European} {Control} {Conference}},
	publisher = {IEEE},
	author = {Jayawardhana, Bayu and Ouyang, Ruiyue and Andrieu, Vincent},
	year = {2011},
	pages = {3234--3239},
}

@article{bhattacharjee_interplay_2017,
	title = {Interplay {Between} {Dissipation} and {Modal} {Truncation} in {Ball}-{Beam} {Impact}},
	volume = {12},
	doi = {10.1115/1.4036830},
	number = {6},
	journal = {ASME Journal of Computational and Nonlinear Dynamics},
	author = {Bhattacharjee, Arindam and Chatterjee, Anindya},
	year = {2017},
	pages = {061018},
}

@phdthesis{cundall_measurement_1971,
	address = {London, The United Kingdom of Great Britain and Northern Ireland},
	type = {Ph. {D}. thesis},
	title = {The {Measurement} and {Analysis} of {Accelerations} in {Rock} {Slopes}},
	school = {Imperial College of Science and Technology, University of London},
	author = {Cundall, Peter Alan},
	year = {1971},
}

@article{moore_collision_1988,
	title = {Collision {Detection} and {Response} for {Computer} {Animation}},
	volume = {22},
	doi = {10.1145/378456.378528},
	number = {4},
	journal = {ACM SIGGRAPH Computer Graphics},
	author = {Moore, Matthew and Wilhelms, Jane},
	year = {1988},
	pages = {289--298},
}

@article{platt_constraint_1988,
	title = {Constraint {Methods} for {Flexible} {Models}},
	volume = {22},
	doi = {10.1145/378456.378524},
	number = {4},
	journal = {ACM SIGGRAPH Computer Graphics},
	author = {Platt, John C. and Barr, Alan H.},
	year = {1988},
	pages = {279--288},
}

@article{terzopoulos_elastically_1987,
	title = {Elastically {Deformable} {Models}},
	volume = {21},
	doi = {10.1145/37402.37427},
	number = {4},
	journal = {ACM SIGGRAPH Computer Graphics},
	author = {Terzopoulos, Demetri and Platt, John and Barr, Alan and Fleischer, Kurt},
	year = {1987},
	pages = {205--214},
}

@article{milehins_incremental_2026,
	title = {Incremental {Collision} {Laws} {Based} on the {Bouc}{\textendash}{Wen} {Model}: {Improved} {Collision} {Models} and {Further} {Results}},
	volume = {21},
	doi = {10.1115/1.4071112},
	number = {6},
	journal = {ASME Journal of Computational and Nonlinear Dynamics},
	author = {Milehins, Mihails and Marghitu, Dan B.},
	year = {2026},
	pages = {061003},
}

@book{euler_theoria_1765,
	address = {Rostock and Greifswald},
	title = {Theoria {Motus} {Corporum} {Solidorum} {Seu} {Rigidorum}},
	url = {https://www.17centurymaths.com/contents/mechanica3.html},
	publisher = {A. F. R{\"o}se},
	author = {Euler, Leonhard},
	year = {1765},
	note = {[The original was downloaded from the Euler Archive - All Works, 289; the citation is based on the translation by Ian Bruce produced in 2017]},
}

@book{flores_contact_2016,
	address = {Cham, The Swiss Confederation},
	series = {Solid {Mechanics} and {Its} {Applications}},
	title = {Contact {Force} {Models} for {Multibody} {Dynamics}},
	volume = {226},
	isbn = {978-3-319-30897-5},
	doi = {10.1007/978-3-319-30897-5},
	publisher = {Springer International Publishing AG Switzerland},
	author = {Flores, Paulo and Lankarani, Hamid M.},
	year = {2016},
}

@article{walton_stress_1986,
	title = {Stress {Calculations} for {Assemblies} of {Inelastic} {Spheres} in {Uniform} {Shear}},
	volume = {63},
	doi = {10.1007/BF01182541},
	number = {1-4},
	journal = {Acta Mechanica},
	author = {Walton, O. R. and Braun, R. L.},
	year = {1986},
	pages = {73--86},
}

@article{witsenhausen_class_1966,
	title = {A {Class} of {Hybrid}-{State} {Continuous}-{Time} {Dynamic} {Systems}},
	volume = {11},
	doi = {10.1109/TAC.1966.1098336},
	number = {2},
	journal = {IEEE Transactions on Automatic Control},
	author = {Witsenhausen, H.},
	year = {1966},
	pages = {161--167},
}

@article{tavernini_differential_1987,
	title = {Differential {Automata} and {Their} {Discrete} {Simulators}},
	volume = {11},
	doi = {10.1016/0362-546X(87)90034-4},
	number = {6},
	journal = {Nonlinear Analysis: Theory, Methods \& Applications},
	author = {Tavernini, Lucio},
	year = {1987},
	pages = {665--683},
}

@article{goebel_hybrid_2004,
	title = {Hybrid {Systems}: {Generalized} {Solutions} and {Robust} {Stability}},
	volume = {37},
	doi = {10.1016/S1474-6670(17)31194-1},
	number = {13},
	journal = {IFAC Proceedings Volumes},
	author = {Goebel, Rafal and Hespanha, Joao and Teel, Andrew R. and Cai, Chaohong and Sanfelice, Ricardo},
	year = {2004},
	pages = {1--12},
}

@inproceedings{cai_results_2005,
	address = {New York, NY},
	title = {Results on {Input}-{To}-{State} {Stability} for {Hybrid} {Systems}},
	doi = {10.1109/CDC.2005.1583021},
	booktitle = {Proceedings of the 44th {IEEE} {Conference} on {Decision} and {Control}, and the {European} {Control} {Conference} 2005},
	publisher = {IEEE},
	author = {Cai, Chaohong and Teel, A.R.},
	year = {2005},
	pages = {5403--5408},
}

@inproceedings{sanfelice_results_2010,
	address = {New York, NY},
	title = {Results on {Input}-{To}-{Output} and {Input}-{Output}-{To}-{State} {Stability} for {Hybrid} {Systems} and {Their} {Interconnections}},
	doi = {10.1109/CDC.2010.5718164},
	booktitle = {Proceedings of the 49th {IEEE} {Conference} on {Decision} and {Control}},
	publisher = {IEEE},
	author = {Sanfelice, Ricardo G.},
	year = {2010},
	pages = {2396--2401},
}

@article{hunter_matplotlib_2007,
	title = {Matplotlib: {A} {2D} {Graphics} {Environment}},
	volume = {9},
	doi = {10.1109/MCSE.2007.55},
	number = {3},
	journal = {Computing in Science \& Engineering},
	author = {Hunter, John D.},
	year = {2007},
	pages = {90--95},
}

@article{tomas_fundamentals_2004,
	title = {Fundamentals of {Cohesive} {Powder} {Consolidation} and {Flow}},
	volume = {6},
	doi = {10.1007/s10035-004-0167-9},
	number = {2},
	journal = {Granular Matter},
	author = {Tomas, J{\"u}rgen},
	year = {2004},
	pages = {75--86},
}

@article{smith_pile-driving_1960,
	title = {Pile-{Driving} {Analysis} by the {Wave} {Equation}},
	volume = {86},
	doi = {10.1061/JSFEAQ.0000281},
	number = {4},
	journal = {Journal of the Soil Mechanics and Foundations Division},
	publisher = {American Society of Civil Engineers},
	author = {Smith, E. A. L.},
	year = {1960},
	pages = {35--61},
}

@article{thakur_micromechanical_2014,
	title = {Micromechanical {Analysis} of {Cohesive} {Granular} {Materials} {Using} the {Discrete} {Element} {Method} {With} an {Adhesive} {Elasto}-{Plastic} {Contact} {Model}},
	volume = {16},
	issn = {1434-7636},
	doi = {10.1007/s10035-014-0506-4},
	number = {3},
	journal = {Granular Matter},
	author = {Thakur, Subhash C. and Morrissey, John P. and Sun, Jin and Chen, J. F. and Ooi, Jin Y.},
	year = {2014},
	pages = {383--400},
}

@article{rathbone_accurate_2015,
	title = {An {Accurate} {Force}{\textendash}{Displacement} {Law} for the {Modelling} of {Elastic}{\textendash}{Plastic} {Contacts} in {Discrete} {Element} {Simulations}},
	volume = {282},
	doi = {10.1016/j.powtec.2014.12.055},
	journal = {Powder Technology},
	author = {Rathbone, Daniel and Marigo, Michele and Dini, Daniele and van Wachem, Berend},
	year = {2015},
	pages = {2--9},
}

@article{tomas_particle_2000,
	title = {Particle {Adhesion} {Fundamentals} and {Bulk} {Powder} {Consolidation}},
	volume = {18},
	doi = {10.14356/kona.2000022},
	journal = {KONA Powder and Particle Journal},
	author = {Tomas, J{\"u}rgen},
	year = {2000},
	pages = {157--169},
}

@phdthesis{morrissey_discrete_2013,
	address = {Edinburgh, The United Kingdom of Great Britain and Northern Ireland},
	type = {Ph. {D}. thesis},
	title = {Discrete {Element} {Modelling} of {Iron} {Ore} {Pellets} to {Include} the {Effects} of {Moisture} and {Fines}},
	school = {University of Edinburgh},
	author = {Morrissey, John Paul},
	year = {2013},
}

@article{naldi_passivity-based_2013,
	title = {Passivity-{Based} {Control} for {Hybrid} {Systems} {With} {Applications} to {Mechanical} {Systems} {Exhibiting} {Impacts}},
	volume = {49},
	doi = {10.1016/j.automatica.2013.01.018},
	number = {5},
	journal = {Automatica},
	author = {Naldi, Roberto and Sanfelice, Ricardo G.},
	year = {2013},
	pages = {1104--1116},
}

@inproceedings{holloway_effects_1978,
	address = {Richardson, TX},
	title = {The {Effects} of {Residual} {Driving} {Stresses} on {Pile} {Performance} {Under} {Axial} {Loads}},
	isbn = {978-1-55563-555-8},
	doi = {10.4043/3306-MS},
	booktitle = {Proceedings of the 10th {Annual} {Offshore} {Technology} {Conference}},
	publisher = {OnePetro},
	author = {Holloway, D. M. and Clough, G. W. and Vesic, A. S.},
	year = {1978},
	pages = {2225--2236},
}

@article{smith_impact_1955,
	title = {Impact and {Longitudinal} {Wave} {Transmission}},
	volume = {77},
	doi = {10.1115/1.4014566},
	number = {6},
	journal = {Transactions of the American Society of Mechanical Engineers},
	author = {Smith, E. A. L.},
	year = {1955},
	pages = {963--971},
}

@incollection{smith_pile-driving_1950,
	address = {New York, NY},
	title = {Pile-{Driving} {Impact}},
	booktitle = {Proceedings: {Industrial} {Computation} {Seminar}},
	publisher = {International Business Machines Corporation},
	author = {Smith, E. A.},
	editor = {Hurd, Cuthbert C.},
	year = {1950},
	pages = {44--51},
}

@phdthesis{ning_elasto-plastic_1995,
	address = {Birmingham, The United Kingdom of Great Britain and Northern Ireland},
	type = {Ph. {D}. thesis},
	title = {Elasto-{Plastic} {Impact} of {Fine} {Particles} and {Fragmentation} of {Small} {Agglomerates}},
	school = {Aston University},
	author = {Ning, Zemin},
	year = {1995},
}

@incollection{walton_numerical_1993,
	address = {Oxford, The United Kingdom of Great Britain and Northern Ireland},
	series = {Butterworth-{Heinemann} {Series} in {Chemical} {Engineering}},
	title = {Numerical {Simulation} of {Inelastic} {Frictional} {Particle}-{Particle}-{Interactions}},
	isbn = {978-0-7506-9275-5},
	booktitle = {Particulate {Two}-{Phase} {Flow}},
	publisher = {Butterworth-Heinemann},
	author = {Walton, Otis R.},
	editor = {Roco, M. C.},
	year = {1993},
}

@inproceedings{heemels_hybrid_2021,
	address = {New York, NY},
	title = {Hybrid {Systems} {With} {Continuous}-{Time} {Inputs}: {Subtleties} in {Solution} {Concepts} and {Existence} {Results}},
	doi = {10.1109/CDC45484.2021.9683389},
	booktitle = {Proceedings of the 60th {IEEE} {Conference} on {Decision} and {Control}},
	publisher = {IEEE},
	author = {Heemels, W. P. M. H. and Bernard, P. and Scheres, K. J. A. and Postoyan, R. and Sanfelice, R. G.},
	year = {2021},
	pages = {5368--5373},
}

@article{grzesikiewicz_mathematical_2018,
	title = {Mathematical {Formulation} of {Soft}-{Contact} {Problems} for {Various} {Rheological} {Models} of {Damper}},
	volume = {2018},
	doi = {10.1155/2018/8675016},
	number = {1},
	journal = {Shock and Vibration},
	author = {Grzesikiewicz, Wies{\l }aw and Zbiciak, Artur},
	year = {2018},
	pages = {8675016},
}

@article{depouhon_drifting_2013,
	title = {A {Drifting} {Impact} {Oscillator} {With} {Periodic} {Impulsive} {Loading}: {Application} to {Percussive} {Drilling}},
	volume = {258},
	doi = {10.1016/j.physd.2013.04.012},
	journal = {Physica D: Nonlinear Phenomena},
	author = {Depouhon, A. and Deno{\"e}l, V. and Detournay, E.},
	year = {2013},
	pages = {1--10},
}

@article{krivtsov_dry_1999,
	title = {Dry {Friction} {Model} of {Percussive} {Drilling}},
	volume = {34},
	doi = {10.1023/A:1004703819275},
	number = {6},
	journal = {Meccanica},
	author = {Krivtsov, Anton M. and Wiercigroch, Marian},
	year = {1999},
	pages = {425--434},
}

@article{pavlovskaia_modeling_2001,
	title = {Modeling of an {Impact} {System} {With} a {Drift}},
	volume = {64},
	doi = {10.1103/PhysRevE.64.056224},
	number = {5},
	journal = {Physical Review E},
	author = {Pavlovskaia, Ekaterina and Wiercigroch, Marian and Grebogi, Celso},
	year = {2001},
	pages = {056224},
}

@article{van_der_schaft_complementarity_1998,
	title = {Complementarity {Modeling} of {Hybrid} {Systems}},
	volume = {43},
	doi = {10.1109/9.664151},
	number = {4},
	journal = {IEEE Transactions on Automatic Control},
	author = {van der Schaft, A.J. and Schumacher, J.M.},
	year = {1998},
	pages = {483--490},
}

@article{tornambe_modelling_1996,
	title = {Modelling and {Controlling} {One}-{Degree}-{Of}-{Freedom} {Impacts} {Under} {Elastic}/{Plastic} {Deformations}},
	volume = {143},
	doi = {10.1049/ip-cta:19960536},
	number = {5},
	journal = {IEE Proceedings - Control Theory and Applications},
	author = {Tornamb{\`e}, A.},
	year = {1996},
	pages = {470--476},
}

@inproceedings{indri_impact_1999,
	address = {New York, NY},
	title = {Impact {Control} {Under} {Elastic}/{Plastic} {Deformations}},
	doi = {10.23919/ECC.1999.7099605},
	booktitle = {Proceedings of the 1999 {European} {Control} {Conference}},
	publisher = {IEEE},
	author = {Indri, M. and Tornamb{\`e}, A.},
	year = {1999},
	pages = {1966--1971},
}

@inproceedings{sanfelice_toolbox_2013,
	address = {New York, NY},
	title = {A {Toolbox} for {Simulation} of {Hybrid} {Systems} in {Matlab}/{Simulink}: {Hybrid} {Equations} ({HyEQ}) {Toolbox}},
	isbn = {978-1-4503-1567-8},
	doi = {10.1145/2461328.2461346},
	booktitle = {Proceedings of the 16th {International} {Conference} on {Hybrid} {Systems}: {Computation} and {Control}},
	publisher = {Association for Computing Machinery},
	author = {Sanfelice, Ricardo and Copp, David and {\~N}a{\~n}ez, Pablo},
	editor = {Belta, Calin and Ivan{\v c}i{\'c}, Franjo},
	year = {2013},
	pages = {101--106},
}

@inproceedings{sanfelice_nested_2008,
	address = {New York, NY},
	title = {A {Nested} {Matrosov} {Theorem} for {Hybrid} {Systems}},
	doi = {10.1109/ACC.2008.4586938},
	booktitle = {Proceedings of the 2008 {American} {Control} {Conference}},
	publisher = {IEEE},
	author = {Sanfelice, Ricardo G. and Teel, Andrew R.},
	year = {2008},
	pages = {2915--2920},
}

@article{copp_zero-crossing_2016,
	title = {A {Zero}-{Crossing} {Detection} {Algorithm} for {Robust} {Simulation} of {Hybrid} {Systems} {Jumping} on {Surfaces}},
	volume = {68},
	doi = {10.1016/j.simpat.2016.07.005},
	journal = {Simulation Modelling Practice and Theory},
	author = {Copp, David A. and Sanfelice, Ricardo G.},
	year = {2016},
	pages = {1--17},
}

@inproceedings{de_moura_lean_2021,
	address = {Cham, The Swiss Confederation},
	series = {Lecture {Notes} in {Artificial} {Intelligence}},
	title = {The {Lean} 4 {Theorem} {Prover} and {Programming} {Language}},
	volume = {12699},
	doi = {10.1007/978-3-030-79876-5_37},
	booktitle = {Automated {Deduction} {\textendash} {CADE} 28},
	publisher = {Springer},
	author = {de Moura, Leonardo and Ullrich, Sebastian},
	editor = {Platzer, Andr{\'e} and Sutcliffe, Geoff},
	year = {2021},
	pages = {625--635},
}

@inproceedings{the_mathlib_community_lean_2020,
	address = {New York, NY},
	title = {The {Lean} {Mathematical} {Library}},
	doi = {10.1145/3372885.3373824},
	booktitle = {Proceedings of the 9th {ACM} {SIGPLAN} {International} {Conference} on {Certified} {Programs} and {Proofs}},
	publisher = {Association for Computing Machinery},
	author = {{The mathlib Community}},
	editor = {Blanchette, Jasmin and Hri{\c t}cu, C{\u a}t{\u a}lin},
	year = {2020},
	pages = {367--381},
}

@misc{doll_foundations_2026,
	title = {Foundations of {Machine}-{Checked} {Control} {Theory} in {Lean}},
	url = {https://arxiv.org/abs/2607.19727},
	doi = {10.48550/arXiv.2607.19727},
	author = {Doll, Moritz and Shames, Iman},
	year = {2026},
	note = {Published: arXiv preprint arXiv:2607.19727},
}

@misc{ladner_coralean_2026,
	title = {{CORALean}: {A} {Lean} 4 {Formalization} of {CORA}},
	url = {https://github.com/TUMcps/CORALean},
	urldate = {2026-11-15},
	author = {Ladner, Tobias and Althoff, Matthias},
	year = {2026},
}

@misc{enes_uniform_2026,
	title = {Uniform {Substitution} for {Differential} {Dynamic} {Logic} in {Lean}},
	url = {https://github.com/eneoli/dLean},
	urldate = {2026-11-15},
	author = {Enes, Oliver and Prebet, Enguerrand},
	year = {2026},
}

@misc{ogmutlu_dlean_2026,
	title = {{dLean}},
	url = {https://github.com/ogmutlu/dLean},
	urldate = {2026-11-15},
	author = {{ogmutlu}},
	year = {2026},
}

@book{platzer_logical_2018,
	address = {Cham, The Swiss Confederation},
	title = {Logical {Foundations} of {Cyber}-{Physical} {Systems}},
	isbn = {978-3-319-63588-0},
	publisher = {Springer},
	author = {Platzer, Andr{\'e}},
	year = {2018},
}

@inproceedings{fulton_keymaera_2015,
	address = {Cham, The Swiss Confederation},
	series = {Lecture {Notes} in {Artificial} {Intelligence}},
	title = {{KeYmaera} {X}: {An} {Axiomatic} {Tactical} {Theorem} {Prover} for {Hybrid} {Systems}},
	volume = {9195},
	doi = {10.1007/978-3-319-21401-6_36},
	booktitle = {Automated {Deduction} - {CADE}-25},
	publisher = {Springer},
	author = {Fulton, Nathan and Mitsch, Stefan and Quesel, Jan-David and V{\"o}lp, Marcus and Platzer, Andr{\'e}},
	editor = {Felty, Amy P. and Middeldorp, Aart},
	year = {2015},
	pages = {527--538},
}

@inproceedings{bohrer_formally_2017,
	address = {New York, NY},
	title = {Formally {Verified} {Differential} {Dynamic} {Logic}},
	doi = {10.1145/3018610.3018616},
	booktitle = {Proceedings of the 6th {ACM} {SIGPLAN} {Conference} on {Certified} {Programs} and {Proofs}},
	publisher = {Association for Computing Machinery},
	author = {Bohrer, Rose and Rahli, Vincent and Vukotic, Ivana and V{\"o}lp, Marcus and Platzer, Andr{\'e}},
	editor = {Bertot, Yves and Vafeiadis, Viktor},
	year = {2017},
	pages = {208--221},
}

@article{huerta_y_munive_predicate_2022,
	title = {Predicate {Transformer} {Semantics} for {Hybrid} {Systems}: {Verification} {Components} for {Isabelle}/{HOL}},
	volume = {66},
	doi = {10.1007/s10817-021-09607-x},
	number = {1},
	journal = {Journal of Automated Reasoning},
	author = {Huerta y Munive, Jonathan Juli{\'a}n and Struth, Georg},
	year = {2022},
	pages = {93--139},
}

@inproceedings{foster_hybrid_2021,
	address = {Cham, The Swiss Confederation},
	series = {Lecture {Notes} in {Computer} {Science}},
	title = {Hybrid {Systems} {Verification} with {Isabelle}/{HOL}: {Simpler} {Syntax}, {Better} {Models}, {Faster} {Proofs}},
	volume = {13047},
	doi = {10.1007/978-3-030-90870-6_20},
	booktitle = {Formal {Methods} ({FM} 2021)},
	publisher = {Springer},
	author = {Foster, Simon and Huerta y Munive, Jonathan Juli{\'a}n and Gleirscher, Mario and Struth, Georg},
	editor = {Huisman, Marieke and P{\u a}s{\u a}reanu, Corina and Zhan, Naijun},
	year = {2021},
	pages = {367--386},
}

@article{huerta_y_munive_isavodes_2024,
	title = {{IsaVODEs}: {Interactive} {Verification} of {Cyber}-{Physical} {Systems} at {Scale}},
	volume = {68},
	doi = {10.1007/s10817-024-09709-2},
	number = {4},
	journal = {Journal of Automated Reasoning},
	author = {Huerta y Munive, Jonathan Juli{\'a}n and Foster, Simon and Gleirscher, Mario and Struth, Georg and Pardillo Laursen, Christian and Hickman, Thomas},
	year = {2024},
	pages = {21},
}

@inproceedings{immler_numerical_2012,
	address = {Berlin, The Federal Republic of Germany},
	series = {Lecture {Notes} in {Computer} {Science}},
	title = {Numerical {Analysis} of {Ordinary} {Differential} {Equations} in {Isabelle}/{HOL}},
	volume = {7406},
	doi = {10.1007/978-3-642-32347-8_26},
	booktitle = {Interactive {Theorem} {Proving} ({ITP} 2012)},
	publisher = {Springer},
	author = {Immler, Fabian and H{\"o}lzl, Johannes},
	editor = {Beringer, Lennart and Felty, Amy},
	year = {2012},
	pages = {377--392},
}

@article{immler_verified_2018,
	title = {A {Verified} {ODE} {Solver} and the {Lorenz} {Attractor}},
	volume = {61},
	doi = {10.1007/s10817-017-9448-y},
	journal = {Journal of Automated Reasoning},
	author = {Immler, Fabian},
	year = {2018},
	pages = {73--111},
}

@inproceedings{abraham-mumm_verification_2001,
	address = {Los Alamitos, CA},
	title = {Verification of {Hybrid} {Systems}: {Formalization} and {Proof} {Rules} in {PVS}},
	doi = {10.1109/ICECCS.2001.930163},
	booktitle = {Proceedings of the 7th {IEEE} {International} {Conference} on {Engineering} of {Complex} {Computer} {Systems}},
	publisher = {IEEE Computer Society},
	author = {{\'A}brah{\'a}m-Mumm, Erika and Hannemann, Ulrich and Steffen, Martin},
	editor = {Andler, Sten F. and Hinchey, Michael G. and Offutt, Jeff},
	year = {2001},
	pages = {48--57},
}

@inproceedings{wu_autoformalization_2022,
	address = {Red Hook, NY},
	title = {Autoformalization with {Large} {Language} {Models}},
	volume = {2344},
	doi = {10.52202/068431-2344},
	booktitle = {Proceedings of the 36th {International} {Conference} on {Neural} {Information} {Processing} {Systems}},
	publisher = {Curran Associates Inc.},
	author = {Wu, Yuhuai and Jiang, Albert Qiaochu and Li, Wenda and Rabe, Markus N. and Staats, Charles and Jamnik, Mateja and Szegedy, Christian},
	editor = {Koyejo, S. and Mohamed, S. and Agarwal, A. and Belgrave, D. and Cho, K. and Oh, A.},
	year = {2022},
	pages = {32353--32368},
}

@misc{math_inc_strong_2025,
	title = {The strong {Prime} {Number} {Theorem}},
	url = {https://github.com/math-inc/strongpnt},
	urldate = {2026-11-17},
	author = {{Math Inc.}},
	year = {2025},
}

@misc{xu_agentic_2026,
	title = {Agentic {Proof} {Automation}: {A} {Case} {Study}},
	url = {https://arxiv.org/abs/2601.03768},
	doi = {10.48550/arXiv.2601.03768},
	author = {Xu, Yichen and Odersky, Martin},
	year = {2026},
}

@misc{liu_numina-lean-agent_2026,
	title = {Numina-{Lean}-{Agent}: {An} {Open} and {General} {Agentic} {Reasoning} {System} for {Formal} {Mathematics}},
	url = {https://arxiv.org/abs/2601.14027},
	doi = {10.48550/arXiv.2601.14027},
	author = {Liu, Junqi and Zhou, Zihao and Zhu, Zekai and Dos Santos, Marco and He, Weikun and Liu, Jiawei and Wang, Ran and Xie, Yunzhou and Zhao, Junqiao and Wang, Qiufeng and Zhi, Lihong and Li, Jia and Li, Wenda},
	year = {2026},
}

@misc{gloeckle_automatic_2026,
	title = {Automatic {Textbook} {Formalization}},
	url = {https://arxiv.org/abs/2604.03071},
	doi = {10.48550/arXiv.2604.03071},
	author = {Gloeckle, Fabian and Rammal, Ahmad and Arnal, Charles and Munos, R{\'e}mi and Cabannes, Vivien and Synnaeve, Gabriel and Hayat, Amaury},
	year = {2026},
}

@misc{ilin_sorries_2026,
	title = {Sorries {Are} {Not} the {Hard} {Part}: {An} {Expert}-{Review} {Case} {Study} of a {Semi}-{Autonomous} {Formalization}},
	url = {https://arxiv.org/abs/2606.13925},
	doi = {10.48550/arXiv.2606.13925},
	author = {Ilin, Vasily and Nugent, Brian},
	year = {2026},
}

@misc{soltani_moakhar_beyond_2026,
	title = {Beyond the {Library}: {An} {Agentic} {Framework} for {Autoformalizing} {Research} {Mathematics}},
	url = {https://arxiv.org/abs/2606.31134},
	doi = {10.48550/arXiv.2606.31134},
	author = {Soltani Moakhar, Arshia and Gholami, Iman and Springer, Max and JafariRaviz, Mahdi and Hajiaghayi, MohammadTaghi},
	year = {2026},
}

@misc{anthropic_formalizing_2026,
	title = {Formalizing {Fermat}'s {Last} {Theorem}},
	url = {https://www.anthropic.com/research/formalizing-fermats-last-theorem},
	urldate = {2026-11-17},
	author = {{Anthropic}},
	year = {2026},
}

@inproceedings{de_moura_lean_2015,
	address = {Cham, The Swiss Confederation},
	series = {Lecture {Notes} in {Computer} {Science}},
	title = {The {Lean} {Theorem} {Prover} ({System} {Description})},
	volume = {9195},
	doi = {10.1007/978-3-319-21401-6_26},
	booktitle = {Automated {Deduction} - {CADE}-25},
	publisher = {Springer},
	author = {de Moura, Leonardo and Kong, Soonho and Avigad, Jeremy and van Doorn, Floris and von Raumer, Jakob},
	editor = {Felty, Amy P. and Middeldorp, Aart},
	year = {2015},
	pages = {378--388},
}

@inproceedings{archer_verifying_1997,
	address = {Berlin, The Federal Republic of Germany},
	series = {Lecture {Notes} in {Computer} {Science}},
	title = {Verifying {Hybrid} {Systems} {Modeled} as {Timed} {Automata}: {A} {Case} {Study}},
	volume = {1201},
	doi = {10.1007/BFb0014724},
	booktitle = {Hybrid and {Real}-{Time} {Systems} ({HART} 1997)},
	publisher = {Springer},
	author = {Archer, Myla and Heitmeyer, Constance L.},
	editor = {Maler, Oded},
	year = {1997},
	pages = {171--185},
}

@incollection{vitt_assertional_1996,
	series = {Lecture {Notes} in {Computer} {Science}},
	title = {Assertional {Specification} and {Verification} {Using} {PVS} of the {Steam} {Boiler} {Control} {System}},
	volume = {1165},
	doi = {10.1007/BFb0027249},
	booktitle = {Formal {Methods} for {Industrial} {Applications}: {Specifying} and {Programming} the {Steam} {Boiler} {Control}},
	publisher = {Springer},
	author = {Vitt, Jan and Hooman, Jozef},
	editor = {Abrial, Jean-Raymond and B{\"o}rger, Egon and Langmaack, Hans},
	year = {1996},
	pages = {453--472},
}

@inproceedings{henzinger_reachability_1998,
	address = {Berlin, The Federal Republic of Germany},
	series = {Lecture {Notes} in {Computer} {Science}},
	title = {Reachability {Verification} for {Hybrid} {Automata}},
	volume = {1386},
	doi = {10.1007/3-540-64358-3_40},
	booktitle = {Hybrid {Systems}: {Computation} and {Control} ({HSCC} 1998)},
	publisher = {Springer},
	author = {Henzinger, Thomas A. and Rusu, Vlad},
	editor = {Henzinger, Thomas A. and Sastry, Shankar},
	year = {1998},
	pages = {190--204},
}

@inproceedings{manna_deductive_1998,
	address = {Berlin, The Federal Republic of Germany},
	series = {Lecture {Notes} in {Computer} {Science}},
	title = {Deductive {Verification} of {Hybrid} {Systems} {Using} {STeP}},
	volume = {1386},
	doi = {10.1007/3-540-64358-3_47},
	booktitle = {Hybrid {Systems}: {Computation} and {Control} ({HSCC} 1998)},
	publisher = {Springer},
	author = {Manna, Zohar and Sipma, Henny B.},
	editor = {Henzinger, Thomas A. and Sastry, Shankar},
	year = {1998},
	pages = {305--318},
}

@inproceedings{urban_130k_2026,
	address = {Wadern, The Federal Republic of Germany},
	series = {Leibniz {International} {Proceedings} in {Informatics}},
	title = {130k {Lines} of {Formal} {Topology} in {Two} {Weeks}: {Simple} and {Cheap} {Autoformalization} for {Everyone}?},
	volume = {382},
	isbn = {978-3-95977-436-9},
	doi = {10.4230/LIPIcs.ITP.2026.31},
	booktitle = {Proceedings of the 17th {International} {Conference} on {Interactive} {Theorem} {Proving}},
	publisher = {Schloss Dagstuhl {\textendash} Leibniz-Zentrum f{\"u}r Informatik},
	author = {Urban, Josef},
	editor = {Komendantskaya, Ekaterina and Nipkow, Tobias},
	year = {2026},
	pages = {31:1--31:9},
}

@inproceedings{lamport_hybrid_1993,
	address = {Berlin, The Federal Republic of Germany},
	series = {Lecture {Notes} in {Computer} {Science}},
	title = {Hybrid systems in {TLA}+},
	volume = {736},
	isbn = {978-3-540-48060-0},
	doi = {10.1007/3-540-57318-6_25},
	booktitle = {Hybrid {Systems} ({HS} 1992, {HS} 1991)},
	publisher = {Springer},
	author = {Lamport, Leslie},
	editor = {Grossman, Robert L. and Nerode, Anil and Ravn, Anders P. and Rischel, Hans},
	year = {1993},
	pages = {77--102},
}

@misc{vargas-garcia_shslib_2026,
	title = {{SHSLib}: {A} {Lean} 4/{Mathlib} {Formalization} of {Stochastic} {Hybrid} {Systems}},
	url = {https://github.com/cavargar/SHSLib},
	author = {Vargas-Garc{\'i}a, Cesar Augusto},
	year = {2026},
}

@article{platzer_differential_2008,
	title = {Differential {Dynamic} {Logic} for {Hybrid} {Systems}},
	volume = {41},
	doi = {10.1007/s10817-008-9103-8},
	number = {2},
	journal = {Journal of Automated Reasoning},
	author = {Platzer, Andr{\'e}},
	year = {2008},
	pages = {143--189},
}

@misc{baudart_putnam_2026,
	title = {Putnam 2025 {Problems} in {Rocq} using {Opus} 4.6 and {Rocq}-{MCP}},
	url = {http://arxiv.org/abs/2603.20405},
	doi = {10.48550/arXiv.2603.20405},
	publisher = {arXiv},
	author = {Baudart, Guillaume and Lelarge, Marc and St{\'e}rin, Tristan and Viennot, Jules},
	year = {2026},
}

@misc{bryant_munkres_2026,
	title = {Munkres' {General} {Topology} {Autoformalized} in {Isabelle}/{HOL}},
	url = {http://arxiv.org/abs/2604.07455},
	doi = {10.48550/arXiv.2604.07455},
	publisher = {arXiv},
	author = {Bryant, Dustin and Huerta y Munive, Jonathan Juli{\'a}n and Kaliszyk, Cezary and Urban, Josef},
	year = {2026},
}

@misc{kappelmann_just_2026,
	title = {Just {Type} {It} in {Isabelle}! {AI} {Agents} {Drafting}, {Mechanizing}, and {Generalizing} from {Human} {Hints}},
	url = {http://arxiv.org/abs/2604.15713},
	doi = {10.48550/arXiv.2604.15713},
	publisher = {arXiv},
	author = {Kappelmann, Kevin and Sch{\"a}ffeler, Maximilian and Stevens, Lukas and Abdulaziz, Mohammad and Popescu, Andrei and Traytel, Dmitriy},
	year = {2026},
}

@misc{kappelmann_pide_2026,
	title = {{PIDE} {MCP}: {Connecting} {AI} {Agents} to {Isabelle}},
	url = {https://hal.science/hal-05748384},
	doi = {10.5281/zenodo.21519364},
	publisher = {HAL},
	author = {Kappelmann, Kevin},
	year = {2026},
}

@misc{jiang_draft_2023,
	title = {Draft, {Sketch}, and {Prove}: {Guiding} {Formal} {Theorem} {Provers} with {Informal} {Proofs}},
	url = {http://arxiv.org/abs/2210.12283},
	doi = {10.48550/arXiv.2210.12283},
	publisher = {arXiv},
	author = {Jiang, Albert Q. and Welleck, Sean and Zhou, Jin Peng and Li, Wenda and Liu, Jiacheng and Jamnik, Mateja and Lacroix, Timoth{\'e}e and Wu, Yuhuai and Lample, Guillaume},
	year = {2023},
}

@misc{ilin_semi-autonomous_2026,
	title = {Semi-{Autonomous} {Formalization} of the {Vlasov}-{Maxwell}-{Landau} {Equilibrium}},
	url = {http://arxiv.org/abs/2603.15929},
	doi = {10.48550/arXiv.2603.15929},
	publisher = {arXiv},
	author = {Ilin, Vasily},
	year = {2026},
}

@techreport{anthropic_model_2023,
	title = {Model {Card} and {Evaluations} for {Claude} {Models}},
	url = {https://www-cdn.anthropic.com/5c49cc247484cecf107c699baf29250302e5da70/claude-2-model-card.pdf},
	institution = {Anthropic},
	author = {{Anthropic}},
	year = {2023},
}

@misc{anthropic_claude_2025,
	title = {Claude {Code}},
	url = {https://github.com/anthropics/claude-code},
	author = {{Anthropic}},
	year = {2025},
}

@misc{openai_ten_2026,
	title = {Ten {Advances} in {Mathematics} and {Theoretical} {Computer} {Science}},
	url = {https://cdn.openai.com/pdf/ten-proofs-oai.pdf},
	urldate = {2026-11-23},
	author = {{OpenAI}},
	year = {2026},
}

@misc{openai_finite_2026,
	title = {Finite {Time} {Blowup} for {Navier}{\textendash}{Stokes}},
	url = {https://cdn.openai.com/pdf/32d9f210-8b73-45e0-91bc-82a30aef8a9a/navier-stokes.pdf},
	urldate = {2026-11-23},
	author = {{OpenAI}},
	year = {2026},
}

@misc{openai_openai_2026,
	title = {{OpenAI} {Research} {Catalog}: 373 {Result} {Families} in 721 {Manuscripts}},
	url = {https://github.com/openai/math},
	urldate = {2026-11-23},
	author = {{OpenAI}},
	year = {2026},
}

@article{hubert_olympiad-level_2026,
	title = {Olympiad-{Level} {Formal} {Mathematical} {Reasoning} {With} {Reinforcement} {Learning}},
	volume = {651},
	issn = {1476-4687},
	doi = {10.1038/s41586-025-09833-y},
	journal = {Nature},
	publisher = {Nature Publishing Group},
	author = {Hubert, Thomas and Mehta, Rishi and Sartran, Laurent and Horv{\'a}th, Mikl{\'o}s Z. and {\v Z}u{\v z}i{\'c}, Goran and Wieser, Eric and Huang, Aja and Schrittwieser, Julian and Schroecker, Yannick and Masoom, Hussain and Bertolli, Ottavia and Zahavy, Tom and Mandhane, Amol and Yung, Jessica and Beloshapka, Iuliya and Ibarz, Borja and Veeriah, Vivek and Yu, Lei and Nash, Oliver and Lezeau, Paul and Mercuri, Salvatore and S{\"o}nne, Calle and Mehta, Bhavik and Davies, Alex and Zheng, Daniel and Pedregosa, Fabian and Li, Yin and von Glehn, Ingrid and Rowland, Mark and Albanie, Samuel and Velingker, Ameya and Schmitt, Simon and Lockhart, Edward and Hughes, Edward and Michalewski, Henryk and Sonnerat, Nicolas and Hassabis, Demis and Kohli, Pushmeet and Silver, David},
	year = {2026},
	pages = {607--613},
}

@misc{kabra_llm-powered_2026,
	title = {{LLM}-{Powered} {Automatic} {Theorem} {Proving} and {Synthesis} for {Hybrid} {Systems} and {Games}},
	url = {http://arxiv.org/abs/2603.00737},
	doi = {10.34727/2026/isbn.978-3-85448-093-8_71},
	author = {Kabra, Aditi and Laurent, Jonathan and Martins, Ruben and Mitsch, Stefan and Platzer, Andr{\'e}},
	year = {2026},
}

@inproceedings{kaliszyk_automating_2017,
	address = {Cham, The Swiss Confederation},
	series = {Lecture {Notes} in {Computer} {Science}},
	title = {Automating {Formalization} by {Statistical} and {Semantic} {Parsing} of {Mathematics}},
	volume = {10499},
	isbn = {978-3-319-66107-0},
	doi = {10.1007/978-3-319-66107-0_2},
	booktitle = {Interactive {Theorem} {Proving} ({ITP} 2017)},
	publisher = {Springer},
	author = {Kaliszyk, Cezary and Urban, Josef and Vysko{\v c}il, Ji{\v r}{\'i}},
	editor = {Ayala-Rinc{\'o}n, Mauricio and Mu{\~n}oz, C{\'e}sar A.},
	year = {2017},
	pages = {12--27},
}

@phdthesis{bobrow_natural_1964,
	address = {Cambridge, MA},
	type = {Ph. {D}. {Thesis}},
	title = {Natural {Language} {Input} for a {Computer} {Problem} {Solving} {System}},
	url = {https://hdl.handle.net/1721.1/145210},
	school = {Massachusetts Institute of Technology},
	author = {Bobrow, Daniel G.},
	year = {1964},
}

@inproceedings{lyaletski_evidence_2010,
	address = {Berlin, The Federal Republic of Germany},
	series = {Lecture {Notes} in {Computer} {Science}},
	title = {Evidence {Algorithm} and {System} for {Automated} {Deduction}: {A} {Retrospective} {View}},
	volume = {6167},
	doi = {10.1007/978-3-642-14128-7_35},
	booktitle = {Intelligent {Computer} {Mathematics} ({CICM} 2010)},
	publisher = {Springer},
	author = {Lyaletski, Alexander and Verchinine, Konstantin},
	editor = {Autexier, Serge and Calmet, Jacques and Delahaye, David and Ion, Patrick D. F. and Rideau, Laurence and Rioboo, Renaud and Sexton, Alan P.},
	year = {2010},
	pages = {411--426},
}

@inproceedings{verchinine_system_2007,
	address = {Berlin, The Federal Republic of Germany},
	series = {Lecture {Notes} in {Artificial} {Intelligence}},
	title = {System for {Automated} {Deduction} ({SAD}): {A} {Tool} for {Proof} {Verification}},
	volume = {4603},
	doi = {10.1007/978-3-540-73595-3_29},
	booktitle = {Automated {Deduction} {\textendash} {CADE}-21},
	publisher = {Springer},
	author = {Verchinine, Konstantin and Lyaletski, Alexander and Paskevich, Andrey},
	editor = {Pfenning, Frank},
	year = {2007},
	pages = {398--403},
}

@inproceedings{cramer_naproche_2010,
	address = {Berlin, The Federal Republic of Germany},
	series = {Lecture {Notes} in {Artificial} {Intelligence}},
	title = {The {Naproche} {Project} {Controlled} {Natural} {Language} {Proof} {Checking} of {Mathematical} {Texts}},
	volume = {5972},
	doi = {10.1007/978-3-642-14418-9_11},
	booktitle = {Controlled {Natural} {Language} ({CNL} 2009)},
	publisher = {Springer},
	author = {Cramer, Marcos and Fisseni, Bernhard and Koepke, Peter and K{\"u}hlwein, Daniel and Schr{\"o}der, Bernhard and Veldman, Jip},
	editor = {Fuchs, Norbert E.},
	year = {2010},
	pages = {170--186},
}

@inproceedings{kaliszyk_developing_2014,
	address = {Cham, The Swiss Confederation},
	series = {Lecture {Notes} in {Computer} {Science}},
	title = {Developing {Corpus}-{Based} {Translation} {Methods} between {Informal} and {Formal} {Mathematics}: {Project} {Description}},
	volume = {8543},
	doi = {10.1007/978-3-319-08434-3_34},
	booktitle = {Intelligent {Computer} {Mathematics} ({CICM} 2014)},
	publisher = {Springer},
	author = {Kaliszyk, Cezary and Urban, Josef and Vysko{\v c}il, Ji{\v r}{\'i} and Geuvers, Herman},
	editor = {Watt, Stephen M. and Davenport, James H. and Sexton, Alan P. and Sojka, Petr and Urban, Josef},
	year = {2014},
	pages = {435--439},
}

@inproceedings{kaliszyk_learning_2015,
	address = {Cham, The Swiss Confederation},
	series = {Lecture {Notes} in {Computer} {Science}},
	title = {Learning to {Parse} on {Aligned} {Corpora} ({Rough} {Diamond})},
	volume = {9236},
	doi = {10.1007/978-3-319-22102-1_15},
	booktitle = {Interactive {Theorem} {Proving} ({ITP} 2015)},
	publisher = {Springer},
	author = {Kaliszyk, Cezary and Urban, Josef and Vysko{\v c}il, Ji{\v r}{\'i}},
	editor = {Urban, Christian and Zhang, Xingyuan},
	year = {2015},
	pages = {227--233},
}

@inproceedings{wang_first_2018,
	address = {Cham, The Swiss Confederation},
	series = {Lecture {Notes} in {Artificial} {Intelligence}},
	title = {First {Experiments} with {Neural} {Translation} of {Informal} to {Formal} {Mathematics}},
	volume = {11006},
	doi = {10.1007/978-3-319-96812-4_22},
	booktitle = {Intelligent {Computer} {Mathematics} ({CICM} 2018)},
	publisher = {Springer},
	author = {Wang, Qingxiang and Kaliszyk, Cezary and Urban, Josef},
	editor = {Rabe, Florian and Farmer, William M. and Passmore, Grant O. and Youssef, Abdou},
	year = {2018},
	pages = {255--270},
}

@inproceedings{wang_exploration_2020,
	address = {New York, NY},
	title = {Exploration of {Neural} {Machine} {Translation} in {Autoformalization} of {Mathematics} in {Mizar}},
	doi = {10.1145/3372885.3373827},
	booktitle = {Proceedings of the 9th {ACM} {SIGPLAN} {International} {Conference} on {Certified} {Programs} and {Proofs}},
	publisher = {Association for Computing Machinery},
	author = {Wang, Qingxiang and Brown, Chad E. and Kaliszyk, Cezary and Urban, Josef},
	editor = {Blanchette, Jasmin and Hri{\c t}cu, C{\u a}t{\u a}lin},
	year = {2020},
	pages = {85--98},
}

@book{lean_fro_lean_2026,
	title = {The {Lean} {Language} {Reference}, {Version} 4.35.0: {Validating} a {Lean} {Proof}},
	url = {https://lean-lang.org/doc/reference/4.35.0/ValidatingProofs/},
	publisher = {Lean FRO},
	author = {{Lean FRO}},
	year = {2026},
}

@misc{lean_fro_comparator_2026,
	title = {comparator: {A} trustworthy judge for {Lean} proofs},
	url = {https://github.com/leanprover/comparator},
	author = {{Lean FRO}},
	year = {2026},
}

@inproceedings{carneiro_lean4lean_2026,
	address = {Dagstuhl, Germany},
	series = {Leibniz {International} {Proceedings} in {Informatics} ({LIPIcs})},
	title = {{Lean4Lean}: {Verifying} a {Typechecker} for {Lean}, in {Lean}},
	volume = {384},
	isbn = {978-3-95977-441-3},
	url = {https://drops.dagstuhl.de/entities/document/10.4230/LIPIcs.TYPES.2025.2},
	doi = {10.4230/LIPIcs.TYPES.2025.2},
	booktitle = {31st {International} {Conference} on {Types} for {Proofs} and {Programs} ({TYPES} 2025)},
	publisher = {Schloss Dagstuhl {\textendash} Leibniz-Zentrum f{\"u}r Informatik},
	author = {Carneiro, Mario},
	editor = {Nordvall Forsberg, Fredrik and McKinna, James},
	year = {2026},
	pages = {2:1--2:23},
}

@misc{de_moura_lean_2026,
	title = {Lean {Now} {Ships} with {Verified} {Kernels}},
	url = {https://leodemoura.github.io/blog/2026-10-15-lean-now-ships-with-verified-kernels/},
	author = {de Moura, Leonardo},
	month = oct,
	year = {2026},
}

@mastersthesis{carneiro_type_2019,
	address = {Pittsburgh, PA},
	title = {The {Type} {Theory} of {Lean}},
	school = {Carnegie Mellon University},
	author = {Carneiro, Mario},
	year = {2019},
}

@article{pang_differential_2008,
	title = {Differential {Variational} {Inequalities}},
	volume = {113},
	doi = {10.1007/s10107-006-0052-x},
	number = {2},
	journal = {Mathematical Programming},
	author = {Pang, Jong-Shi and Stewart, David E.},
	year = {2008},
	pages = {345--424},
}

\end{document}